\documentclass[a4paper,11pt]{article}

\usepackage[UKenglish]{babel}
\usepackage{amssymb,amsmath,amsthm}
\usepackage{microtype}
\usepackage{authblk}
\usepackage{a4wide}
\usepackage{complexity}
\usepackage[dvipsnames]{xcolor}
\usepackage{hyperref} 
\hypersetup{colorlinks=true,citecolor=RoyalBlue,linkcolor=RoyalBlue,urlcolor=RoyalBlue}
\usepackage{bm}
\usepackage{thm-restate}

\usepackage[noabbrev,capitalize]{cleveref}
\usepackage{graphicx}
\usepackage[shortlabels]{enumitem}

\usepackage{todonotes}

\usepackage[a4paper,top=1in,]{geometry}

\theoremstyle{plain}
\newtheorem{theorem}{Theorem}[section]
\newtheorem{lemma}[theorem]{Lemma}
\newtheorem{corollary}[theorem]{Corollary}

\newtheorem{claim}[theorem]{Claim}
\theoremstyle{definition}
\newtheorem{remark}[theorem]{Remark}
\newtheorem{definition}[theorem]{Definition}
\newtheorem{fact}[theorem]{Fact}
\newtheorem{example}[theorem]{Example}

\usepackage[utf8]{inputenc}

\crefname{claim}{Claim}{Claims}
\Crefname{claim}{Claim}{Claims}

\newcommand{\RR}{\mathbb{R}} % R is taken

\newcommand{\scalprod}[2]{\left\langle#1,#2\right\rangle}

\DeclareMathOperator{\onethreeSAT}{1-in-3-SAT} %1-in-3-SAT
\DeclareMathOperator{\twofourSAT}{(\ge 2)-in-4-SAT} %2-in-4-SAT
\DeclareMathOperator{\threeXOR}{3-XOR} %3-XOR
\DeclareMathOperator{\threeSAT}{3-SAT} %3-SAT
\DeclareMathOperator{\fourSAT}{4-SAT} %4-SAT
\DeclareMathOperator{\NAE}{NAE} %NAE
\DeclareMathOperator{\twoSAT}{2-SAT} %2-SAT

\DeclareMathOperator{\PCSP}{PCSP}
\DeclareMathOperator{\fiPCSP}{fiPCSP}
\DeclareMathOperator{\fPCSP}{fPCSP}

\DeclareMathOperator{\CSP}{CSP}
\DeclareMathOperator{\Pol}{Pol} % Polymorphism

\DeclareMathOperator{\OR}{\mathsf{OR}}
\DeclareMathOperator{\AND}{\mathsf{AND}}
\DeclareMathOperator{\NAND}{\mathsf{NAND}}
\DeclareMathOperator{\NOR}{\mathsf{NOR}}
\DeclareMathOperator{\xNOR}{\mathsf{ANDNOR}}
\newcommand{\XOR}{\oplus}
\newcommand{\BigXOR}{\bigoplus}

\newcommand{\id}{\mathsf{id}}             % Idempotized function
\newcommand{\idneg}[1]{\id\overline{#1}}  % Idempotized and negated function
\newcommand{\negate}[1]{\overline{#1}}  % Negated function

\DeclareMathOperator{\Maj}{\mathsf{Maj}} % Majority
\DeclareMathOperator{\MajFam}{\bm{\mathsf{Maj}}}

\DeclareMathOperator{\Par}{\mathsf{Par}} % Odd parity
\DeclareMathOperator{\ParFam}{\bm{\mathsf{Par}}}

\DeclareMathOperator{\invMaj}{\overline{\Maj}} % Minority
\DeclareMathOperator{\invMajFam}{\overline{\MajFam}} % Minority
\DeclareMathOperator{\invPar}{\overline{\Par}} % Even parity
\DeclareMathOperator{\invParFam}{\overline{\ParFam}} % Even parity

\DeclareMathOperator{\AT}{\mathsf{AT}} % Alternating threshold
\DeclareMathOperator{\ATFam}{\bm{\mathsf{AT}}}

\newcommand{\Minion}{\mathcal{M}}

\usepackage{pifont}% http://ctan.org/pkg/pifont
\newcommand{\xmark}{\ding{55}}
\newcommand{\gcmark}{{\color{teal}\checkmark}}
\newcommand{\rxmark}{{\color{red}\xmark}}

\newcommand\bstr[1]{\ensuremath{\texttt{#1}}}

\usepackage{rotating}

\usepackage{xcolor}
\usepackage{tabularray}
\UseTblrLibrary{booktabs}

\DefTblrTemplate{firsthead}{caption}{%
  \makebox[\tablewidth]{\parbox{\columnwidth}{%
    \UseTblrTemplate{caption}{normal}%
  }}%
}
\SetTblrTemplate{firsthead}{caption}

\DefTblrTemplate{middlehead,lasthead}{default}{Continued from previous page}

\usepackage{environ}
\newcommand{\OmitTables}{
  \NewEnviron{omittedtable}{\missingfigure{Table omitted here, comment out the OmitTables macro in main.tex to show it.}}
  \let\longtblr\omittedtable
  \let\endlongtblr\endomittedtable
}

\usepackage[normalem]{ulem}
\usepackage{cancel}

\usepackage[format=hang,labelsep=colon]{caption}

\title{On the Usefulness of Promises}

\ifdefined\anonymize
\author{Anonymous submission}
\else
\author{Per Austrin}
\author{Johan Håstad}
\author{Björn Martinsson}
\affil{KTH Royal Institute of Technology}
\fi
\begin{document}

\maketitle

\begin{abstract}
A Boolean predicate $A$ is defined to be promise-useful if $\PCSP(A,B)$ is tractable for some non-trivial $B$ and otherwise it is promise-useless.  We initiate investigations of this notion and derive sufficient conditions for both promise-usefulness and promise-uselessness (assuming $\text{P} \ne \text{NP}$).  While we do not obtain a complete characterization, our conditions are sufficient to classify all predicates of arity at most $4$ and almost all predicates of arity $5$.  We also derive asymptotic results to show that for large arities a vast majority of all predicates are promise-useless.

Our results are primarily obtained by a thorough study of the ``Promise-SAT'' problem, in which we are given a $k$-SAT instance with the promise that there is a satisfying assignment for which the literal values of each clause satisfy some additional constraint.

The algorithmic results are based on the basic LP + affine IP algorithm of Brakensiek et al.~(SICOMP, 2020) while we use a number of novel criteria to establish NP-hardness.
\end{abstract}

\tableofcontents

\clearpage

\section{Introduction}

The class of Constraint Satisfaction Problems ($\CSP$s) gives a very general framework that includes many well-known problems studied in mathematics and computer science. You are given a set of constraints over a set of variables, with each constraint depending only on a  constant number of the variables, and your goal is to find an assignment that satisfies all constraints. By limiting the constraint language $A$, i.e., the types of constraints allowed, it is possible to create a wide range of different problems $\CSP(A)$. Central examples are given by 
$\threeSAT$ and graph $k$-coloring.

A fundamental question for $\CSP$s is to classify them as being tractable or being NP-hard.  In the case when the variables are Boolean-valued this was done already in 1978 when Schaefer \cite{Schaefer} gave a full classification.  Feder and Vardi conjectured in 1997 \cite{FederVardi} that such a dichotomy between polynomial time solvability and NP-hardness holds also for general finite domains, and after a long sequence of partial results this conjecture was fully proven independently by Bulatov \cite{Bulatov} and Zhuk \cite{Zhuk} in 2017.

A majority of CSPs are NP-hard and hence to allow for efficient algorithms some relaxation is needed.  One such relaxation is in the form of approximation algorithms where we ask for an assignment that may not satisfy all constraints but satisfies a non-trivial number of constraints. For example, even if we cannot find an assignment that satisfies all the constraints, maybe we can find an assignment that satisfies 90\% of the constraints using the promise that there is an assignment that satisfies all the constraints.  

A different relaxation is obtained if we instead ask for an assignment that satisfies all constraints in a relaxed form.  Such problems are called Promise Constraint Satisfaction Problems ($\PCSP$s).  One example of a $\PCSP$ is the classical problem of $k$-coloring a $3$-colorable graph. The study of this type of problem is very old, but the concept of $\PCSP$s was only recently formally introduced by Austrin, Guruswami and Håstad in their study of the ``$(2+\epsilon)$-SAT'' problem \cite{2pluseps}.   In $\PCSP(A, B)$, you are given instances of $\CSP(B)$ with the promise that the instance has a solution if interpreted as a $\CSP(A)$ instance. In the case of $k$-coloring a $3$-colorable graph, $\CSP(A)$ is the $3$-coloring problem, and $\CSP(B)$ is the $k$-coloring problem. Note that $\PCSP$s generalize $\CSP$s, since $\PCSP(A, A)$ is the same as $\CSP(A)$. This implies that understanding $\PCSP$s is at least as difficult as understanding $\CSP$, but we also expect that some tools useful for understanding $\CSP$s will be useful for the study of $\PCSP$s.

In the study of $\CSP$s and the resolution of the Feder-Vardi conjecture, algebraic methods play a crucial role, and one of the key concepts here is that of polymorphisms, introduced by Jeavons \cite{Jeavons97, Jeavons98}.   A polymorphism is a function that given multiple solutions to a problem, combines them in some way to form another solution. Polymorphisms are key to understanding the complexity of $\CSP$s, and the set of polymorphisms of a $\CSP$ determines its complexity.  In other words if two $\CSP$s have the same set of polymorphisms, they are of equal computational complexity.  In general, the fewer polymorphisms a $\CSP$ has, the higher is its computational complexity.

Polymorphisms can be extended to $\PCSP$s and the requirement now is that the function, given a set of solutions that satisfy the constraints of $A$, produces a solution to $B$.  It turns out that, also in this case, two problems with the same set of polymorphisms are of equal complexity \cite{symmetricBooleanDichotomy}. This gives a possible avenue of attack to understand the complexity of $\PCSP(A,B)$, but it is only a starting point.  It is known that the existence of some specific polymorphisms immediately gives efficient algorithms and that establishing some specific properties of the set of polymorphisms yields NP-hardness, but the known conditions are far from complementary.  On top of this, given $A$ and $B$, it is many times difficult to get a grip on the corresponding set of possible polymorphisms.  Given these difficulties, the study of $\PCSP$s is still at an early stage and for the rest of this paper we focus on the Boolean case which already is quite challenging.  In other words we assume that the inputs to both $A$ and $B$ are Boolean.  Even in this simple case our understanding is rather limited.  There is dichotomy when the predicates are symmetric \cite{symmetricBooleanDichotomy, symmetricBooleanDichotomy2}.  Brandts and Živný \cite{NAE} study the question when $A$ is a close to the $t$-in-$k$ predicate.  The general understanding of even Boolean $\PCSP$s is however very far from complete and a potential analogue of Schaefer's Theorem for this setting remains out of reach.

In the rest of this paper we restrict our attention to Boolean $\PCSP$s where the constraint language is given by a single pair of predicates (rather than a general collection of predicates).  Thus from now on, whenever we discuss $\PCSP(A, B)$, $A$ and $B$ are understood to be predicates on Boolean strings of some fixed arity.  That said, our general tractability and hardness results are all polymorphism-based and thus in principle they also apply to general constraint languages.

In the setting of approximation algorithms,  Austrin and Håstad \cite{usefulness_predicates} introduced a notion of useless predicates.  We do not give the exact definition here, but the general idea is that a predicate is useless if, even given a promise that there is an assignment that satisfies almost all the constraints, no non-trivial solution can be found in the approximating sense even when the algorithm is allowed to replace the predicate by any other predicate of its choosing.  It turns out that in this situation there is a simple and elegant characterization: assuming the Unique Games Conjecture, a predicate is useless if and only if the set of accepted strings can support a pairwise independent distribution.

We introduce and study an analogous notion of ``promise-uselessness'' for $\PCSP$s.  
In particular, let us say that a predicate $A$ is promise-useful if and only if there exists a (non-trivial) $B$ such that $\PCSP(A, B)$ is solvable in polynomial time.  Otherwise, $A$ is said to be promise-useless.  Apart from being natural in its own right, we hope this notion will be a helpful focus for further exploration and classification of the complex landscape of Boolean $\PCSP$s.

Note that by this definition, if $\CSP(A)$ is tractable, then $A$ is promise-useful. But even if $\CSP(A)$ is NP-hard, it is still possible that $A$ is promise-useful. For example, $\onethreeSAT$ is NP-hard by Schaefer's characterization, but $\onethreeSAT$ is promise-useful since $\PCSP(\onethreeSAT, \threeXOR)$ is tractable (since $\CSP(\threeXOR)$ is solved by Gaussian elimination over $\mathbf{F}_2$).  In general, for any tractable $\CSP(B)$, any predicate $A$ which implies $B$ is promise-useful.  A more interesting example is $\twofourSAT$ (where the objective is to find an assignment satisfying at least $2$ out of the $4$ literals in each clause): it is NP-hard and does not imply any tractable predicate.  However it is still promise-useful, because $\PCSP(\twofourSAT, \fourSAT)$ is tractable, as shown in the characterization of the ``$(2+\epsilon)$-SAT'' problem \cite{2pluseps}.  A basic example of a promise-useless predicate is $\threeSAT$, since the only non-trivial possible $B$ is $\threeSAT$, and $\PCSP(\threeSAT, \threeSAT) = \CSP(\threeSAT)$ which is NP-hard.

Even for a constraint language defined by a single fixed predicate, it is natural in the Boolean setting to consider the so-called \emph{folded} setting, where we allow negation of variables\footnote{This can be phrased in $\CSP$ terminology as including the ``not-equal'' predicate in the template.}.  Another natural variant is the so-called \emph{idempotent} setting, where we allow fixing some variables to constants\footnote{In $\CSP$ terminology this is equivalent to including all singleton unary predicates in the template.}.  Because of this, the notion of promise-useful/promise-useless, even in the single-predicate Boolean case, comes in four different flavors. 

In this work, our focus is the folded case.  It turns out that whether or not we work in the folded idempotent case or folded non-idempotent case does not make a significant difference.  For the rest of the introduction, unless explicitly mentioned otherwise, promise-uselessness is understood to be the folded idempotent Boolean case, a variant we, following \cite{bgs2025}, denote by $\fiPCSP$.  In this setting, it is not difficult to see that understanding which predicates are promise-useful boils down to understanding for which predicates $A$ the $\fiPCSP(A, \OR)$ problem is tractable (we establish this connection formally in \cref{sec:useful intro}).  We refer to this family of problems -- where we are given a $k$-SAT instance with an additional promise that there is a satisfying assignment which satisfies a stronger predicate $A$ for each clause -- as \emph{Promise-SAT} problems.  While the name is new, Promise-SAT was essentially the main focus of early work on general PCSPs \cite{2pluseps}.  In particular, its complexity for symmetric $A$ is well-understood from more general results \cite{bababe2021, symmetricBooleanDichotomy, symmetricBooleanDichotomy2}, but for arbitrary $A$ it has as far as we are aware not been studied in detail before.

\subsection{Results}

We derive general conditions to determine whether a predicate $A$ is promise-useful or promise-useless (assuming $\text{P} \ne \text{NP}$).  We apply these conditions to all predicates of arity up to five\footnote{Code for verifying these results is available at \protect\href{https://github.com/bjorn-martinsson/On-the-Usefulness-of-Promises}{github.com/bjorn-martinsson/On-the-Usefulness-of-Promises}} and it turns out that we can successfully characterize the promise-usefulness of most predicates.
The raw count of the classification is given in \Cref{tab:promise_overview_intro} below. For $k=5$, there are $59$ different (non-equivalent) predicates where we have been unable to determine whether the predicate is promise-useful or promise-useless.  The numbers in this table apply both for the idempotent and non-idempotent settings, as both our algorithms and hardness results are agnostic to this property.

 \begin{table}[h]
    \caption{Summary of classification of promise-usefulness and promise-uselessness of predicates $A$ of arity up to $5$.\label{tab:promise_overview_intro}}
    \centering
    \begin{tblr}{
      colspec = {l|rrrr},
      rowhead = 1,
      row{odd} = {blue9},
      row{1} = {gray9},
      columns = {colsep=4pt},
    } 
         & Total & Useful & Useless & Unknown \\
    \hline
        $k=2$ & $4$ & $4$ & $0$ & $0$ \\
        $k=3$ & $20$ & $16$ & $4$ & $0$ \\
        $k=4$ & $400$ & $230$ & $170$ & $0$ \\
        $k=5$ & $1\,228\,156$ & $156\,135$ & $1\,071\,962$ & $59$ \\
    \end{tblr}
\end{table}

 On the algorithmic side, our general condition for usefulness (yielding the results summarized in this table) is particularly simple and natural.  Somewhat informally, we establish (see \cref{thm:useful_condition} for a formal statement) that $A$ is promise-useful in the folded setting (both in idempotent and non-idempotent settings) if either
 \begin{enumerate}
        \item There is a weighted majority which is satisfied by all satisfying assignments, or
        \item There is a non-trivial parity constraint which is satisfied by all satisfying assignments.
\end{enumerate}
This is more or less a direct consequence of well-known tractability results for Boolean PCSPs, but what is more interesting is that despite our best efforts we have been unable to find any evidence of additional tractable cases, raising the tantalizing (if perhaps unlikely) possibility that this simple condition might exactly characterize promise-usefulness (in the folded Boolean setting).

As mentioned in the preceding section, understanding promise-usefulness boils down to understanding the complexity of the Promise-SAT problem, $\fiPCSP(A, \OR)$.  We consider this an interesting problem in its own right, and in fact most of the work in this paper is devoted to it (yielding the results for promise-usefulness as a byproduct).  Again we derive general tractability and hardness conditions, and apply them to all predicates of arity up to five.  The results are summarized in \Cref{tab:overview_intro} below (the total number of predicates is larger than for usefulness, because there are fewer direct equivalences between predicates).
For arity five we classify all except 189 out of the roughly 18.6 million genuinely different predicates.  Note that for $k=2$, $\fiPCSP(A, \OR)$ is always tractable since $\twoSAT$ is tractable, so the choice of $A$ does not matter. 
 
 \begin{table}[h]
    \caption{Summary of classification of complexity of $\fiPCSP(A, \OR)$ for $A$ of arity up to $5$. \label{tab:overview_intro}}
    \centering
    \begin{tblr}{
      colspec = {l|rrrr},
      rowhead = 1,
      row{odd} = {blue9},
      row{1} = {gray9},
      columns = {colsep=4pt},
    } 
         & Total & Tractable & NP-hard & Unknown \\
    \hline
        $k=2$ & $5$ & $5$ & $0$ & $0$ \\
        $k=3$ & $39$ & $33$ & $6$ & $0$ \\
        $k=4$ & $1\,991$ & $956$ & $1\,035$ & $0$ \\
        $k=5$ & $18\,666\,623$ & $1\,290\,862$ & $17\,375\,572$ & $189$ \\
    \end{tblr}
\end{table}

 From our classification for small arities it is tempting to guess that as the arity increases most predicates tend to be promise-useless.  This is indeed true, and we prove the following asymptotic version of this fact, showing that even exponentially sparse predicates are useless.
\begin{theorem}[Informal version of \cref{random useless}]
For any $s = \omega(k \cdot 2^{5k/6})$, a uniformly random $k$-ary Boolean predicate with $s$ satisfying assignments is promise-useless with high probability, in all four settings (folded/non-folded and idempotent/non-idempotent).
\end{theorem}
  The bound on $s$ in this result is likely far from tight and we suspect that this result is true with a much smaller bound on $s$, perhaps even polynomial, though we have not been able to prove this and new methods would be needed to push $s$ down to $2^{o(k)}$.
  Supplementing this result, we show (see \cref{thm:BLP+AIP random threshold} for precise statement) that the BLP+AIP algorithm (which, as described in the next section, is the source of our tractability results) stops working once a random predicate of arity $k$ accepts $s = \omega(k)$ strings out the $2^k$ possible inputs.  This leaves a large set of predicates that we are unable to classify.

\subsection{Overview of Techniques}

A surprisingly powerful algorithm for establishing that $\PCSP(A,\OR)$ is tractable is the basic LP + affine IP algorithm of Brakensiek et al.~\cite{BLPAFFINE}.  This algorithm is applicable assuming that the $\PCSP$ admits block-symmetric polymorphisms of arbitrarily large arity.  While for larger domains there are known examples of tractable $\PCSP$s that require different algorithms \cite{LichterPago}, no such examples are (currently) known for Boolean $\PCSP$s, and BLP+AIP is the only algorithm used in the current paper to derive positive results.

The specific polymorphisms we use with BLP+AIP are partly the expected ones, but it has a slight twist.  Unsurprisingly we have the three standard polymorphism families majority, odd parity and alternating threshold, but it turns out that in some cases the tractability of $\fiPCSP(A,\OR)$ needs idempotent versions of minority and even parity which, as far as we are aware, have not been employed before.  It is not difficult to give necessary and sufficient conditions in terms of $A$ for each of these families to be applicable.  On top of being explicit these conditions also turn out to be easy to check by computer for a given predicate $A$.  Thus the we do not introduce any essentially new techniques to establish tractability, we simply make good use of existing techniques.

 On the other hand when it comes to establishing that $A$ is promise-useful, the algorithmic condition becomes even simpler as described in the preceding section.   While for $\fiPCSP(A, \OR)$ we need all five polymorphism families described above (majority, odd parity, alternating threshold, idempotized minority, and idempotized even parity), it turns out that out of these five only majority and odd parity are needed to establish promise-usefulness: if any of the other three families can be used to establish promise-usefulness, then either majority or odd parity can as well.

 Our hardness results are obtained through the study of the minion of polymorphisms of $\fiPCSP(A,\OR)$.  In a small extension of previous results we show that if all such polymorphisms have a small fixing assignment then the problem is NP-hard.  In other words, for every polymorphism, it is possible to fix a constant number of inputs in a way such that the output is determined.  This is the basic criterion we use for hardness but we develop a number of distinct, but similar, conditions to establish that a given polymorphism minion has this property.

These conditions are based on establishing properties shared by all polymorphisms and that some particular (low arity) functions are not polymorphisms.  A toy example would be to show that all polymorphisms of $\fiPCSP(A, \OR)$ are monotone and that there is no polymorphism $f$ of arity $t+1$ such that if its first input is fixed to $0$ then the restricted function becomes an $\AND$ of $t$ variables.  It is then not difficult to prove that all polymorphisms of $\fiPCSP(A, \OR)$ have a fixing set of size $t-1$ from which it follows that $\fiPCSP(A, \OR)$ is NP-hard by known theorems.  In reality we use several more complicated properties of the polymorphism minion, culminating in four different conditions (\cref{match+ADA,,invmatch+ADA,,thm:unate+ADA,thm:split}).  Even for a given $A$ these conditions are most of the time too cumbersome to check by hand, but sufficiently simple that they can be checked relatively quickly by a computer.

As described in the preceding section, these diverse techniques are (somewhat surprisingly) sufficient to completely understand the complexity of $\fiPCSP(A, \OR)$ for all predicates up to arity four --  the small fixing assignment condition exactly complements the block-symmetry condition of the BLP+AIP algorithm for these predicates.  For the unclassified predicates of arity five, we know that they are not solved by BLP+AIP, but we do not know whether or not they satisfy the small fixing assignment condition.  In other words it is conceivable that the ``gap'' in our knowledge here is due to the concrete conditions (\cref{match+ADA,,invmatch+ADA,,thm:unate+ADA,thm:split}) only being sufficient for small fixing assignments, and not necessary.

\subsection{Organization}
\Cref{sec:Preliminaries} covers notation and background material used throughout the paper.
Then in \cref{sec:useful intro} we formally define and give some initial observations on promise-usefulness. 
Following this we analyze tractability in \cref{sec:tract}, applying the BLP+AIP algorithm and discussing the five families of block-symmetric polymorphisms that appear. 
We then turn to hardness in \cref{sec:hard} and give the four different conditions for small fixing assignments, Theorems \ref{match+ADA}, \ref{invmatch+ADA}, \ref{thm:unate+ADA} and \ref{thm:split}. 
%As a brief interlude we discuss some computational aspects of both tractability and hardness conditions in \Cref{sec:computational aspects}, describing how to algorithmically verify that these conditions hold.
In \cref{sec:small} we apply these methods to analyze the tractability and hardness of all predicates of arity up to five. 
The asymptotic bounds for large arities are established in \cref{sec:large arity}, and we give some general conclusions and discuss open problems in \cref{sec:conclusions}.
%\Cref{fixing proof appendix} contains a proof that small fixing assignments imply NP-hardness based on a result by Banakh and Kozik \cite{BK24}. 
%\Cref{blocksymmetry proofs} lists and proves necessary and sufficient conditions for when $\MajFam, \ParFam$ and $\ATFam$ are polymorphisms. 
%\Cref{sec:unkn5} contains the tables that were excluded from \cref{sec:small} because of their large size.

\section{Preliminaries} \label{sec:Preliminaries}

For a logical statement $P$, we use the Iverson notation where $[P]$ is defined to be 1 if $P$ is true and 0 otherwise.  For a positive integer $n$, $[n]$ denotes $\{1,\ldots, n\}$ (and in particular a generic binary string $x \in \{0,1\}^n$ is $1$-indexed and written as $x_1x_2\ldots x_n$).

Throughout the paper, we use the following notation for binary strings. The exclusive or of two binary strings $x$ and $y$ (of equal length) is denoted by $x \XOR y$.  For $b \in \{0,1\}$, $b^\ell$ denotes the all-$b$ string of length $\ell$. The number of ones in a binary string $x$ is denoted as $w(x) = \sum x_i$.  For a subset $S = \{j_1, \ldots, j_r\} \subseteq [k]$, $x_S$ is the projection of $x$ onto $S$, i.e., the binary string $(x_{j_1}, x_{j_2}, \ldots, x_{j_r})$ of length $|S| = r$. 
We let $\neg$ denote Boolean negation and extend this operation to bit-strings by applying it component-wise.  A Boolean function $f: \{0,1\}^\ell \rightarrow \{0,1\}$ is \emph{folded} if $f(\neg x)=\neg f (x)$ for every $x$, and \emph{idempotent} if $f(0^\ell)=0$ and $f(1^\ell)=1$.  We often blur the distinction between sets and Boolean vectors and for a set $S \subseteq [n]$ we let $f(S)$ be $f$ applied to the indicator vector of $S$ (i.e., the vector $x \in \{0,1\}^n$ defined by $x_i=[i \in S]$).

We identify a \emph{predicate} $A$ of \emph{arity} $k$ with a subset of $\{0,1\}^k$ (the set of \emph{accepting assignments} of $A$).   In order to avoid trivial cases we do not allow $A$ to be empty or to contain all strings of length $k$.  On the other hand we do allow $A$ to be independent of some of its coordinates and also that all strings in $A$ share the same value of some coordinate.

For a $k$-ary predicate $A \subseteq \{0,1\}^k$ and a binary string $p \in \{0,1\}^k$, we define $A \XOR p = \{\, a \XOR p \,:\, a \in A\,\}$.  In particular $A \XOR 1^k$ is the predicate formed by negating all accepting assignments of $A$.

We shall frequently be concerned with matrices where the columns are accepting assignments of a predicate. Given a predicate $A \subseteq \{0,1\}^k$ and an integer $\ell$, let $A^\ell$ be the set of all $k \times \ell$ matrices $M$ where each column is an accepting assignment of $A$. Its columns are denoted by $M^1,\ldots,M^\ell$, and its rows are denoted by $M_1,\ldots,M_k$. The entry at row $i$ column $j$ is denoted by $M_i^j$. 

For a set $X \subseteq \RR^k$ we denote by $K(X)$ the convex hull of $X$.
The following standard (see for instance \cite{boyd2004convex}),
theorem about separating convex sets by hyperplanes is useful for us.

\begin{theorem}\label{thm:separate}
Let $K_1$ and $K_2$ be two disjoint convex sets in $\mathbb{R}^k$ and suppose $K_1$ is closed. Then there exist real numbers such $c_1 \ldots c_k$ and $b$ such that $\sum_{i=1}^k c_i x_i \geq b$ for
any $x\in K_1$ while $\sum_{i=1}^k c_i x_i < b$ for
any $x\in K_2$.
\end{theorem}

\subsection{CSPs, PCSPs, and Polymorphisms}

Given a predicate $A$ we can define the $\CSP(A)$ problem.

\begin{definition}
    Let $A \subseteq \{0,1\}^k$ be a $k$-ary predicate.  An instance $\mathcal{I}=(C,X)$ of $\CSP(A)$ consists of a set of variables $X = \{x_1,\ldots x_n\}$ and a set of constraints $C = \{c_1,\ldots,c_m\}$, where each constraint $c_i \in X^k$ is a $k$-tuple of variables. An assignment $\alpha: X \rightarrow \{0,1\}$ of values to the variables is a satisfying assignment to $\mathcal{I}$ if $\alpha(c_i) \in A$ for every constraint $c_1, \ldots c_m$.  $\mathcal{I}$ is a Yes-instance if there exists a satisfying assignment to $\mathcal{I}$. Otherwise, $\mathcal{I}$ is a No-instance.

    We, most of the time, apply predicates to literals and formally in this case we have $2n$ variables $X = \{x_1, \overline{x_1}, x_2, \overline{x_2}, \ldots, x_n,\overline{x_n}\}$, and an assignment $\alpha: X \rightarrow \{0,1\}$ where $\alpha(x_i) = \neg \alpha(\overline{x_i})$ for every $i$.  We sometimes allow constants and in such a case we have a $\CSP(A)$ instance extended with a variable denoted $x^b$ that always takes the value $b$.
\end{definition}

For two $k$-ary predicates $A$ and $B$ such that $A$ implies $B$ (i.e., $A \subseteq B$), we can define $\PCSP(A, B)$.

\begin{definition}
\label{def:pcsp}
    Let $A, B \in \{0,1\}^k$ be $k$-ary predicates such that $A \subseteq B$.
    An instance $\mathcal{I}=(C,X)$ of the $\PCSP(A, B)$ problem consists of a $\CSP(A)$ instance, and the goal is to distinguish between:
    \begin{description}
        \item[Yes] $\mathcal{I}$ is a satisfiable $\CSP(A)$ instance
        \item[No] $\mathcal{I}$ is not even satisfiable when interpreted as a $\CSP(B)$ instance (i.e., there is no $\alpha: X \rightarrow \{0,1\}$ such that $\alpha(c) \in B$ for every constraint $c \in C$).
    \end{description}
    \end{definition}

Note that CSPs are a special case of PCSPs since $\PCSP(A,A) = \CSP(A)$.  \cref{def:pcsp} describes the decision version of $\PCSP(A, B)$.  There is also a search version of the problem, where we are given an instance that is promised to be a Yes instance, and the goal is to find a satisfying for the corresponding $\CSP(B)$ instance.  Unlike basic $\CSP$s, it is a major open problem whether the decision and search versions of $\PCSP$s are equivalent (with some recent progress indicating that search may be harder than decision \cite{LarrauriPCSPSearch}).  

In the definition of $\CSP$s and $\PCSP$s, the distinction of allowing or disallowing repetition (the same variable appearing multiple times in a clause) is unimportant since the two are polynomial time equivalent (\cite{symmetricBooleanDichotomy}, Section 6.6).

Let us make the useful and obvious observation that stronger promises cannot be worse than weaker promises.

\begin{fact}
\label{pcsp monotonicity fact}
Suppose $A \subset A' \subset B$, then $\PCSP (A', B)$ is not easier than $\PCSP(A, B)$.  In particular if the latter is NP-hard so is the former and if the former is tractable so is the latter.  Similarly if $A \subset B \subset B'$ then $\PCSP(A,B)$ is not easier than $\PCSP(A, B')$.
\end{fact}

Let us proceed to define polymorphisms of $\PCSP$s. These are functions $f$ which, when given multiple solutions to an instance of $\CSP(A)$, return a solution to the corresponding $\CSP(B)$ instance. The formal definition is the following.

\begin{definition}
    A Boolean function $f:\{0,1\}^\ell \rightarrow \{0, 1\}$ is a \emph{polymorphism} of $\PCSP(A, B)$ if and only if for all matrices $M \in A^\ell$, $f(M) \in B$. Here $f(M)$ denotes the column vector of length $k$ obtained by evaluating $f$ on each row of $M$, i.e., $f(M) = (f(M_1), f(M_2), \ldots , f(M_k))$.  The set of polymorphisms of $\PCSP(A, B)$ is denoted by $\Pol(\PCSP(A,B))$.

     The cases of allowing negations or fixed constants to the $\PCSP$ problem restrict the set of polymorphisms as follows.  If we apply our predicates to literals, then all polymorphism must be folded and we denote the problem by $\fPCSP$ and if we also allow constants then all polymorphisms must be idempotent and we denote this class by $\fiPCSP$.  In other words.
    \begin{itemize}
        \item $\Pol(\fPCSP(A, B)) = \{ \, f \in \Pol(\PCSP(A, B)) \,:\,\text{$f$ is folded} \,\}$
        \item $\Pol(\fiPCSP(A, B)) = \{ \, f \in \Pol(\PCSP(A, B)) \,:\,\text{$f$ is folded \emph{and} idempotent}  \}$
    \end{itemize}
\end{definition}

If a Boolean function is not a polymorphism of $\PCSP(A,B)$, then there must exist some obstruction $M \in A^\ell$ witnesses this fact. 

\begin{definition}
    Given a $\PCSP(A,B)$, $M \in A^\ell$ is called an \emph{obstruction} for $f$ if $f(M) \not \in B$.
\end{definition}

\subsection{Minors and Minions}

We start with a simple but useful definition.

\begin{definition}
    Let $f:\{0,1\}^\ell \rightarrow \{0,1\}$ be a Boolean function. 
    For any $\pi: [\ell] \rightarrow [\ell']$, the function $f_\pi: \{0,1\}^{\ell'} \rightarrow \{0,1\}$ defined by
    \begin{equation*}
        f_\pi(x_1, \ldots, x_{\ell'}) = f(x_{\pi(1)}, \ldots x_{\pi(\ell)})
    \end{equation*}
    is called a \emph{minor} of $f$.
\end{definition}

One informal way to view this is that a minor is obtained by identifying some sets of variables.   Note that it is not allowed to fix variables to constants.

\begin{definition}
    A (Boolean function) \emph{minion} $\Minion$ is a set of Boolean functions (not necessarily of the same arity) such that for every $f \in \Minion$, every minor $f_\pi$ of $f$ is also a member of $\Minion$.
\end{definition}

It is a well-known and easy to prove fact that the set of polymorphisms of a PCSP forms a minion.
A useful consequence of this is that once we have established that some small, constant size $g$ is not a polymorphism we get a structure theorem on the set of all polymorphisms, because not containing $g$ as a minor anywhere is indeed a severe restriction when $f$ is of large arity.
One general instantiation of this phenomenon that is of use for us is the following.

\begin{definition}
    A family $\mathcal{F}$ of Boolean functions is \emph{essentially minion-atomic} if for every Boolean function minion $\Minion$ it holds that either $\mathcal{F} \subseteq \Minion$, or $\mathcal{F} \cap \Minion$ is finite.
\end{definition}

In particular, if $\mathcal{F}$ is essentially minion-atomic and some fixed $f \in \mathcal{F}$ is not a polymorphism of $\PCSP(A, B)$, then $\PCSP(A, B)$ only admits finitely many polymorphisms from $\mathcal{F}$.  As we see below in \cref{sec:symmetric_examples} some families of standard Boolean functions have this property.

\subsection{Tractability Conditions for PCSPs}

For $\PCSP$s, one of the more remarkable algorithms that makes use of polymorphisms is the so-called basic LP relaxation + affine IP relaxation algorithm by Brakensiek et al.~\cite{BLPAFFINE} which we from now on simply refer to as ``BLP+AIP''.  It combines two commonly seen relaxations of the following Boolean problem. 

Given an instance $\mathcal{I}$ of $\CSP(A)$, we naturally have a Boolean variable for each variable in $\mathcal{I}$, and  we add one variable for each constraint and each element of $A$.  This variable is supposed to be true if the constraint is satisfied by the indicated element of $A$.  We require the sum of all Boolean variables corresponding to the same constraint in $\mathcal{I}$ to be $1$ and the assignment to the variables and constraints to be consistent.

The basic LP relaxation of this Boolean problem is defined as the relaxation where the Boolean variables are relaxed to rational numbers in $[0,1]$. The affine IP relaxation takes the Boolean variables and instead relaxes the Boolean variables to integer variables. 

Brakensiek et al. showed that if $\Pol(\PCSP(A, B))$ contains infinitely many symmetric (or block-symmetric with increasing block sizes) polymorphisms, then the combination of these two relaxations can be used to solve (the decision version) of $\PCSP(A, B)$ \cite{BLPAFFINE}. It turns out that this algorithm is able to solve all tractable Boolean $\CSP$s (Schaefer's dichotomy theorem). There are also natural generalization of this algorithm requiring consistency of larger (but still constant) size subsets of the variables.  This potentially makes it more powerful but we do not know of a Boolean CSP solved by this generalization and not by the basic combination.

\begin{definition}
    A function $f: \{0,1\}^\ell \rightarrow \{0,1\}$ is said to be \emph{symmetric} if $f(x)$ depends only on $w(x)$, and \emph{block-symmetric} with parameters $b$ and $m$ which are both constants independent of $\ell$, if there exists a partition $S_1, S_2,\ldots, S_b$ of $[\ell]$, of sizes $\geq m$, such that $f(x)$ depends only on $w(x_{S_1}),\ldots,w(x_{S_b})$. 
\end{definition}

\begin{remark}
    Any symmetric polymorphism is also a block symmetric polymorphism, with a single block.  In the block-symmetric case, the parameter $m$ is called the width of $f$.  The case of two blocks is also of special interest to us and we use the notation $(\ell_1, \ell_2)$-block-symmetric polymorphism for such a function with block sizes $\ell_1$ and $\ell_2$.
\end{remark}

The power of the BLP+AIP algorithm is then characterized by the following theorem.

\begin{theorem}{\cite[Theorem 5.1]{BLPAFFINE}}
    \label{thm:blpaffine}
    The following three statements are equivalent:
    \begin{itemize}
        \item The $\PCSP(A,B)$ decision problem can be solved in polynomial time by the BLP+AIP algorithm.
        \item $\Pol(\PCSP(A,B))$ contains infinitely many block-symmetric polymorphisms of arbitrary large width.
        \item For every $\ell \geq 1$, $\Pol(\PCSP(A,B))$ contains an $(\ell, \ell +1)$-block-symmetric polymorphism. 
    \end{itemize}
\end{theorem}

The last condition is very useful as it can be efficiently checked for concrete $A$ and $B$ and small values of $\ell$.  The non-existence for any specific value of $\ell$ rules out the possibility to apply BLP+AIP.

Note that while BLP+AIP only in general solves the decision version of a PCSP, if the block-symmetric polymorphisms are explicit enough it is possible to solve the search version.  In particular this is true in all applications of BLP+AIP to explicit predicates in the current paper.

\subsection{Hardness Conditions for PCSPs} \label{sec:fixing_set}

There are many ways that polymorphisms can be used to show that a $\PCSP$ is NP-hard. A general approach is to make use of dictator-like properties of the polymorphisms themselves to construct a reduction from gap label cover to $\PCSP$.  There are many possible notions of being ``dictator-like'' and the one we use is based on the following definition.

\begin{definition}
    A function $f: \{0,1\}^\ell \rightarrow \{0,1\}$ has a \emph{$t$-fixing set}
    if there exists a set $T \subseteq [\ell]$ of size $t$ such that $f(x)=1$ whenever $x_T = 1^t$.  More generally $f$ has a {\em $t$-fixing assignment $(T, \alpha)$} if there exists $T \subseteq [\ell]$ of size $t$ and a partial assignment $\alpha \in \{0,1\}^t$ such that $f(x)=1$ whenever $x_T = \alpha$.
\end{definition}

\begin{remark}
    Since our focus in this paper is on folded functions $f$, it makes no difference in the notion of fixing assignments whether we also allow the function to be fixed to $0$ (since we can simply negate $\alpha$ to fix the function to the opposite value).
\end{remark}

Brakensiek and Guruswami \cite[Theorem 5.3]{symmetricBooleanDichotomy} showed that if all polymorphisms of a PCSP have small fixing sets then the PCSP is NP-hard.  This can be naturally extended \cite[Corollary 5.13]{algebraicApproach} to a general setting where we can define a ``choice function'' $C(f)$ which for each polymorphism $f$ identifies a small number of ``relevant'' coordinates, in a way that behaves consistently across minors (meaning that $\pi(C(f)) \cap C(f_\pi) \ne \emptyset$ for every minor $f_\pi$ of $f$).  This approach, which characterizes when the standard reductions from the basic Gap Label Cover problem is applicable, has subsequently been generalized to ``layered choice'' \cite{BWZ21} (corresponding to reductions from Layered Label Cover) and ``injective layered choice'' \cite{BK24} (corresponding to reductions from Smooth Layered Label Cover).  The hardness condition based on small fixing sets can naturally be relaxed to only require small fixing assignments.

\begin{restatable}{theorem}{FixingAssignmentHardness}
    \label{thm:fixing}
    If there exists a $t$ such that every $f \in \Pol(\fPCSP(A,B))$ has a $t$-fixing assignment, then $\fPCSP(A,B)$ is NP-hard. Likewise, if there is a $t$ such that every $f \in \Pol(\fiPCSP (A,B))$ has a $t$-fixing assignment, then $\fiPCSP (A,B)$ is NP-hard.
\end{restatable}

Essentially this theorem was shown by Guruswami and Sandeep \cite{rainbow_fixing_assignment} though they only state it for rainbow coloring.  Unlike the fixing set condition, \cref{thm:fixing} does not immediately follow from the basic Gap Label Cover problem but instead seems to require Smooth Label Cover.  For completeness, we give a short proof of this result in \cref{fixing assigns intersect} based on the injective choice condition of Banakh and Kozik \cite{BK24}.

The extension from fixing sets to fixing assignments is not difficult but it is needed for some of our results.  There are many cases where $\Pol(\fiPCSP( A, B))$ contains functions without small fixing sets but where all functions have small fixing assignments.

\subsection{Standard Boolean Functions} \label{sec:symmetric_examples}

Let us recall various standard (and a few not so standard) families of Boolean functions that are relevant for us.  We first recall three families of (block-)symmetric functions that have previously been successfully used in conjunction with the BLP+AIP algorithm.

\begin{definition}
    For odd $\ell \ge 1$, the majority function $\Maj_\ell: \{0,1\}^\ell \rightarrow \{0,1\}$ is defined by $\Maj_\ell(x) = [w(x) \ge \ell/2]$.
\end{definition}

\begin{claim}
\label{maj atomic}
    The family $\MajFam = \{\,\Maj_\ell \,|\,\text{$\ell$ odd}\,\}$ is essentially minion-atomic.
\end{claim}

\begin{proof}
    Let $\Minion$ be a minion and suppose $\Maj_\ell \not\in \Minion$.  For any odd $\ell' = d \cdot \ell + r$ ($d \ge 1$ and $0 \le r < \ell$), consider the function
    \[
    f(x_1, \ldots, x_\ell) = \left[ \sum_{i=1}^r x_i + d \sum_{i=1}^{\ell} x_i \ge \ell'/2 \right]
    \]
    Note that $f$ is a minor of $\Maj_{\ell'}$ (obtained by identifying $r$ groups of $(d+1)$ variables, and $\ell-r$ groups of $d$ variables).  Furthermore, if $d \ge \ell$ then $f = \Maj_\ell$.  This shows that $\Minion$ cannot contain $\Maj_{\ell'}$ for any $\ell' \ge \ell^2$, i.e., $\Minion$ contains only finitely many majority functions and hence $\MajFam$ is essentially minion-atomic.
\end{proof}

\begin{definition}
    \emph{Parity} of arity $\ell$ is the function $\Par_\ell: \{0,1\}^\ell \rightarrow \{0,1\}$, defined by
    \[
    \Par_\ell(x) = \BigXOR_{i \in [\ell]} x_i.
    \]
\end{definition}

Note that $\Par_\ell$ is folded if and only if the arity $\ell$ is odd.  The following claim is easy to verify.

\begin{claim}
\label{par atomic}
    The family $\ParFam = \{\,\Par_\ell \,|\,\text{$\ell$ odd}\,\}$ is essentially minion-atomic.
\end{claim}

\begin{definition}
    For odd $\ell \ge 1$, the \emph{alternating threshold} function $\AT_\ell: \{0,1\}^\ell \rightarrow \{0,1\}$ is defined by
    \begin{equation}
        \AT_\ell(x) = \left[\sum_i (-1)^{1+i+x_i} > 0\right].
    \end{equation}
\end{definition}

Again it easy to see that this family is minion-atomic and we leave the verification to the reader.
\begin{claim}
\label{at atomic}
    The family $\ATFam = \{\,\AT_\ell \,|\,\text{$\ell$ odd}\,\}$ is essentially minion-atomic.
\end{claim}

For the hardness results, we are also interested in various other, mostly standard, functions and frequently need the fact that these are minion-atomic, which we record in the following simple claim. 

\begin{claim}
    \label{obstructions atomic}
    Let $\{f_\ell\}$ be one of the following function families:
    \begin{align*}
        \AND_\ell(x_1, \ldots, x_\ell) &= x_1 \wedge x_2 \wedge \ldots \wedge x_\ell \\
        \OR_\ell(x_1, \ldots, x_\ell) &= x_1 \vee x_2 \vee \ldots \vee x_\ell \\
        \NAND_\ell(x_1, \ldots, x_\ell) &= \neg \AND_\ell(x_1, \ldots, x_\ell) = \OR_\ell(\neg x_1, \ldots, \neg x_\ell) \\
        \NOR_\ell(x_1, \ldots, x_\ell) &= \neg \OR_\ell(x_1, \ldots, x_\ell) = \AND_\ell(\neg x_1, \ldots, \neg x_\ell) \\
        \xNOR_\ell(x_1, \ldots, x_\ell) &= x_1 \wedge \NOR_{\ell-1}(x_2, \ldots, x_\ell)
    \end{align*}
    For any minion $\Minion$, if $f_{\ell} \not\in \Minion$ for some $\ell \ge 2$ then $f_{\ell+1} \not\in \Minion$.  In particular, each of these five families of functions is essentially minion-atomic.
\end{claim}

That this is true is not difficult to see as identifying two appropriate variables in $f_{\ell+1}$ results in $f_\ell$.

\section{Promise-usefulness}
\label{sec:useful intro}

The main new notion in this paper is the concept of promise-usefulness (and promise-uselessness), formally defined as follows.

\begin{definition}
\label{def:promise-useful}
    A predicate $\emptyset \ne A \subsetneq \{0,1\}^k$ is \emph{promise-useful} if there exists a non-trivial relaxation $B \supseteq A$ such that $\PCSP(A, B)$ is solvable in polynomial time. Otherwise $A$ is \emph{promise-useless}.
\end{definition}

By ``non-trivial'' we mean that $B \ne \{0,1\}^k$. In the non-folded case (without negated literals) we would also demand that $B$ contains neither $0^k$ nor $1^k$ since otherwise every instance has a trivial satisfying assignment.  Throughout the paper we assume that $\text{P} \ne \text{NP}$, and whenever we make a claim that some predicate is promise-useful, this is always under this assumption.

To make it apparent which situation we are in (folded and/or idempotent), we use the terminology $\fPCSP$-useful for the folded case and $\fiPCSP$-useful for the folded idempotent case.

As a first step, let us note that promise-usefulness in the folded setting boils down to understanding the Promise-SAT problem.

\begin{lemma} \label{lemma:OR}
    The predicate $A$ is $\fPCSP$-useful (resp. $\fiPCSP$-useful), if and only if there exists $b \not \in A$ such that $\fPCSP(A \XOR b, \OR)$ (resp. $\fiPCSP(A \XOR b, \OR)$) is in P.
\end{lemma}

\begin{proof}
Suppose $A$ is promise-useful, i.e., there is a non-trivial $B$ such that $\fPCSP(A, B)$ is tractable, and let $b$ be an arbitrary assignment that does not belong to $B$. Then $\fPCSP(A, \{0,1\}^k \setminus \{b\})$ is also tractable by \cref{pcsp monotonicity fact}. Furthermore, since we have negations, $\fPCSP(A, \{0,1\}^k \setminus \{b\})$ is equivalent with $\fPCSP(A \XOR b, \OR)$.  Conversely, if $\fPCSP(A \XOR b, \OR)$ is tractable for some $b \not \in A$ then $\fPCSP(A, \OR{} \XOR{} b)$ is tractable and witnesses that $A$ is promise-useful.  The proof is the same for the idempotent case.
\end{proof}

So in the folded case, the concept of promise-usefulness boils down to understanding the complexity of various Promise-SAT problems and we turn to this question in \cref{sec:tract} and \cref{sec:hard} below.

\begin{remark}
    For the non-folded case, where we do not have negations, a lemma similar to \cref{lemma:OR} is true, but with the OR predicate replaced by the not-all-equal ($\NAE$) predicate: $A$ is promise-useful if and only if $\PCSP(A, \NAE )$ is tractable.  However, since our main focus in this paper is the folded case, we refrain from discussing this further.
\end{remark}

\begin{remark}
    One can define promise-usefulness as either an \emph{adaptive} or \emph{non-adaptive} notion.  In the non-adaptive version (which is what \cref{def:promise-useful} defines), there must exist a fixed $B$ depending only on $A$, such that $A$ is promise-useful for $B$ ($\PCSP(A, B)$ is tractable).  In the adaptive version, an algorithm would be allowed to choose $B$ also based on the given instance of $\CSP(A)$ -- i.e., given as input a satisfiable $\CSP(A)$ instance, the goal becomes to find some non-trivial $B \supseteq A$ and satisfy the instance as a $\CSP(B)$ instance.   In the non-folded setting this is equivalent with finding an assignment to the variables such that none of the input $k$-tuples is constant, i.e., a $2$-coloring of the underlying hypergraph, and thus adaptive vs.~non-adaptive are trivially equivalent here.  In the folded setting, it is equivalent to finding an assignment to the variables such that not all $k$-tuples appear.  If $\fPCSP(A, B)$ is NP-hard for all non-trivial $B$ then $A$ is $\fPCSP$-useless also in this adaptive sense as we can concatenate the results of all hardness reductions on disjoint sets of variables.  In the opposite direction, we are not aware of a proof that an adaptive useful algorithm must yield an efficient solution to $\PCSP(A,B)$ for a fixed $B$, although this seems intuitively likely.
\end{remark}

\subsection{The Non-Idempotent Case}
\label{sec:non-idempotent OR}

As stated above we restrict attention to the folded setting and we study both  $\fiPCSP(A, B)$ as well as $\fPCSP(A, B)$.  By the preceding discussion, our primary interest is the case when $B=\OR$.  As our main tool is to study polymorphisms, the following simple observation is useful, implying that it is sufficient to understand the idempotent polymorphisms.

\begin{lemma}
    \label{lemma:idempotent}
    Let $A$ be a predicate that does not contain $0^k$.  If $A$ contains
    $1^k$, then all polymorphism of $\fPCSP (A, \OR)$ are idempotent. If $A$ does not contain $1^k$ then for every $f \in \Pol(\fPCSP (A, \OR))$ it holds that either $f$ is idempotent or $\neg f$ is an idempotent polymorphism of $\fPCSP (A \XOR 1^k,\OR)$.
\end{lemma}

\begin{proof}
    If $A$ contains $1^k$ and $f$ is a folded polymorphism then $f$ must be idempotent to avoid that the all-ones matrix is an obstruction.  In the second case if $f$ is not idempotent then $\neg f$ is idempotent.  As $\neg f (x)= f(\neg x)= f(x \XOR 1^k)$, $\neg f$ is a polymorphism of $\fPCSP(A \XOR 1^k,\OR)$.
\end{proof}

Note in particular that for a general $A$ not containing $1^k$ this implies:
\begin{enumerate}
    \item $\fPCSP(A, \OR)$ is tractable via the BLP+AIP algorithm if and only if $\fiPCSP(A, \OR)$ \emph{or} $\fiPCSP(A \XOR 1^k, \OR)$ is tractable via the BLP+AIP algorithm.
    \item $\fPCSP(A, \OR)$ has small fixing assignments (i.e., is NP-hard via \cref{thm:fixing}) if and only if both $\fiPCSP(A, \OR)$ \emph{and} $\fiPCSP(A \XOR 1^k, \OR)$ have small fixing assignments.
\end{enumerate}

It is an interesting question whether this holds true in general regardless of algorithm or NP-hardness proof used, in other words: \emph{Is it the case that $\fPCSP(A, \OR)$ is tractable if and only if $\fiPCSP(A, \OR)$ or $\fiPCSP(A \XOR 1^k, \OR)$ is tractable?}

For promise-usefulness, combining this  observation with \cref{lemma:OR} we obtain the following conclusion.

\begin{corollary}
    \label{promise constants irrelevant}
    A predicate $A$ is $\fPCSP$-useful via the BLP+AIP algorithm if and only if it is $\fiPCSP$-useful via the BLP+AIP algorithm, and it is $\fPCSP$-useless via small fixing assignments if and only if it is $\fiPCSP$-useless via small fixing assignments.
\end{corollary}

Thus when it comes to promise-usefulness, being in the idempotent setting (i.e., allowing fixed constants) does not make a difference with the current techniques we have for proving tractability and hardness.
\section{Tractability of Promise-SAT}\label{sec:tract}

Based on our preliminary observations in \cref{sec:useful intro},  understanding promise-usefulness boils down to understanding the Promise-SAT problem, $\fiPCSP(A, \OR)$.  In this section we analyze conditions for this problem to be tractable.

\subsection{Identifying Families of Block-symmetric Polymorphisms}

To identify tractable cases of $\fiPCSP(A, \OR)$, we use the BLP+AIP algorithm, characterized in \cref{thm:blpaffine}.  The requirement for this algorithm is that $\fiPCSP(A, \OR)$ contains infinitely many block-symmetric polymorphisms of arbitrary large arity.  By \cref{thm:blpaffine} this is equivalent to the existence of block-symmetric polymorphisms with two blocks, one with size $\ell$ and the other with size $\ell + 1$ for each value of $\ell$.

For small values of $k$ and $\ell$ the condition that such polymorphisms exist can be checked by computer.  There are then two outcomes. If there is no solution for some $\ell$ then  $\fiPCSP(A, \OR)$ cannot be solved using the BLP+AIP algorithm. On the other hand, if a solution is found this can help us identify block-symmetric polymorphisms of general interest.   Performing such experiments, we have identified five different families of block-symmetric polymorphisms, where for each family $\mathcal{F}$, there exists a predicate $A$ such that $\Pol(\fiPCSP (A, \OR))$ contains infinitely many polymorphisms from $\mathcal{F}$, but only finitely many from the other families. 

Three of these families are the standard idempotent polymorphisms $\MajFam$, $\ParFam$ and $\ATFam$ defined in \cref{sec:symmetric_examples}. The other two are non-standard, closely related to the negated functions minority $\invMaj_\ell(x) = \neg \Maj_\ell(x)$ and even parity $\invPar_\ell(x) = \neg \Par_\ell(x)$.  These two functions are not idempotent and thus are not polymorphisms of $\fiPCSP(A, \OR)$, but we can simply change the values on the constant strings to make them idempotent.

For a function $f: \{0,1\}^\ell \rightarrow \{0,1\}$, we denote by $\id f: \{0,1\}^\ell \rightarrow \{0,1\}$ the \emph{idempotized} function
\[
\id f(x) = \begin{cases} 
   0 & \text{if $x = 0^\ell$} \\
   1 & \text{if $x = 1^\ell$}\\
   f(x) & \text{otherwise}
   \end{cases}
\]
We have the following simple claim, analogous to \cref{maj atomic,par atomic}.  The proofs are analogous to the non-idempotized cases and left to the reader.
\begin{claim}
\label{id atomic}
   The families $\idneg{\MajFam} = \{\,\idneg{\Maj}_\ell \,|\,\text{$\ell$ odd}\,\}$ and $\idneg{\ParFam} = \{\,\idneg{\Par}_\ell \,|\,\text{$\ell$ odd} \,\}$ of idempotized minority and idempotized even parity of odd arities, are both essentially minion-atomic.
\end{claim}

The families $\idneg{\MajFam}$ and $\idneg{\ParFam}$, together with the aforementioned $\MajFam$, $\ParFam$, and $\ATFam$ families then make up our quintet of block-symmetric polymorphism families and we have the following result.

\begin{restatable}{lemma}{FiveFamilyLemma}
    \label{lemma:easy}
    Consider the following five families of idempotent (block-)symmetric functions:
    majority ($\MajFam$), odd parity ($\ParFam$), alternating threshold ($\ATFam$), idempotized minority ($\idneg{\MajFam}$), and idempotized even parity ($\idneg{\ParFam}$).
    For each of these families $\mathcal{F}$, there exists a predicate $A$ such that $\fiPCSP(A, \OR)$ admits infinitely many polymorphisms from $\mathcal{F}$, but only finitely many from the other four families.
    \begin{enumerate}
    % [1,2,3], k=2
    \item $A =\{\bstr{01}, \bstr{10}, \bstr{11}\}$ is an example that only admits $\MajFam$.
    % [1, 2, 4, 7], k=3
    \item $A = \{\bstr{001}, \bstr{010}, \bstr{100}, \bstr{111}\}$ is an example that only admits $\ParFam$.
    % [3, 5, 6, 8, 16], k = 5
    \item $A = \{\bstr{00011}, \bstr{00101}, \bstr{00110}, \bstr{01000}, \bstr{10000}\}$ is an example that only admits $\ATFam$.
    % [3,4,6,8,9], k = 4
    \item $A =\{\bstr{0011}, \bstr{0100}, \bstr{0110}, \bstr{1000}, \bstr{1001}\}$ is an example that only admits $\idneg{\MajFam}$.
    % [7, 10, 13, 16, 19], k = 5
    \item $A = \{\bstr{00111}, \bstr{01010}, \bstr{01101}, \bstr{10000}, \bstr{10011}\}$ is an example that only admits $\idneg{\ParFam}$,
    \end{enumerate}
\end{restatable}

Before proceeding with the proof, which is given in \cref{sec lemma easy proof}, we need characterizations of what it means for $\fiPCSP(A, \OR)$ to admit infinitely many functions from each of these families.  The conditions in the following sections are stated in terms of $\fPCSP(A, \OR)$ admitting infinitely many polymorphisms from one of the families, but note that this is equivalent with $\fiPCSP(A, \OR)$ doing so, since $\fiPCSP$ has all idempotent polymorphisms of $\fPCSP$ and no others.

\subsection{Conditions for \texorpdfstring{$\MajFam$, $\ParFam$, and $\ATFam$}{Maj, Par, and AT}}
\label{sec:Maj Par AT conditions}

We start by establishing necessary and sufficient conditions for the existence of infinitely many polymorphisms from the standard families $\MajFam$, $\ParFam$, and $\ATFam$.  Recall that $K(A)$ denotes the convex hull of $A$. The proofs of the following three standard lemmas can be found in \cref{blocksymmetry proofs}.  The lemmas can be considered "folklore" one possible original source is \cite{bgs2025}.

The existence of majority as a polymorphism of arbitrarily large odd arities can be expressed as a linear program based on the following lemma.

\begin{restatable}{lemma}{MajConditionLemma} \label{lemma:test_maj}
    The following statements are equivalent:
    \begin{enumerate}
        \item\label{item_maj} $\Pol(\fPCSP (A, \OR))$ contains infinitely many polymorphisms from $\MajFam$.
        \item\label{item_K}$[0,1/2)^k \cap K(A) = \emptyset$.
        \item\label{item_sep} There exists integers $c_1, \ldots c_k \geq 0$ such that $\sum_{j=1}^k c_j a_j \geq \sum_{j=1}^k c_j / 2 > 0$ for all $a \in A$. 
    \end{enumerate}
\end{restatable}

The test for odd parity is based on solving an affine equation modulo two.

\begin{restatable}{lemma}{ParConditionLemma} \label{lemma:test_odd}
    The following statements are equivalent:
    \begin{enumerate}
        \item \label{item_oddpari} $\Pol(\fPCSP (A, \OR))$ contains infinitely many polymorphisms from $\ParFam$.
        \item \label{item_oddset} For every odd sized subset $B$ of $A$, $\BigXOR_{s \in B} s \neq 0^k$.
        \item \label{item_alwaysodd} There exists a non-empty subset $\beta \subseteq [k]$, such that $\BigXOR_{i \in \beta} a_i = 1$ for all $a \in A$. 
    \end{enumerate}
\end{restatable}

Finally, for alternating threshold, the characterization corresponds to all accepting assignments of $A$ satisfying some linear equation over the integers.

\begin{restatable}{lemma}{ATConditionLemma} \label{lemma:test_AT}
    The following statements are equivalent:
    \begin{enumerate}
        \item \label{item_AT} $\Pol(\fPCSP (A, \OR))$ contains infinitely many polymorphisms from $\ATFam$.
        \item \label{item_diff} $\{x - y: x,y \in K(A)\} \cap (-\infty,0)^k = \emptyset$.
        \item \label{item_linear} There exists integers $c_1, \ldots c_k \geq 0$, not all $0$, such that $\sum_{j=1}^k c_j a_j$ takes the same value for all $a \in A$. 
    \end{enumerate}
\end{restatable}

\subsection{Conditions for \texorpdfstring{$\idneg{\MajFam}$ and $\idneg{\ParFam}$}{idMaj and idPar}}

Let us now turn to necessary and sufficient conditions for when $\Pol(\fPCSP (A, \OR))$ contains infinitely many polymorphisms of either of the two non-standard families $\idneg{\MajFam}$ and $\idneg{\ParFam}$. These conditions are formulated in terms of $\Pol(\fPCSP (A, \OR))$ containing infinitely many polymorphisms from $\invMajFam$ and $\invParFam$, respectively (and \cref{lemma:test_maj,lemma:test_odd}  can easily be modified to give efficient tests for verifying these containments).  Throughout this section we frequently use the fact that $\idneg{\MajFam}$ and $\idneg{\ParFam}$ are essentially minion-atomic (\cref{id atomic}).  In other words, having only finitely many polymorphisms from e.g.~$\idneg{\MajFam}$ is equivalent to not having $\idneg{\Maj}_\ell$ for some odd $\ell > 0$.

For a predicate $A \subseteq \{0,1\}^k$ and $b \in \{0,1\}$, we say that $A$ has a \emph{forced $1$-bit} if there is an $i$ such that $a_i=1$ for all $a \in A$.  Note that whenever $A$ has a forced $1$-bit, $\Pol(\fPCSP(A, \OR))$ contains \emph{all} idempotent odd functions.
We first show that, for $\fPCSP(A \OR)$, admitting $\idneg{\MajFam}$ or $\idneg{\ParFam}$ implies admitting their non-idempotized counterparts, except in the trivial case when $A$ has a forced $1$-bit.

\begin{lemma}\label{lemma:noidemmaj}
Suppose $A \subseteq \{0,1\}^k$ does not have a forced $1$-bit.  If $\fPCSP (A, \OR)$ admits an infinite number of polymorphisms from $\idneg{\MajFam}$ then it also admits an infinite number of polymorphisms from $\invMajFam$.
\end{lemma}

\begin{proof}
    Since $A$ does not have a forced $1$-bit, there are $k$ (not necessarily distinct) strings $x_1 \ldots, x_{k} \in A$ such that $x_i$ has a $0$ in position $i$.

    Suppose there is an obstruction matrix $M \in A^{\ell}$ showing that $\invMaj_\ell$ is not a polymorphism of $\fPCSP(A, B)$ for some (odd) $\ell$.  Thus in every row of $M$ there are at least $(\ell+1)/2$ ones.  Create a new matrix $M'$ with $\ell' = (k+1) \ell + k$ columns by taking $(k+1)$ copies of each column, and adding the $k$ columns $x_1, \ldots, x_{k}$.  Now $M'$ has at least $(k+1)(\ell+1)/2 = (\ell'+1)/2$ ones in every row, and no row is identically $1$.  Thus this is an obstruction that shows that $\idneg{\Maj}_{\ell'}$ is not a polymorphism and the lemma now follows by the fact that both families are essentially minion-atomic.
\end{proof}

By a similar reasoning we establish the same fact for $\idneg{\ParFam}$ and $\invParFam$.

\begin{lemma}\label{lemma:noidempar}
Suppose $A \subseteq \{0,1\}^k$ does not have a forced $1$-bit.  If $\fPCSP (A, \OR)$ admits an infinite number of polymorphisms from $\idneg{\ParFam}$ then it also admits an infinite number of polymorphisms from $\invParFam$.
\end{lemma}

\begin{proof}[Proof sketch]
  We proceed as in the proof of \cref{lemma:noidemmaj}, but when creating $M'$ we do not make any copies of the existing rows and instead only add two copies each of $x_1, \ldots, x_k$.  This maintains parity of every row while making sure no row is all-ones (and no row of $M$ was all-zeros since $M$ was an obstruction for $\invPar_\ell$ being a polymorphism of $\fPCSP(A, \OR)$), ensuring $\idneg{\Par}$ and $\invPar$ behave the same on $M'$.
\end{proof}

In what follows, we use the notation $A_S^0$ to denote the set $\{\, a_S \,|\, a \in A \text{ such that } a_{\overline{S}}=0^{|\overline{S}|} \,\}$, where $A \subseteq \{0,1\}^k$ is a $k$-ary predicate and $S \subseteq [k]$. This makes $A^0_S$ a kind of projection which only keeps accepting assignments $a \in A$ such that $a_i=0$ for all $i \notin S$.

In the other direction, we have the following general property.

\begin{lemma}
\label{lemma:finite idneg}
    Let $\mathcal{F}$ be an essentially minion-atomic infinite family of idempotent Boolean functions.
    If $\fPCSP(A, \OR)$ admits only finitely many polymorphisms from $\idneg{\mathcal{F}}$ then there exists a subset $S \subseteq [k]$ such that (i) $A_S^0$ is non-empty and does not have a forced $1$-bit, and (ii) $\fPCSP (A_S^0, \OR)$ admits only finitely many polymorphisms from $\overline{\mathcal{F}}$. 
\end{lemma}

\begin{proof}
  Suppose that $\Pol(\fPCSP (A, \OR))$ contains only finitely many polymorphisms from $\idneg{\mathcal{F}}$. Let $M$ be an obstruction for $\idneg{f}$ for some $f \in \mathcal{F}$. Note that since $M$ is an obstruction for an idempotent polymorphism, no row in $M$ is all-ones. Let $S$ be the set of rows that are not equal to $0^\ell$, $k'=|S|$ and let $M'$ denote the $k' \times \ell$ sub-matrix of $M$ given by the rows of $S$.  
  The predicate $A_S^0$ is clearly not empty since the columns of $M'$ come from $A_S^0$.
  And since $M'$ has no all-ones row, $A_S^0$ cannot have a forced $1$-bit, establishing property (i).
  
  By construction $M'$ has no constant rows and hence $\overline{f}$ and $\idneg{f}$ behave the same on $M'$.  Thus $M'$ is also an obstruction of $\overline{f}$ for $\Pol(\fiPCSP(A_S^0, \OR))$, which together with $\overline{\mathcal{F}}$ being minion-atomic establishes property (ii).
\end{proof}

With these pieces in place let us formulate the characterization for $\idneg{\MajFam}$.

\begin{lemma} \label{lemma:condition_mino}
    $\fPCSP(A, \OR)$ admits only finitely many polymorphisms from $\idneg{\MajFam}$ if and only if there exists a subset $S \subseteq [k]$ such that $A_S^0$ is non-empty and does not have a forced $1$-bit, and $\fPCSP (A_S^0, \OR)$ admits only finitely many polymorphisms from $\invMajFam$. 
\end{lemma}

\begin{proof}
   The forward direction is \cref{lemma:finite idneg} with $\mathcal{F} = \MajFam$.  For the other direction, suppose that for some $S \subseteq [k]$, $A_S^0$ is non-empty with no forced $1$-bit and $\fPCSP (A_S^0, \OR)$ admits only finitely many polymorphisms from $\invMajFam$.  By \cref{lemma:noidemmaj}, $\fPCSP(A_S^0, \OR)$ admits also only finitely many polymorphisms from $\idneg{\MajFam}$, but since $\Pol(\fPCSP(A_S^0, \OR)) \supseteq \Pol(\fPCSP(A, \OR))$, the latter then also contains only finitely many functions from $\idneg{\MajFam}$.
\end{proof}

The characterization for $\idneg{\ParFam}$ is analogous.

\begin{lemma}\label{lemma:condition_pario}
    $\fPCSP(A, \OR)$ admits only finitely many polymorphisms from $\idneg{\ParFam}$ if and only if there exists a subset $S \subseteq [k]$ such that $A_S^0$ is non-empty and does not have a forced $1$-bit, and $\fPCSP (A_S^0, \OR)$ admits only finitely many polymorphisms from $\invParFam$.
\end{lemma}

\begin{proof}[Proof sketch.]
   The proof is identical to the proof of \cref{lemma:condition_mino}, with $\Maj$ replaced by $\Par$ and the invocation of \cref{lemma:noidemmaj} replaced by \cref{lemma:noidempar}.
\end{proof}

\subsection{Proof of \texorpdfstring{\cref{lemma:easy}}{Lemma 4.2}} \label{sec lemma easy proof}

In this section we establish our main tractability lemma for $\fiPCSP$s, restated here for convenience.

\FiveFamilyLemma*

\begin{proof}[Proof of \cref{lemma:easy}]
    We verify the five cases one by one.  Note that by \cref{maj atomic,par atomic,at atomic,id atomic}, the five families of functions are all essentially minion-atomic and hence to prove that we do not have infinitely many from one of the families, it suffices to exhibit a single obstruction for that family.
    \begin{enumerate}
    \item $A = \{\bstr{01}, \bstr{10}, \bstr{11}\}$.  In this case we have $a_1 + a_2 \ge 1$ for all $a \in A$, so by \cref{lemma:test_maj} we get infinitely many polymorphisms from $\MajFam$.  Obstructions for the other families are:
    \begin{equation*}
    \Par_3: 
    \begin{pmatrix}
        011 \\
        101
    \end{pmatrix}, 
    \AT_3:
    \begin{pmatrix}
        011 \\
        110    
        
    \end{pmatrix},
    \idneg{\Maj}_3:
    \begin{pmatrix}
        011 \\
        101
    \end{pmatrix},
    \idneg{\Par}_5:
    \begin{pmatrix}   
        00111 \\
        11001
    \end{pmatrix}.
    \end{equation*}

    \item $A = \{\bstr{001}, \bstr{010}, \bstr{100}, \bstr{111}\}$.  In this case we have $a_1 \XOR a_2 \XOR a_3 = 1$ for all $a \in A$, so by \cref{lemma:test_odd} we get infinitely many polymorphisms from $\ParFam$.  Obstructions for the other families are:
    \begin{equation*}
    \Maj_3:
    \begin{pmatrix}
        001 \\
        010 \\
        100
        
    \end{pmatrix}, 
    \AT_3:
    \begin{pmatrix}
        011 \\    
        010 \\
        110
    \end{pmatrix},
    \idneg{\Maj}_5:
    \begin{pmatrix}
        00111 \\    
        01111 \\
        10111
    \end{pmatrix},
    \idneg{\Par}_3:
    \begin{pmatrix}
        001\\    
        010\\
        100
    \end{pmatrix}.
    \end{equation*}

    % \{\bstr{11000}, \bstr{10100}, \bstr{01100}, \bstr{00010}, \bstr{00001}\}
    \item $A = \{\bstr{00011}, \bstr{00101}, \bstr{00110}, \bstr{01000}, \bstr{10000}\}$.  In this case we have $2a_1 + 2a_2 + a_3 + a_4 + a_5 = 2$ for all $a \in A$, so by \cref{lemma:test_AT} we get infinitely many polymorphisms from $\ATFam$.  Obstructions for the other families are:
    \begin{equation*}
    \Maj_3: 
    \begin{pmatrix}
        001 \\    
        010 \\
        000 \\
        100 \\
        100
    \end{pmatrix}, 
    \Par_3:
    \begin{pmatrix}
        000 \\    
        000 \\
        011 \\
        101 \\
        110
    \end{pmatrix},
    \idneg{\Maj}_3:
    \begin{pmatrix}
        000 \\    
        000 \\
        011 \\
        101 \\
        110
    \end{pmatrix},
    \idneg{\Par}_3:
    \begin{pmatrix}
        001 \\    
        010 \\
        000 \\
        100 \\
        100
    \end{pmatrix}.
    \end{equation*}

    % \{\bstr{1100}, \bstr{0010}, \bstr{0110}, \bstr{0001}, \bstr{1001}\}
    \item $A = \{\bstr{0011}, \bstr{0100}, \bstr{0110}, \bstr{1000}, \bstr{1001}\}$. Checking the conditions of \cref{lemma:condition_mino} requires an exhaustive search over sets $S \subseteq [4]$ that we omit in this version.  Obstructions for the other families are:
    \begin{equation*}
    \Maj_3: 
    \begin{pmatrix}
        001 \\    
        010 \\
        100 \\
        100
    \end{pmatrix}, 
    \Par_5:
    \begin{pmatrix}
        00011 \\    
        01100 \\
        10100 \\
        10001 
    \end{pmatrix},
    \AT_5:
    \begin{pmatrix}
        00011 \\    
        01100 \\
        11000 \\
        10010
    \end{pmatrix},
    \idneg{\Par}_3:
    \begin{pmatrix}
        001 \\    
        010 \\
        100 \\
        100
    \end{pmatrix}.
    \end{equation*}

    % \{\bstr{11100}, \bstr{01010}, \bstr{10110}, \bstr{00001}, \bstr{11001}\}
     \item $A = \{\bstr{00111}, \bstr{01010}, \bstr{01101}, \bstr{10000}, \bstr{10011}\}$. Checking the conditions of \cref{lemma:condition_pario} requires an exhaustive search over sets $S \subseteq [5]$ that we omit in this version.  Obstructions for the other families are:
    \begin{equation*}
    \Maj_5: 
    \begin{pmatrix}
        00011 \\    
        01100 \\
        10100 \\
        11000 \\
        10100
    \end{pmatrix}, 
    \Par_3:
    \begin{pmatrix}
        000 \\    
        011 \\
        101 \\
        110 \\
        101
    \end{pmatrix},
    \AT_7:
    \begin{pmatrix}
        0001111 \\    
        0110000 \\
        1100000 \\
        1011010 \\
        1101010
    \end{pmatrix},
    \idneg{\Maj}_3:
    \begin{pmatrix}
        000 \\    
        011 \\
        101 \\
        110 \\
        101
    \end{pmatrix}.
    \end{equation*}

    \end{enumerate}
    
\end{proof}

\subsection{Tractability in the Non-idempotent Case}

To find the (block)-symmetric polymorphisms needed to apply the BLP+AIP algorithm in the non-idempotent case \cref{lemma:noidemmaj} and \cref{lemma:noidempar} tells us that $\idneg{\ParFam}$ and  $\idneg{\MajFam}$ are not needed since they can be replaced by their non-idempotent counter-parts.  It turns out that also $\ATFam$ is not essential.

\begin{lemma}\label{lemma:at}
    If $\Pol(\fPCSP (A,\OR))$ contains infinitely many polymorphisms from $\ATFam$, then $\Pol(\fPCSP(A,\OR))$ also contains infinitely many polymorphisms from at least one of $\MajFam$ or $\invMajFam$.  
\end{lemma}

\begin{proof}
    According to \cref{lemma:test_AT}, if $\Pol(\fPCSP(A,\OR))$ contains infinitely many polymorphisms from $\ATFam$ then there exists integers $c_1,\ldots,c_k \geq 0$, not all $0$, such that $\sum_{i=1}^k c_i a_i$ has the same value $b$ for all $a \in A$.

    If $b \geq \sum_{i=1}^k c_i/2$, then according to \cref{lemma:test_maj}, $\Pol(\fPCSP(A,\OR))$ contains infinitely many polymorphisms from $\MajFam$. It easy to see, by an argument similar to the proof of  \cref{lemma:test_maj}  that if $b \leq \sum_i c_i/2$, then $\Pol(\fPCSP (A,\OR))$ contains infinitely many polymorphisms from $\invMajFam$.
\end{proof}

We conclude from \cref{lemma:noidemmaj,lemma:noidempar,lemma:at} that if $\Pol(\fPCSP(A, \OR))$ contains infinitely many polymorphisms from $\ATFam, \idneg{\ParFam},$ or $\idneg{\MajFam}$, then $\Pol(\fPCSP (A, \OR))$ also contains infinitely many polymorphisms from at least one of $\MajFam, \ParFam, \invMajFam,$ or $ \invParFam$. Thus, to establish tractability of $\fPCSP(A, \OR)$, it is not necessary to consider any of $\ATFam, \idneg{\ParFam},$ or $\idneg{\MajFam}$.

\subsection{Conditions for Promise-usefulness} \label{sec:useful_results}

In the previous sections, we identified some specific families of block-symmetric polymorphisms, $\MajFam, \ParFam, \ATFam, \idneg{\MajFam},$ or $ \idneg{\ParFam}$, that can be used to show that $A$ is $\fiPCSP (A, \OR)$ is tractable.  It turns out that only the first two out of the five are relevant to establish promise usefulness and we have the following tractability theorem.

\begin{theorem} \label{thm:useful_condition}
    A predicate $A \subseteq \{0,1\}^k$ is $\fiPCSP$-useful and $\fPCSP$-useful if either
    \begin{enumerate}
        \item \label{item_first} There exists integers $\alpha_1, \ldots, \alpha_n$ that are not all $0$, such that for all $a \in A$, $\sum_i \alpha_i (a_i - 1/2) \geq 0$.
        \item \label{item_second} There exists a non-empty subset $\beta \subseteq [k]$ such that for all $a \in A$, $\XOR_{i \in \beta} a_i$ is constant (either $0$ or $1$).
    \end{enumerate}
    Furthermore, if $A$ does not satisfy any of the conditions above, then for all $b \notin A$, $\Pol(\fPCSP (A \XOR b, \OR))$ contains at most finitely many polymorphisms from $\MajFam, \ParFam, \ATFam, \idneg{\MajFam},$ and $ \idneg{\ParFam}$.
\end{theorem}

\begin{proof}
That the conditions are sufficient to establish tractability is more or less already established in \cref{lemma:test_maj} and \cref{lemma:test_odd}, respectively. For majority (the first case) we define $b \in \{0,1\}^k$ as
    \begin{equation*}
        b_i = 
        \begin{cases}
            0 \text{ if }\alpha_i\geq0 \\
            1 \text{ otherwise.}
        \end{cases}
    \end{equation*}
    and it is easy to see that $\fiPCSP(A  \XOR b, \OR)$ admits $\MajFam$, and thus it is tractable.  In the second case, if the parity is odd (the constant is $1$), then $\fiPCSP(A, \OR)$ admits $\ParFam$ and thus $\fiPCSP(A, \OR)$ is tractable. Otherwise, if the parity is even (the constant is $0$), then by toggling one bit in $\beta$ we can effectively achieve odd parity. Let $b \in \{0,1\}^k$ be $0$ everywhere expect for a single index $i \in \beta$ where $b_i = 1$. We get that $\fiPCSP(A \XOR b, \OR)$ admits $\ParFam$, and thus $\fiPCSP(A \XOR b, \OR)$ is tractable.
    
    We now prove the final part of the theorem: if neither condition holds, then $A$ is not $\fPCSP$-useful via any of $\MajFam$, $\ParFam$, $\ATFam$, $\idneg{\MajFam}$, or $\idneg{\ParFam}$. The cases of $\MajFam$ or $\ParFam$ follows directly from the corresponding conditions stated in \cref{lemma:test_maj} and \cref{lemma:test_odd}.
    
    By \cref{lemma:noidemmaj,lemma:noidempar,lemma:at}, if $\fPCSP(A \XOR b, \OR)$ admits one of $\ATFam, \idneg{\ParFam},$ or $\idneg{\MajFam}$, then $\fPCSP (A \XOR b, \OR)$ admits at least one of $\MajFam, \ParFam, \invMajFam,$ or $\invParFam$. Since we have already established that $\MajFam$ and $\ParFam$ cannot be polymorphisms, the only remaining cases are $\invMajFam,$ and $\invParFam$. Note that if $\fPCSP(A \XOR b, \OR)$ admits $\invMajFam$ or $\invParFam$, then $\fPCSP(A \XOR b \XOR 1^k, \OR)$ admits $\MajFam$ or $\ParFam$, which we have already ruled out as a possibility. The conclusion is that $A$ cannot be $\fPCSP$-useful via any of $\MajFam$, $\ParFam$, $\ATFam$, $\idneg{\MajFam}$, or $\idneg{\ParFam}$.

\end{proof}
\section{Hardness conditions for Promise-SAT}\label{sec:hard}

Our general tool to prove hardness is \cref{thm:fixing} -- whenever all polymorphisms of $\fiPCSP(A, \OR)$ have small fixing assignments, it is NP-hard.  However it is not clear for a general $A$ how to check this condition, so in this section we develop a number of general conditions that are sufficient (but not necessary) to guarantee that all functions in a minion $\Minion$ have small fixing assignments, while still being relatively easy to test for.

Throughout this section, we often apply terminology for functions to minions, meaning that all functions in the minion have a property.  For example a folded minion is a minion in which all functions are folded.  Similarly we say that $\Minion$ has small fixing sets/assignments if all $f \in \Minion$ have a fixing set/assignment of size at most $t$ for some constant $t$.

\newcommand{\minw}{\operatorname{minw}}

Let us also introduce some notation used in this section.   
For a minion $\Minion$, let $\Minion^0 \supseteq \Minion$ be the minion consisting of all
functions $f$ obtained by taking a function $g \in \Minion$ and fixing
some, possibly empty, set of variables to $0$.  Analogously we denote by $\Minion^1$ and $\Minion^{0,1}$ the minions obtained by allowing fixing variables to $1$, and to both $0$ and $1$, respectively. It is good to keep in mind that even if $\Minion$ only contains folded and idempotent functions this is no longer true for these derived minions.
For a function
$f: \{0,1\}^\ell \rightarrow \{0,1\}$ which is not identically $0$,
let $\minw(f) = \min \{ \,w(x)\,|\,f(x) = 1 \,\}$.

\subsection{Overview}

An obvious \emph{necessary} condition for a folded minion $\Minion$ to have small fixing assignments is that all (non-constant) $f \in \Minion$ have bounded $\minw(f)$ (i.e., a low-weight assignment $x$ such that $f(x) = 1$).  Note that in the case of monotone (non-decreasing) functions, this condition is also sufficient, but in general it is not.  It is not clear whether characterizing this low-weight property is any easier than characterizing small fixing assignments, but our first step is to observe that this property is easily characterized in $\Minion^0$ (which gives a simple sufficient condition for $\Minion$ to have the property):

\begin{claim}
    \label{AND assignment}
    Let $\Minion$ be a minion.
    Then every $f \in \Minion^0$ which is not identically $0$ has $\minw(f) \le t-1$ if and only if $\AND_t \not\in \Minion^0$.
\end{claim}

\begin{proof}
    Clearly if $\AND_t \in \Minion^0$ then this is a function with $\minw(f) \ge t$.  In the other direction, suppose $\AND_t \not\in \Minion^0$, take any not identically $0$ function $f \in \Minion^0$, and let
    $S$ be a minimum-cardinality set of coordinates such that $f(S) = 1$.  Consider the function $g \in \Minion^0$ of arity $|S|$ obtained from $f$ by fixing all variables outside $S$ to $0$. By the minimality of $S$, we see that $g = \AND_{|S|}$.  Hence we conclude that $\minw(f) = |S| < t$ (using the simple fact that $\Minion^0$ also does not contain any $\AND_{t'}$ for $t' > t$ as shown in \cref{obstructions atomic}).
\end{proof}

In the following subsections we proceed to give several incomparable conditions, which, together with $\AND_t \not\in \Minion^0$, are sufficient to guarantee small fixing assignments.  Some needed definitions are given below but in summary the conditions are as follows.
\begin{enumerate}
\item $\Minion$ only contains functions with small \emph{matching number} (\cref{lemma:matching} in \cref{sec:matching}).  In this case we even obtain small fixing sets.
\item $\Minion$ only contains functions with small \emph{inverted matching number} (\cref{lemma:invmatch} in \cref{sec:invmatch}).
\item $\Minion$ only contains \emph{unate} functions, and additionally $\Minion^0$ does not contain arbitrarily large $\xNOR$ functions (\cref{lemma:unate+g} in \cref{sec:unate}), where 
\[
\xNOR_{t}(x) = x_1 \wedge \NOR(x_2, \ldots, x_t) = x_1 \wedge \neg x_2 \wedge \ldots \wedge \neg x_t.
\]
Forbidding large $\xNOR$  is a natural variant of the $\AND$ condition which guarantees that a small number of additional variables can be set to $0$ to extend a low-weight $1$-assignment to a fixing assignment.
\end{enumerate}

Many predicates $A$ yield a minion that satisfies one of the three conditions above but it turns out that a ``bottleneck'' (in the sense that a large fraction of hard promises $A$ are not covered by it) is the shared condition $\AND_t \not\in \Minion^0$ for guaranteeing that $\minw(f)$ is small.  To address this, we identify a second, more complicated, explicit condition for bounding $\minw(f)$, which we refer to as $\Minion$ being $t$-ADA-free (\cref{ADA definition})  This condition can replace the $\AND$ condition in all three of the fixing assignment results mentioned above, resulting in stronger versions (Theorems \ref{match+ADA}, \ref{invmatch+ADA} and \ref{thm:unate+ADA}) of these results which significantly reduces the number of predicates $A$ that we are unable to classify.

As one last step to further sharpen these results, we give a different condition for the unate case (item 3) above, where instead of forbidding the $\xNOR$ function (which seems to be the main bottleneck after introducing ADA-freeness), we forbid certain functions that we refer to as UnCADAs and UnDADAs (\cref{thm:split} in \cref{sec:split}).  Thus we have in total four different hardness conditions and \cref{fig:hardness overview} gives a graphical overview of these.

\begin{figure}[ht]
    \centering
    \includegraphics{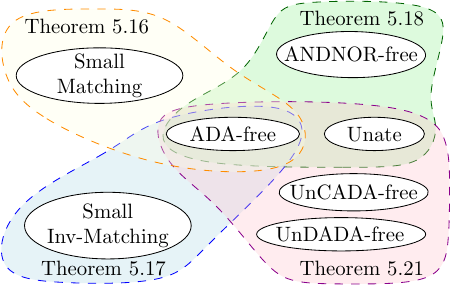}
    \caption{Overview of the conditions required for Theorems \ref{match+ADA}, \ref{invmatch+ADA}, \ref{thm:unate+ADA} and \ref{thm:split} that ensure the existence of small fixing assignments.}
    \label{fig:hardness overview}
\end{figure}

Crucially, given a concrete predicate $A \subseteq \{0,1\}^k$, checking all the various properties shown in \cref{fig:hardness overview} for $\Minion = \Pol(\fiPCSP(A, \OR))$ can be done relatively efficiently because they all boil down to (non-)existence of certain polymorphisms of small constant arity.  We discuss these computational aspects in more detail in \cref{sec:computational aspects}.

\subsection{Bounded Matchings}
\label{sec:matching}

An immediate consequence of \cref{AND assignment} is, as mentioned in the overview, that if all functions in $\Minion$ are monotone and $\AND_t \not\in \Minion^0$, then all $f \in \Minion$ have fixing sets of size at most $t-1$.  However, the monotonicity property can be relaxed, leading us to the following notion which quantifies to what extent a function can take the value $1$ on a large number of disjoint sets.

\begin{definition}
    Let $f: \{0,1\}^\ell \rightarrow \{0,1\}$.  A \emph{matching} of $f$ of size $t$ is a collection of $t$ disjoint subsets $S_1,S_2, \ldots, S_{t}$ of $[\ell]$ such that $f(S_i) = 1$ for all $i \in [t]$.
    The \emph{matching number} of $f$ is the maximum size of any matching of $f$.
\end{definition}

A folded idempotent Boolean function is monotone if and only if it has matching number $1$.  This is because two sets $S \subset T$ such that $f(S)=1$ and $f(T)=0$ implies that $f$ is $1$ on the two disjoint sets $S$ and $\overline{T}$.  Hence the notion of having bounded matching number can be viewed as a generalization of monotonicity. If the functions of a minion $\Minion$ have bounded matching number, then this rules out families such as $\Par$ and $\AT$ from $\Minion$, but it does not rule out $\Maj$.

\begin{remark}
    If a folded function $f$ has a $t$-fixing set then it has a matching number at most $t$. This follows since if $f(S)=1$ then $S$ must intersect the fixing set, otherwise the fixing set would force $f(\overline{S})=1$ contradicting that $f$ is folded.  In other words, if we are looking for small fixing sets, then bounded matching number is a necessary condition to achieve this.
\end{remark}

Having bounded matching number and forbidding $\AND$ in $\Minion^0$ is enough to conclude the existence of small fixing sets.

\begin{lemma} \label{lemma:matching}
    Let $\Minion$ be a folded idempotent minion, and suppose that there exists constants $t_1, t_2$, such that
    \begin{enumerate}
        \item $\AND_{t_1} \not\in \Minion^0$.
        \item Every $f \in \Minion$ has matching number $\leq t_2$.
    \end{enumerate}
    Then all $f \in \Minion$ have a fixing set of size at most $(t_1-1) t_2$.
\end{lemma}

\begin{proof}
     Let $f \in \Minion$ and let $S_1$ be the smallest cardinality subset such that $f(S_1) = 1$ (note that $\Minion$ being idempotent guarantees $S_1$ is non-empty). Let $S_2$ be the smallest cardinality subset disjoint from $S_1$ such that $f(S_2) = 1$. Continue this procedure until there is no set $S$ disjoint from all selected sets such that  $f(S) = 1$. Suppose the sequence we end up with is $S_1,S_2,\ldots S_m$. Since $f$ has matching number $\leq t_2$, we know that $m \leq t_2$ and since $f$ is folded, $f(S) = 1$ for any set that contains the union of the selected sets and hence $\bigcup_{i \in [m]} S_i$ is a fixing set.

     Finally by \cref{AND assignment} it follows that we can ensure $|S_i| \le t_1-1$ for each $i$ and hence the fixing set $\bigcup_{i \in [m]} S_i$ is of size at most $(t_1-1)t_2$, as desired.
\end{proof}

\subsection{Bounded Inverted Matchings}
\label{sec:invmatch}

For the condition in \cref{sec:matching} using bounded matching numbers, it is natural to consider a variant where a function cannot have many disjoint sets of variables that can flip the function from $1$ to $0$.   This leads us to the following definition.

\begin{definition}
    For a function $f: \{0,1\}^{\ell} \rightarrow \{0,1\}$ and a, possibly empty, set $S \subseteq [\ell]$ such that $f(S) = 1$, an \emph{inverted matching} of size $t$ of $f$ with respect to $S$ is a collection of $t$ disjoint subsets $T_1, \ldots, T_t$ of $[\ell]$ such that $f(S \cup T_i) = 0$ for all $i \in [t]$.  The \emph{inverted matching number} of $f$ is the maximum size of any inverted matching of $f$ with respect to any $S \in f^{-1}(1)$.
\end{definition}

To obtain small fixing assignments for functions with low inverted matching number, we also need a condition that guarantees that we can always make the sets $T_i$ involved in an inverted matching small.
With these definitions, we then have the following direct generalization of \cref{lemma:matching}.

\begin{lemma}
    \label{NAND+invmatch}
    Let $\Minion$ be a minion such that:
    \begin{enumerate}
    \item $\NAND_{t_1} \not\in \Minion^{0,1}$, and
    \item Every $f \in \Minion$ has inverted matching number $\le t_2$.
    \end{enumerate}
    Then for any $f \in \Minion$ and $S$ such that $f(S) = 1$, there is a set $T$, disjoint from $S$, of size at most $(t_1-1) t_2$ such that setting $x_S = 1$ and $x_T = 0$ is a fixing assignment of $f$. 
\end{lemma}

\begin{proof}
    The proof is analogous to that of \cref{lemma:matching}.  Create a sequence $T_1, \ldots, T_m$ of subsets
    where $T_i$ is a minimum-cardinality set of coordinates disjoint from $S$ and $T_1, \ldots, T_{i-1}$ such that $f(S \cup T_i) = 0$, stopping when no more such sets exist.  Since $f$ has bounded inverted matching number, we have $m \le t_2$.  Furthermore, it is clear that fixing $S$ to $1$ and the coordinates of $T := \bigcup_{i=1}^m T_i$ to $0$ is a fixing assignment for $f$, so it remains to bound the size $|T|$.

    Consider the function $g \in M^{0,1}$ of arity $|T_i|$ obtained from $f$ by fixing all variables of $S$ to $1$ and all variables outside $S \cup T_i$ to $0$.  By the minimality of $T_i$, the function $g$ is the $\NAND_{|T_i|}$ function and hence $|T_i| \le t_1-1$.  Thus $|T| = \sum |T_i| \le (t_1-1) \cdot t_2$, as desired.
\end{proof}

For folded minions, bounded inverted matching number implies not having $\NAND$.

\begin{claim}
    \label{folded+inverted implies NAND}
    Let $\Minion$ be a folded minion.  If all $f \in \Minion$ have inverted matching number $< t$, then $\NAND_{t} \not\in \Minion^{0,1}$.
\end{claim}
\begin{proof}
    Suppose for contradiction that $\NAND_t \in \Minion^{0,1}$.  Since $\Minion$ is folded this happens if and only if $\NOR_t \in \Minion^{0,1}$ (because in the folded setting, switching which variables are fixed to $0$ and which variables are fixed to $1$ takes a function $f$ to its dual).  In other words, if $\NAND_t \in \Minion^{0,1}$ there is a function $f \in \Minion$ of arity $t+2$ such that $f(0,1,x) = \NOR_t(x)$.  In particular $f(0,1,0^t) = 1$, and $f(0,1,e_i) = 0$ for $1 \le i \le t$ where $e_i$ is the $i$th unit vector, so $f$ has inverted matching number at least $t$.
\end{proof}

Combining \cref{AND assignment,NAND+invmatch,folded+inverted implies NAND} yields the following immediate corollary.

\begin{lemma}  \label{lemma:invmatch}
Let $\Minion$ be a folded minion such that $\AND_{t_1} \not\in \Minion^0$ and all $f \in \Minion$ have inverted matching number $\le t_2$.  Then all $f \in \Minion$ have a fixing assignment of size at most $t_1-1 + t_2^2$.
\end{lemma}

\begin{proof}
    By \cref{folded+inverted implies NAND} we know that $\NAND_{t_2+1} \notin \Minion^{0,1}$ and by \cref{AND assignment} we have $\minw(f) \leq t_1-1$ for any $f \in \Minion$.  The lemma now follows by \cref{NAND+invmatch}.
\end{proof}

Note that unlike \cref{lemma:matching}, this result does not need $\Minion$ to be idempotent.  It is also possible to refrain from using \cref{folded+inverted implies NAND} and instead get the same result with the requirement that $\Minion$ is folded replaced by $\NAND_t \not\in \Minion^{0,1}$.

\subsection{Unate Functions}
\label{sec:unate}

Let us turn to a different generalization of being monotone.  A function is {\em unate} if for each variable $x_i$, $f$ is either positive in $x_i$ or negative in $x_i$ (or both, if $f$ does not depend on $x_i$).  In other words there cannot exist an $i$ and two sets $S$, $T$ such that both $f(S) < f(S \cup \{i\})$ and $f(T) > f(T \cup \{i\})$.  Note that in a fixing assignment for a unate function, we can without loss of generality set all positive variables included to $1$, and all negative variables included to $0$.  This motivates the following strategy: first pick a small set $S$ of positive variables such that $f(S) = 1$ (as guaranteed by e.g.~\cref{AND assignment}).  Then pick a small number $T$ of negative variables, such that even if all other negative variables outside $T$ are set to $1$, $f$ still equals $1$.  Fixing $S$ to $1$ and $T$ to $0$ then yields a fixing assignment of size $|S|+|T|$.

In order to bound the size of $T$ in this plan, we can employ a similar idea as in \cref{AND assignment}.  The function to forbid is now the less natural function $\xNOR_t = x_1 \wedge \NOR(x_2, \ldots, x_t)$.  We have the following lemma.

\begin{lemma}
    \label{xNOR assignment}
    Let $\Minion$ be an idempotent and unate minion such that $\xNOR_{t} \not \in \Minion^0$ for some $t \ge 2$.  Then for every $f \in \Minion$ and every $S$ such that $f(S) = 1$, there is a set $T \subseteq [\ell] \setminus S$ of size $|T| \le t-2$ such that setting $x_S = 1$ and $x_T = 0$ is a fixing assignment of $f$ (of size $|S|+|T|$).
\end{lemma}

\begin{proof}
    Let $Y$ be the set of negative variables of $f$ (except those included in $S$, if any) and let $T \subseteq Y$ be minimal such that $f(S \cup (Y \setminus T)) = 1$ (such a $T$ exists since $T=Y$ satisfies the requirement).   Construct a function $g \in \Minion^0$ of arity $|T|+1$ by fixing all variables of $[\ell] \setminus (S \cup Y)$ to $0$, identifying all variables of $S \cup (Y \setminus T)$ into a new variable $x_0$, and keeping the variables of $T$ as $x_1, \ldots, x_{|T|}$.
    By the minimality of $T$, the function $g$ satisfies
    \begin{align*}
        g(1, x_1, \ldots, x_{|T|}) &= \NOR(x_1, \ldots, x_t)
    \end{align*}
    Furthermore, since $\Minion$ is idempotent we have $g(0,0,\ldots,0) = 0$, which together with $g$ being non-positive in $x_1, \ldots, x_{|T|}$ implies that $g(0, x_1, \ldots, x_{|T|}) = 0$.
    We thus see that $g = \xNOR_{|T|+1}$ and hence since $\xNOR_{t} \not\in \Minion^0$ (and using \cref{obstructions atomic}) we have $|T|+1 < t$.

    To see that setting $x_S = 1$ and $x_T = 0$ yields a fixing assignment, note that since $f$ is unate we have for any such $x$ that $f(x) \ge f(S \cup (Y \setminus T)) = 1$.
\end{proof}

\cref{AND assignment,xNOR assignment} yield the following immediate corollary.

\begin{lemma}  \label{lemma:unate+g}
Let $\Minion$ be an idempotent and unate minion such that $\AND_{t_1} \not\in \Minion^0$ and $\xNOR_{t_2} \not\in \Minion^0$  for some $t_1, t_2$.  Then every $f \in \Minion$ has a $(t_1+t_2-3)$-fixing assignment.
\end{lemma}

\begin{proof}
Let $f \in \Minion$.  If $f$ is identically $0$ it has an empty fixing assignment, otherwise \cref{AND assignment} yields an $S$ of size at most $t_1-1$ such that $f(S) = 1$, and then by \cref{xNOR assignment} there is a $T$ of size at most $t_2-2$ which together with $S$ forms a fixing assignment of size $|S|+|T|=t_1+t_2-3$.
\end{proof}

\subsection{Beyond \texorpdfstring{$\AND$}{AND}: Approximate Double-Ands}
\label{sec:ADA}

While not á priori clear, it turns out from experiments that the main bottleneck in the hitherto described hardness conditions is \cref{AND assignment} which guarantees that all $f \in \Minion$ have low-weight $1$-assignments whenever $\AND_t \not\in \Minion^0$.  To alleviate this we now give a weaker sufficient condition to bound $\minw (f)$.  The idealized type of function we now want to forbid are functions of the form $\AND(x_S) \vee \AND(x_T)$ for two overlapping but incomparable sets $S$ and $T$ of size $\ge t$.  This is still somewhat restrictive, and the actual functions we forbid can be viewed as approximate versions of this.  Let us give the formal definition.

\begin{definition}
\label{ADA definition}
    For positive integers $c, d$, a $(c,d)$-ADA (Approximate Double-AND) is a $(c+2d)$-ary function $f: \{0,1\}^d \times \{0,1\}^c \times \{0,1\}^d \rightarrow \{0,1\}$ such that the following holds:
    \begin{enumerate}
    \item[(a)] $f(1^d, 1^c, 0^d) = f(0^d, 1^c, 1^d) = 1$,
    \item[(b)] $f(x, y, z) = 0$ if $w(x)+w(y)+w(z) < c+d$.
    \item[(c)] $f(x, y, z) = 0$ unless at least two of $x, y, z$ are all-$1$s
    \end{enumerate}
    We say that a minion $\Minion$ is \emph{$t$-ADA-free}, $t\geq2$, if it does not contain a $(t-d,d)$-ADA for any $1 \le d \le t-1$.
\end{definition}

It is easy to see that being $t$-ADA-free is a weaker property than not having $\AND_t$.

\begin{claim} \label{claim:AND_free_ADA_free}
    Let $\Minion$ be a minion.  If $\AND_t \not\in \Minion^0$ then $\Minion^0$ is $t$-ADA-free.
\end{claim}

\begin{proof}
    To prove the contrapositive, note that if $\Minion^0$ contains a $(c,d)$-ADA $f(x,y,z)$ for some $c+d=t$ then fixing $z=0^d$ yields the $\AND_t$ function.
\end{proof}

We have the following technical lemma, which effectively generalizes \cref{AND assignment}.

\begin{lemma}
\label{ADA-free sets}
Let $\Minion$ be a minion such that $\Minion^0$ is $t$-ADA-free for some $t \ge 2$.  Then for
every $f$ in $\Minion^0$ which is not identically $0$, at least one of the following two properties holds:
\begin{enumerate}
\item $\minw(f) \le t-1$, or
\item $S \cap T = \emptyset$ for every pair of inclusion-wise minimal $S, T \in f^{-1}(1)$.
\end{enumerate}
\end{lemma}

In the next section, we show how this enables us to strengthen the hardness conditions of the previous sections that were based on \cref{AND assignment}.
But first, let us prove the technical lemma.

\begin{proof}
  We proceed by induction on the arity $\ell$ of $f$.  The base cases $\ell \le t$ clearly holds -- either $\minw(f) \le t-1$, or the only set in
  $f^{-1}(1)$ is $[\ell]$ itself.
  
  For the inductive step consider a function $f: \{0,1\}^{\ell} \rightarrow \{0,1\}$ and suppose the
  lemma is true for all $\ell' < \ell$.
  Suppose for contradiction that the lemma fails for $f$, i.e.,
  \begin{enumerate}
  \item $\minw(f) \ge t$, \emph{and}
  \item There exists inclusion-wise minimal $S, T \in f^{-1}(1)$ such that $S \cap T \ne \emptyset$.
  \end{enumerate}

  Note that we must have $S \cup T = [\ell]$, since otherwise fixing a
  variable of $f$ outside $S \cup T$ to $0$ yields a contradiction to
  the inductive hypothesis.
  Among all possible ways of choosing the pair $(S, T)$, choose one which
  minimizes $(|S|, |T|)$ (i.e., first minimize $|S|$, then subject to
  this minimize $|T|$).  We then have the following claim.

  \begin{claim}
    $|S| = |T| = t$.
  \end{claim}

  \begin{proof}
    Let us start with $|S|$ and suppose for contradiction that $|S| >
    t \ge 2$.  Let $i, j \in S$ be such that either both are in $T$ or
    neither are in $T$ (since $|S| \ge 3$ such $i,j$ must exist).
    Consider the minor $g$ of $f$ obtained by identifying $x_i$ and
    $x_j$, and the sets of coordinates $S'$ and $T'$ obtained from $S$
    and $T$ by identifying these two coordinates; note that
    $|S'|=|S|-1$, while $|T'|$ is either $|T|-1$ or $|T|$ depending on
    whether the identified coordinates $i,j$ are in $T$ or not.

    Since $S'$ and $T'$ are inclusion-wise minimal sets in
    $g^{-1}(1)$ that intersect, the inductive hypothesis implies that
    $\minw(g) = t-1$.  It follows that there exists a set
    $X \subseteq [\ell]$ of size $|X| = t < |S|$, containing $i,j$, such
    that $f(X) = 1$.  Since $S$ is inclusion-wise minimal we cannot have $X \subseteq S$.  But then $(X, S)$ would have been a valid choice
    of the sets $(S, T)$ as they intersect at $i$.  Since $S$ was chosen with minimum possible cardinality, this contradicts the assumption that $|S| > t$.

    Having established $|S|=t$, suppose for contradiction that $|T| >
    t$.  This means that $|T \setminus S| \ge 2$ (since $S$ cannot be
    contained in $T$).  Repeating the argument above with two elements in $T \setminus S$ we establish the existence of an $X \subseteq [\ell]$ of size $|X|=t$
    and intersecting $T \setminus S$ with $f(X) = 1$.  By
    inclusion-wise minimality of $T$, $X$ must intersect $[\ell]
    \setminus T = S \setminus T$ and hence $(S, X)$ would have been a
    valid choice of the sets $(S, T)$.
  \end{proof}

  Let $c = |S \cap T|$ and $d = |S \setminus T| = |T \setminus S| = t-c$.
  By reordering the $\ell = |S \cup T| = c+2d$ variables of $f$ we can write $f$ as a function $f(x, y, z)$ on $\{0,1\}^d
  \times \{0,1\}^c \times \{0,1\}^d$, with $x$ being the variables
  from $S \setminus T$, $z$ the variables from $T \setminus
  S$, and $y$ the variables from $S \cap T$.
  We claim that with this ordering of the variables, $f$ is a $(c,d)$-ADA.  To wit, let us verify the properties:
    \begin{enumerate}
    \item[(a)] $f(1^d, 1^c, 0^d) = f(0^d, 1^c, 1^d) = 1$.  This is clear since the first value is $f(S)$ and the second is $f(T)$.
    
    \item[(b)] $f(x, y, z) = 0$ if $w(x)+w(y)+w(z) < c+d$.  This follows since $\minw(f) \ge t = c+d$.
    
    \item[(c)] $f(x, y, z) = 0$ unless at least two of $x, y, z$ are all-$1$s.
        Suppose for contradiction that this is not true and take a counterexample $(x, y, z)$ of minimum weight $w(x)+w(y)+w(z)$.  Assume without loss of generality that $x \ne 1^d$ (otherwise, we can swap the roles of $S$ and $T$ which swaps $x$ and $z$).  We then have $(y, z) \ne (1^c, 1^d)$, and in particular the subset $\tilde{S} \subseteq [\ell]$ of $1$-coordinates of $(x, y, z)$ satisfies the following two properties:
        \begin{itemize}
            \item $\tilde{S} \not\supseteq S$ (since $x \ne 1^d$) and $\tilde{S} \not\supseteq T$ (since $(y,z) \ne (1^c, 1^d)$).
            \item $\tilde{S}$ is inclusion-wise minimal in $f^{-1}(1)$ (since $(x, y, z)$ were chosen with minimum possible weight)
        \end{itemize}
        Now consider the function $\tilde{f} \in \Minion^0$ of arity $\ell-1$ obtained by fixing some coordinate in $S \setminus \tilde{S}$ to $0$.  Then $(\tilde{S}, T)$ is a pair of inclusion-wise minimal but intersecting sets in $\tilde{f}^{-1}(1)$ and hence $\minw(\tilde{f}) \le t-1$ by the induction hypothesis, but then $\minw(f) \le \minw(\tilde{f}) \le t-1$. 

    \end{enumerate}

  Thus we have now established that $f \in \Minion^0$ is a
  $(t-d, d)$-ADA, which contradicts the fact that $\Minion^0$
  is $t$-ADA-free, and the inductive step follows.
\end{proof}

A useful corollary of the technical lemma is the following.

\begin{corollary}
\label{ADA assignment}
    Let $\Minion$ be a folded minion such that $\Minion^0$ is $t$-ADA-free.  Then for every $f \in \Minion$ it holds that $\minw(f) \le t-1$.
\end{corollary}

\begin{proof}
    Proof by induction on the arity $\ell$ of $f$.  As base case we have $\ell < 2t-1$ for which the result is trivially true since every folded function $f$ of arity $\ell$ has $\minw(f) \le \lceil \ell/2 \rceil$.

    For the inductive case assume the claim is true for $\ell-1$ and assume for contradiction that it is false for some $f$ of arity $\ell \ge 2t-1 > t$.   For two arbitrary coordinates $i,j \in [\ell]$, consider the minor $g$ of $f$ formed by identifying $i$ and $j$.  By the induction hypothesis, there is a set $S'$ of size $\le t-1$ for $g$ such that $g(S')=1$.  Since $f$ does not have such a set, $S'$ must include the new coordinate.  Replacing the new coordinate by $i$ and $j$ we get a set $S \subseteq [\ell]$ of size $t$ such that $f(S) = 1$ and $i,j \in S$.

    Let $i' \in S$ and $j' \not\in S$, and repeat this process to obtain a set $T \subseteq [\ell]$ of size $t$ such that $f(T) = 1$ and $i',j' \in S$.  This gives two intersecting sets contradicting \cref{ADA-free sets}.
\end{proof}

\subsection{Strengthened Hardness Conditions}

Using the notion of $t$-ADA-free minions we can strengthen the previous conditions for small fixing assignments.

\begin{theorem}[Strengthening of \cref{lemma:matching}]
    \label{match+ADA}
    Let $\Minion$ be a folded idempotent minion such that $\Minion^0$ is $t_1$-ADA-free and every $f \in \Minion$ has matching number at most $t_2$.  Then all $f \in \Minion$ have a fixing set of size at most $(t_1-1) \cdot t_2$.
\end{theorem}

\begin{theorem}[Strengthening of \cref{lemma:invmatch}]
    \label{invmatch+ADA}
    Let $\Minion$ be a folded minion such that $\Minion^0$ is $t_1$-ADA-free and all $f \in \Minion$ have inverted matching number $\le t_2$.  Then all
    $f \in \Minion$ have a fixing assignment of size at most $t_1-1 + t_2^2$.
\end{theorem}

\begin{theorem}[Strengthening of \cref{lemma:unate+g}]
    \label{thm:unate+ADA}
    Let $\Minion$ be a folded, idempotent, and unate minion such that $\Minion^0$ is $t_1$-ADA-free and $\xNOR_{t_2} \not \in \Minion^0$.  Then every $f$ in $\Minion$ has a $(t_1+t_2-3)$-fixing assignment.
\end{theorem}

\cref{invmatch+ADA,thm:unate+ADA} follow immediately from \cref{ADA assignment} combined with the previous results.  For \cref{invmatch+ADA} we use \cref{NAND+invmatch} and \cref{folded+inverted implies NAND} while for \cref{thm:unate+ADA} we use \cref{xNOR assignment}.

\cref{match+ADA} is not quite as immediate, because the proof of the original condition \cref{lemma:matching} uses the conclusion of \cref{AND assignment} that all $f \in \Minion^0$ have small $\minw(f)$ (as opposed to the other two results which only use the property that all $f \in \Minion$ have small $\minw(f)$), and this is not guaranteed by \cref{ADA-free sets} or \cref{ADA assignment}.  Nevertheless, we can extend the proof of \cref{lemma:matching} to obtain the desired result.

\begin{proof}[Proof of \cref{match+ADA}]
  We begin as in the proof \cref{lemma:matching}, picking a sequence
  of disjoint $S_1, \ldots, S_m$ such that each $f(S_i)=1$ and $S_i$
  is of minimum possible cardinality subject to being disjoint from
  $S_1 \cup \ldots \cup S_{i-1}$ (the assumption that $f$ is idempotent guarantees $S_i \ne \emptyset$).  The union of these forms a fixing
  set and $m \le t_2$.  If all $|S_i| \le t_1-1$ then we are done.

  Otherwise, let $i$ be the first index such that $|S_i| \ge t_1$, and
  consider the function $g \in \Minion^0$ obtained by fixing the
  coordinates of $S_1, \ldots, S_{i-1}$ to $0$.  Since $\minw(g) =
  |S_i| \ge t_1$, we have by \cref{ADA-free sets} that all
  inclusion-wise minimal sets in $g^{-1}(1)$ are pairwise disjoint.
  This implies that $S_i, \ldots, S_m$ are the only inclusion-wise minimal
  sets in $g^{-1}(1)$ and in particular 
  any $T \in g^{-1}(1)$ contains $S_j$ for some $j
  \in \{i, \ldots, m\}$.  Thus picking one coordinate each
  from the sets $S_i, \ldots, S_m$ yields a hitting set for $g^{-1}(1)$
  and setting these together with $S_1, \ldots, S_{i-1}$ to $0$ yields
  a fixing assignment forcing $f(x)=0$ of size at most $(i-1)(t_1-1) + m-i+1 \le (t_1-1) \cdot
  m \le (t_1-1) \cdot t_2$.  As $f$ is folded, the variables set to 0 is  fixing set.
\end{proof}

\subsection{Beyond \texorpdfstring{$\xNOR$}{xNOR}: Unate Controlled ADAs}
\label{sec:split}

In this section we focus on unate minions and describe another sufficient condition for such minions to have small fixing assignments.  We say that a function $f$ is $(p,q)$-unate if it can be written as $f: \{0,1\}^p \times \{0,1\}^q \rightarrow \{0,1\}$ where $f(x,y)$ is non-negative in $x$, and non-positive in $y$ (if $f$ does not depend on some variable it can be placed in either group, so it is possible for a function to be both $(p,q)$-unate and $(p+1,q-1)$-unate).  

For fixing assignments of unate functions, we say that $S \subseteq [p], T \subseteq [q]$ of sizes $|S|=s$ and $|T|=t$ is an $(s,t)$-fixing assignment for a $(p,q)$-unate function $f: \{0,1\}^p \times \{0,1\}^q \rightarrow \{0,1\}$ if $f(S, \overline{T}) = 1$ (i.e., setting all non-negative variables of $S$ to $1$, and all non-positive variables of $T$ to $0$, fixes $f$ to $1$).  Our general goal in this section is to identify finite conditions which are sufficient to guarantee that all functions in a unate minion have $(t-1,1)$-fixing assignments for some $t$, i.e., fixing assignments of size $t$ which sets one variable to $0$ and $t-1$ variables to $1$.

Similar to the approach in \cref{sec:ADA} and the notion of $t$-ADA-free minions, the condition is based on forbidding certain functions that look similar to a disjunction of two overlapping ANDs, but in this case each of the ANDs is ``controlled'' by some negative inputs and we generally refer to them as Controlled ADAs.  There are (unfortunately) two very similar but subtly different types of functions.  Let us define the first (main) type.

\begin{definition}
\label{def:split}
    A \emph{$(c,d)$-UnCADA} (Unate Controlled Approximate Double-AND) is a $(c+2d+1,3)$-unate folded and idempotent function $f: \{0,1\}^{c+2d+1} \times \{0,1\}^3 \rightarrow \{0,1\}$ such that
    \begin{enumerate}
    \item[(a)] $f(1^d1^c0^d0, 011) = f(0^d1^c1^d0, 101) = 1$
    \item[(b)] for every $x \in \{0,1\}^{c+2d}$ and $y \in \{0,1\}^2$ such that $w(x) \le c+d-1$ and $w(y) 
    \ge 1$ it holds that $f(x0, y1) = 0$.
    \end{enumerate}
    A minion $\Minion$ is \emph{$t$-UnCADA-free}, $t\geq2$, if it does not contain a $(t-d, d)$-UnCADA for any $1 \le d \le t-1$.
\end{definition}

It may be instructive to compare \cref{def:split} with \cref{ADA definition} of $(c,d)$-ADAs, since the definitions are very similar.  The conditions (a) in both definitions are essentially the same, saying that $f$ has two $1$-assignments with a certain overlap pattern determined by $c$ and $d$.  Likewise the conditions (b) are analogous, with the main difference stemming from \cref{def:split} requiring that $f$ is positive in the $x$-variables and negative in the $y$-variables (unlike an ADA, which has no such requirements).  Another minor difference is the presence of the last positive variable in \cref{def:split}, which is $0$ in all prescribed function values.  This takes a similar role as $M^0$ does in previous arguments -- rather than dropping this variable and using $\Minion^0$ being UnCADA-free, we explicitly include it in the definition and operate directly on the original $\Minion$.  The reason for this is to be able to require $f$ being positive in this variable.

The second type of function has an easier definition.

\begin{definition}\label{def:UnDADA}
    For $t \ge 3$, a \emph{$t$-UnDADA} (Unate Double-controlled Approximate Double-AND) is a $(t,4)$-unate folded and idempotent function $f: \{0,1\}^t \times \{0,1\}^4 \rightarrow \{0,1\}$ such that
    \begin{enumerate}
        \item[(a)] $f(1^{t-1}0, 0011) = f(01^{t-1}, 1001) = 1$
        \item[(b)] for every $x \in \{0,1\}^t$ and $y \in \{0,1\}^3$ such that $w(x) \le t-1$ and $w(y) \ge 2$, it holds that $f(x, y1) = 0$
    \end{enumerate}
\end{definition}

Thus in a UnDADA, two negative variables (``control bits'') need to be set to $0$ to activate each of the ``ANDs'', as opposed to a single negative variable in an UnCADA.  One important but subtle difference in the definition is that an UnCADA has an extra positive variable which is $0$ in all prescribed function values, but the UnDADA does not have this.  While this may seem like a minor detail, it turns out to be very important -- there are many minions of interest which do not have $t$-UnDADAs, but that would have them if we were to add the additional $0$-variable in the UnDADA definition. 

We are finally ready to state our fixing assignment condition.

\begin{theorem}
\label{thm:split}
    Let $\Minion$ be a folded, idempotent, and unate minion which is $t$-ADA-free, $t$-UnCADA-free, and does not have a $t$-UnDADA.  Then every $f \in \Minion$ has a $(t-1,1)$-fixing assignment.
\end{theorem}

Let us first prove the following preparatory general claim, which does
not require any specific properties of $\Minion$ but captures how the
inductive hypothesis is used in the proof.

\begin{claim}
  \label{split inductive step}
  Let $\Minion$ be a unate minion such that all $(p-1,q)$-unate $f \in \Minion$ and all
  $(p,q-1)$-unate $f \in \Minion$ have a $(t-1,1)$-fixing
  assignment.  Then for any $(p,q)$-unate $f \in \Minion$, either
  $f$ has a $(t-1,1)$-fixing assignment, or the following conditions all hold:
  \begin{enumerate}
  \item[(a)] For any $i \in [p]$ and $j \in [q]$, $f$ has a $(t,1)$-fixing assignment $(S, T)$ such that $i \in
    S$ and $j \not\in T$.
  \item[(b)] For any $i \in [p]$ and $j \in [q]$, $f$ has a $(t-1,2)$-fixing assignment $(S,T)$ such that $i \not\in
    S$ and $j \in T$
  \item[(c)] For any $j, j' \in [q]$, $f$ has a $(t-1,2)$-fixing assignment $(S, \{j,j'\})$.
  \end{enumerate} 
\end{claim}

\begin{proof}
  For (a) and (b), identify the positive variable $i \in [p]$ and the
  non-positive $j \in [q]$ to obtain a minor $g$.  The function $g$ is
  either a $(p-1,q)$-unate or $(p,q-1)$-unate function, so by the
  assumption of the lemma, $g$ has a $(t-1,1)$-fixing assignment
  $(S',T')$.  Note that the newly created variable cannot be included
  in the fixing assignment -- if it was, it would either correspond to
  a $(t-1,2)$-fixing assignment which sets the negative variable $j$ to
  $1$, or a $(t,1)$-fixing assignment which sets the positive
  variable $i$ to $0$, and in either case one variable can be dropped
  to obtain a $(t-1,1)$-fixing assignment for $f$.

  This implies that $(S' \cup \{i\}, T')$ and $(S', T' \cup \{j\})$
  are fixing assignments of $f$, establishing items (a) and (b).

  The last item (c) follows in a similar manner by identifying the two
  negative variables $j$ and $j'$, obtaining a $(p,q-1)$-unate
  function.  The resulting minor $g$ must now have a $(t-1,1)$-fixing
  where the non-positive variable used is the newly identified
  variable.
\end{proof}

\begin{proof}[Proof of \cref{thm:split}]
  Consider a general $(p,q)$-unate function $f \in \Minion$.  We
  do double induction on $p$ and $q$.  
  The base cases $p \le t-1$ are trivial as $\Minion$ is idenpotent.  For the inductive step, we have to treat the first step $p=t$ differently, so we divide into two cases $p=t$ and $p > t$.

  \paragraph{Case 1: $p=t$}

  The case $p=t$ is special and requires a bit of extra care (in fact
  it may be more instructive to read the general case $p > t$ below
  first).  In this case the base cases $q \le 2$ follow by $f$ being
  folded -- any folded $(p,q)$-unate function has a $(\lceil p/2
  \rceil, \lceil q/2 \rceil)$-fixing assignment.  So we may assume $q
  \ge 3$.

  For $i \in [p]$ and $j,j' \in [q]$, let us say that the triple $(i,
  j, j')$ is fixing if $([p]-i, \{j,j'\})$ is a $(t-1,2)$-fixing
  assignment of $f$.
  By \cref{split inductive step} we have the following two properties:
  \begin{itemize}
  \item For every $i \in [p]$ and $j \in [q]$, there is a $j' \in [q]$ such that $(i, j, j')$ is fixing.
  \item For every $j, j' \in [q]$, there is an $i \in [p]$ such that $(i, j, j')$ is fixing.
  \end{itemize}

  These two properties imply that there are two fixing triples $(i_1,
  j_1, j')$ and $(i_2, j_2, j')$ with $i_2 \ne i_1$ and $j_2 \ne j_1$.
  (To see this, note that if this was not the case then the second
  property would imply that there is a unique $i$ such that for every
  $j,j'$ only $(i,j,j')$ is fixing, but this contradicts the first
  property.)
  
  Let $S_1 = [p]-i_1$, $T_1 = \{j_1, j'\}$, $S_2 = [p]-i_2$, and $T_2
  = \{j_2, j'\}$ -- these are now two $(t-1,2)$-fixing assignments
  such that $S_1 \ne S_2$ but not disjoint, and $T_1 \ne T_2$ but not
  disjoint.
  Create a minor $g(x, y)$ of $f$ on
  $\{0,1\}^{t} \times \{0,1\}^{4}$ as follows:
  \begin{itemize}
  \item $x$ are the $t$ coordinates of $[p]$, with $i_1$ as $x_1$ and $i_2$ as $x_t$. 
  \item $y_1,y_2,y_3$ are the $3$ coordinates $j_1,j',j_2$ of $T_1 \cup T_2$.
  \item $y_4$ is the remaining $q-3$ coordinates of $[q]$, identified to a single variable.
  \end{itemize}
  
  We can now check that $g$ is in fact a $t$-UnDADA.  It is folded and idempotent since $f$ is.  Let us verify the other properties.
  \begin{enumerate}
  \item[(a)]
  $g(1^{t-1}0, 0011) = f(S_1, \overline{T_1}) = 1$\\
  $g(01^{t-1}, 1001) = f(S_t, \overline{T_2}) = 1$
  \item[(b)] If $w(x) \le t-1$ and $w(y) \ge 2$ then $g(x, y1) = 0$ since $f$ does not have a $(t-1,1)$-fixing assignment.
  \end{enumerate}

  This contradicts the assumption that $\Minion$ does not have a  $t$-UnDADA, so
  $f$ must have a $(t-1,1)$-fixing assignment.

  \paragraph{Case 2: $p > t$.}

  When $p > t$, the base cases $q \in \{0,1\}$ follow by $\Minion$ being $t$-ADA-free and \cref{ADA assignment},
  so we may assume $q \ge 2$.  By \cref{split inductive step} and the
  inductive hypothesis, $f$ has a $(t,1)$-fixing assignment
  $(S_1,T_1 = \{j_1\})$.  Applying \cref{split inductive step} again to some $i
  \not\in S_1$ (such $i$ exists since $p \ge t+1$) and $j_1$, we
  see that $f$ also has a $(t,1)$-fixing assignment $(S_2, T_2 = \{j_2\})$ such that $i
  \in S_2$ and $T_2 = \{j_2\} \ne \{j_1\} = T_1$.  In particular $S_1 \ne
  S_2$ and $T_1 \ne T_2$.
  
  Let $c = |S_1 \cap S_2|$ and $d = |S_1 \setminus S_2| = |S_2 \setminus S_1|$.  Create a minor $g(x, y)$ of $f$ on
  $\{0,1\}^{c+2d+1} \times \{0,1\}^{3}$ as follows:
  \begin{itemize}
  \item $x_1 \ldots x_d$ are the coordinates of $S_1 \setminus S_2$. 
  \item $x_{1+d} \ldots x_{c+d}$ are the coordinates of $S_1 \cap S_2$.
  \item $x_{c+d+1} \ldots x_{c+2d}$ are the coordinates of $S_2 \setminus S_1$.
  \item $x_{c+2d+1}$ are the remaining positive coordinates, $[p] \setminus (S_1 \cup S_2)$, identified to a single variable.
  \item $y_1$ is the coordinate $T_1$.
  \item $y_2$ is the coordinate $T_2$.
  \item $y_3$ are the remaining negative coordinates $[q] \setminus (T_1 \cup T_2)$, identified to a single variable.
  \end{itemize}
  Some of these (in particular the variables identified into $x_{c+2d+1}$ and $y_3$) might be empty, in which case
  $g$ simply does not depend on that variable.
  
  We can now check that $g$ is in fact a $(c, d)$-UnCADA.  It is idempotent and folded since $f$ is.  Let us verify the other properties.
  \begin{enumerate}
  \item[(a)] $g(1^d1^c0^d0, 011) = f(S_1, \overline{T_1}) = 1$\\
  $g(0^d1^c1^d0, 101) = f(S_2, \overline{T_2}) = 1$
  \item[(b)] For any $x \in \{0,1\}^{c+2d}$ and $y \in \{0,1\}^2$ such that $w(x) \le c+d-1=t-1$ and $w(y) \ge 1$, it holds that $g(x0, y1)=0$ since $f$ does not have a $(t-1,1)$-fixing assignment.
  \end{enumerate}

  This contradicts the assumption that $\Minion$ is $t$-UnCADA-free, so
  $f$ must have a $(t-1,1)$-fixing assignment.  This concludes the $p > t$ case of the proof, and finishes the overall proof.
\end{proof}

\begin{example}
\label{example:UnDADA}
    The main reason for the introduction of the hardness condition given by \cref{thm:split} that uses UnCADA and UnDADA is to capture the hardness of $\fiPCSP(\{\bstr{0011}, \bstr{0101}, \bstr{0110}, \bstr{1000}, \bstr{1001}\}, \OR)$. This predicate is the lone example for $k=4$ of a $\PCSP$ whose hardness is not captured by any of the other hardness conditions (Theorems \ref{match+ADA}, \ref{invmatch+ADA} and \ref{thm:unate+ADA}). The inspiration for \cref{thm:split} came directly from studying the polymorphisms of this $\PCSP$.  
\end{example}

\section{Results for Small Arities} \label{sec:small}

In this section we use the tools developed to classify predicates of small arities.  It turns out that we are able to completely characterize all predicates of arity $k \le 4$.  For arity $k=5$, we are able to classify the vast majority of predicates but there is a small fraction ($\approx 10^{-5}$) of predicates left unclassified by the conditions in \cref{sec:hard}.   We do not study the case of $k \geq 6$ due to the large number of different predicates which makes them difficult to study individually even by computer.  

There are several algorithmic aspects to consider when applying our algorithm and hardness conditions to a given Promise-SAT problem.  These are discussed in some detail in \cref{sec:computational aspects}.

\subsection{Results for Promise-SAT (with idempotence)}
\label{sec:small arities fiPCSP}

Let us start by discussing the complexity of $\fiPCSP(A, \OR)$ for all predicates of arity $k \le 5$.  

\subsubsection{Equivalence Classes}
\label{sec:equivalence classes}

Since the $\OR$ predicate is symmetric, taking a predicate $A$ and permuting the $k$ input bits has no effect on the complexity of $\fiPCSP(A, \OR)$.  Thus any two predicates which differ only by such a permutation are considered equivalent, and we only consider one predicate from each equivalence class.  Thus while the total number of predicates of arity $k$ equals $2^{2^k-1}-1$ (one $-1$ comes from $A$ not accepting the all-$0$ string, the other from $A$ being non-empty), the total number of equivalence classes is approximately $\frac{2^{2^k-1}}{k!}$ (but a bit larger because some equivalence classes are smaller than $k!$).

The representative of each equivalence class is chosen as follows. Each $A$ is naturally identified by the integer $r(A) = \sum_{x \in A} 2^{b(x)}$, where $b(x)$ means that we interpret $x \in \{0,1\}^k$ as an integer in base $2$.  E.g., the predicate $A = \{\bstr{001}, \bstr{011}\}$ on three bits would be identified by $r(A) = 2^1 + 2^3 = 10$.  For each equivalence class of predicates we choose the member for which this identifier is minimized as the representative of the class.
For example, in the case of $k=2$, we have $5$ different equivalence classes.  The representatives of these are $\{\bstr{01}\}, \{\bstr{01}, \bstr{10}\}, \{\bstr{11}\}, \{\bstr{01}, \bstr{11}\}, \{\bstr{01}, \bstr{10}, \bstr{11}\}$.

\subsubsection{Summary}

It turns out that for $k \leq 4$, we are able to determine the complexity of $\fiPCSP(A, \OR)$ for all $A$ while for arity $k=5$,  we are able to classify all except $189$ out of the over $18.6$ million predicates.   A summary is given in \cref{tab:overview}.

\begin{table}[h]
    \caption{Summary of classification of complexity of $\fiPCSP(A, \OR)$ for $A$ of arity up to $5$.\label{tab:overview}}
    \centering
    \begin{tblr}{
      colspec = {l|rrrr},
      rowhead = 1,
      row{odd} = {blue9},
      row{1} = {gray9},
      columns = {colsep=4pt},
    } 
         & Total & Tractable & NP-hard & Unknown \\
    \hline
        $k=2$ & $5$ & $5$ & $0$ & $0$ \\
        $k=3$ & $39$ & $33$ & $6$ & $0$ \\
        $k=4$ & $1\,991$ & $956$ & $1\,035$ & $0$ \\
        $k=5$ & $18\,666\,623$ & $1\,290\,862$ & $17\,375\,572$ & $189$ \\
    \end{tblr}
\end{table}

For the tractable predicates, it is interesting to analyze which families of block-symmetric polymorphisms are admitted by $\fiPCSP(A, \OR)$.   As discussed in \cref{sec:tract} we have only seen five families of polymorphisms that enable us to apply the BLP+AIP algorithm.
The relative frequencies of the five families are given in \cref{tab:tractable_overview}.  It is interesting to note that majority is, by a huge margin, the most useful polymorphism.

\begin{table}[h]
    \caption{Overview of sources of tractability for $\fiPCSP(A, \OR)$.  The number after the slash is the count of predicates that admit the corresponding family of polymorphisms. The number before the slash is the count of predicates that \emph{only} admit this family and none of the other four families.\label{tab:tractable_overview}}
    \centering
    \begin{tblr}{
      colspec = {l|rrrrr},
      rowhead = 1,
      row{odd} = {blue9},
      row{1} = {gray9},
      columns = {colsep=4pt},
    } 
         & $\MajFam$ & $\ParFam$ & $\ATFam$ & $\idneg{\MajFam}$ & $\idneg{\ParFam}$ \\
    \hline
        $k=2$ & $1/5$ & $0/4$ & $0/4$ & $0/4$ & $0/4$ \\
        $k=3$ & $13/31$ & $1/19$ & $0/18$ & $0/17$ & $0/17$ \\
        $k=4$ & $720/915$ & $31/219$ & $0/172$ & $1/163$ & $0/162$ \\
        $k=5$ & $1\,267\,621/1\,282\,927$ & $7\,259/21\,557$ & $15/11\,485$ & $260/11\,970$ & $11/11\,182$ \\
    \end{tblr}
\end{table}

For the hardness results, we have several different conditions which guarantee small fixing assignments.  The frequencies at which they apply are shown in \cref{tab:hard_overview}.  It may seem surprising that all four conditions give essentially the same set of hard predicates, but this is because most predicates are hard for many reasons.  A perhaps more fair picture is given by only looking at the minimal hard $A$'s.  Note also that \cref{thm:unate+ADA} (for the unate case) does not yield any hardness results that are not also given by one of the three other results.  This is not surprising as the more complicated \cref{thm:split} was designed to provide stronger hardness in the unate case.  However we are not aware of a proof that \cref{thm:split} covers all results implied by \cref{thm:unate+ADA} so we include also the latter in the overview.

\begin{table}[h]
    \caption{An overview of all known NP-hard predicates for each $k$. The table counts how many different predicates satisfy each category. The second number of each pair is the number of predicates that satisfy this condition while the first only counts the number of predicates that satisfy no other condition.\label{tab:hard_overview}}
    \centering
    \begin{tblr}{
     colspec = {l|rrrr},
      rowhead = 1,
      row{odd} = {blue9},
      row{1} = {gray9},
      columns = {colsep=4pt},
    } 
         & \cref{match+ADA} & \cref{invmatch+ADA} & \cref{thm:unate+ADA} & \cref{thm:split} \\
    \hline
        $k=2$ & $0/0$ & $0/0$ & $0/0$ & $0/0$ \\
        $k=3$ & $0/6$ & $0/6$ & $0/6$ & $0/6$ \\
        $k=4$ & $2/1\,029$ & $0/1\,031$ & $0/1\,030$ & $1/1\,032$ \\
        $k=5$ & $409/17\,368\,311$ & $93/17\,373\,401$ & $0/17\,371\,388$ & $687/17\,355\,043$ \\
        \end{tblr}
\end{table}

\subsubsection{Detailed results}
\label{sec:fiPCSP detailed}

Let us now provide a more detailed inspection of the tractable and NP-hard predicates.  For each arity up to $4$, we describe the minimal NP-hard, and maximal tractable predicates.  Minimal/maximal here means that if we remove/add one satisfying assignment from/to that predicate, then the complexity for the predicate changes.  By the monotonicity of $\fiPCSP(A, \OR)$ (\cref{pcsp monotonicity fact}), the complexity of any other predicate can be derived from these two lists of minimal NP-hard and maximal tractable predicates.

For the tractable predicates, we indicate which families of polymorphisms enable applying the BLP+AIP algorithm for this predicate, and for the NP-hard predicates we indicate which of the different hardness conditions yields the hardness result.

\paragraph{Arity $2$.}
In the case of $k=2$, the situation is very simple as all $5$ predicates are tractable -- they can be solved by $\twoSAT$.  In polymorphism terminology they all allow the $\MajFam$ family.

\paragraph{Arity $3$.}

For $k=3$ there are a total of $39$ non-equivalent predicates of which $33$ are tractable and $6$ are NP-hard.  There are two maximal easy predicates and these are given in \cref{tab:3hardest_easy}. A \texttt{*} indicates that this coordinate can take any value.   It is easy to see that the first predicate is parity and the second $\twoSAT$ in the first two variables.  There is a single minimal NP-hard predicate, given in \cref{tab:3easiest_hard}.  It is ``$1$-in-$3$-SAT'' augmented with a single assignment of weight two.

% \vspace{-0.3cm} % SODA formatting fix

\begin{longtblr}[
  caption = {A list of the two maximal tractable predicates for $k=3$. The second number in the last column indicates the number of tractable predicates that are subsets of this predicate.  The first number counts the number of such predicates not subset of any other maximal tractable predicate.},
  label = {tab:3hardest_easy},
]{
  colspec = {l@{}l|cccccr},
  rowhead = 1,
  row{odd} = {blue9},
  row{1} = {gray9},
  columns = {colsep=4pt},
} 
     & Predicate & $\MajFam$ & $\ParFam$ & $\ATFam$ & $\idneg{\MajFam}$ & $\idneg{\ParFam}$ & Dep. \\
\hline
    1. & $\{\bstr{001}, \bstr{010}, \bstr{100}, \bstr{111}\}$ &  & \gcmark &  &  &  & $2/7$ \\
    2. & $\{\bstr{*01}, \bstr{*10}, \bstr{*11}\}$ & \gcmark &  &  &  &  & $26/31$ \\
\end{longtblr}

% \vspace{-0.8cm} % SODA formatting fix

\begin{longtblr}[
  caption = {The only minimal NP-hard predicate for $k=3$.},
  label = {tab:3easiest_hard},
]{
  colspec = {l|cccc},
  rowhead = 1,
  row{odd} = {blue9},
  row{1} = {gray9},
  columns = {colsep=4pt},
} 
    Predicate & $\ref{match+ADA}$ & $\ref{invmatch+ADA}$ & $\ref{thm:unate+ADA}$ & $\ref{thm:split}$ \\
\hline
    $\{\bstr{001}, \bstr{010}, \bstr{011}, \bstr{100}\}$ & \gcmark & \gcmark & \gcmark & \gcmark \\
\end{longtblr}

\paragraph{Arity $4$.}
For $k=4$ there are a total of $1991$ predicates of which $956$ are tractable and $1035$ are NP-hard. There are six maximal tractable predicates, given in \cref{tab:4hardest_easy}.   Easy to recognize predicates are parity of all variables (predicate number 3) and parity of the first three variables (number 4).  Our old friend $\twoSAT$ is present as number 5 while the sixth predicate is (non-strict) majority of all four variables.  The first two predicates are less standard.  There are $14$ minimal NP-hard predicates of arity $4$, given in \cref{tab:4easiest_hard}. These are more difficult to name and we indicate what properties we can establish of the polymorphisms.

% \vspace{-0.3cm} % SODA formatting fix

\begin{longtblr}[
  caption = {A list of the 6 maximal tractable predicates for $k=4$. The last column is as in Table~\ref{tab:3hardest_easy} and the $*$s indicate that this position is free to take any value.},
  label = {tab:4hardest_easy},
]{
  colspec = {l@{}l|cccccr},
  rowhead = 1,
  row{odd} = {blue9},
  row{1} = {gray9},
  columns = {colsep=4pt},
} 
     & Predicate & $\MajFam$ & $\ParFam$ & $\ATFam$ & $\idneg{\MajFam}$ & $\idneg{\ParFam}$ & Dep. \\
\hline
    1. & $\{\bstr{0011}, \bstr{0101}, \bstr{0110}, \bstr{1000}\}$ &  &  & \gcmark &  & \gcmark & $1/7$ \\
    2. & $\{\bstr{0011}, \bstr{0100}, \bstr{0110}, \bstr{1000}, \bstr{1001}\}$ &  &  &  & \gcmark &  & $1/17$ \\
    3. & \begin{tabular}{@{}l@{}l@{}l@{}l@{}l@{}l@{}l}$\{ $&$\bstr{0001},\,$&$ \bstr{0010},\,$&$ \bstr{0100},\,$&$ \bstr{0111},\,$&$ \bstr{1000},\,$&$ \bstr{1011}, $ \\
 &$ $&$  $&$  $&$  $&$ \bstr{1101},\,$&$ \bstr{1110}\}$\end{tabular} &  & \gcmark &  &  &  & $11/34$ \\
    4. & $\{\bstr{*001}, \bstr{*010}, \bstr{*100}, \bstr{*111}\}$ &  & \gcmark &  &  &  & $24/79$ \\
    5. & $\{\bstr{**01}, \bstr{**10}, \bstr{**11}\}$ & \gcmark &  &  &  &  & $684/905$ \\
    6. & \begin{tabular}{@{}l@{}l@{}l@{}l@{}l@{}l@{}l}$\{ $&$\bstr{0011},\,$&$ \bstr{0101},\,$&$ \bstr{0110},\,$&$ \bstr{0111},\,$&$ \bstr{1001},\,$&$ \bstr{1010}, $ \\
 &$ $&$ \bstr{1011},\,$&$ \bstr{1100},\,$&$ \bstr{1101},\,$&$ \bstr{1110},\,$&$ \bstr{1111}\}$\end{tabular} & \gcmark &  &  &  &  & $10/179$ \\
\end{longtblr}

\begin{longtblr}[
  caption = {A list of the 14 minimal NP-hard predicates for $k=4$. The last column called dependency lists how many predicates are NP-hard by direct implication and how many of these which are not implied by any other predicate.},
  label = {tab:4easiest_hard},
]{
  colspec = {l@{}l|ccccr},
  rowhead = 1,
  row{odd} = {blue9},
  row{1} = {gray9},
  columns = {colsep=4pt},
} 
     & Predicate & $\ref{match+ADA}$ & $\ref{invmatch+ADA}$ & $\ref{thm:unate+ADA}$ & $\ref{thm:split}$ & Dep \\
\hline
    1. & $\{\bstr{0001}, \bstr{0010}, \bstr{0011}, \bstr{0100}\}$ & \gcmark & \gcmark & \gcmark & \gcmark & $45/570$ \\
    2. & $\{\bstr{0001}, \bstr{0011}, \bstr{0101}, \bstr{0110}, \bstr{1000}\}$ &  & \gcmark & \gcmark & \gcmark & $1/360$ \\
    3. & $\{\bstr{0011}, \bstr{0100}, \bstr{0111}, \bstr{1000}\}$ & \gcmark & \gcmark & \gcmark & \gcmark & $36/694$ \\
    4. & $\{\bstr{0011}, \bstr{0101}, \bstr{0110}, \bstr{0111}, \bstr{1000}\}$ & \gcmark &  &  &  & $1/240$ \\
    5. & $\{\bstr{0010}, \bstr{0100}, \bstr{0110}, \bstr{1001}\}$ & \gcmark & \gcmark & \gcmark & \gcmark & $13/520$ \\
    6. & $\{\bstr{0010}, \bstr{0011}, \bstr{0100}, \bstr{0111}, \bstr{1001}\}$ & \gcmark & \gcmark & \gcmark & \gcmark & $5/560$ \\
    7. & $\{\bstr{0011}, \bstr{0101}, \bstr{0110}, \bstr{1000}, \bstr{1001}\}$ &  &  &  & \gcmark & $3/360$ \\
    8. & $\{\bstr{0010}, \bstr{0100}, \bstr{0111}, \bstr{1000}, \bstr{1001}\}$ & \gcmark & \gcmark &  &  & $1/398$ \\
    9. & $\{\bstr{0011}, \bstr{0100}, \bstr{0111}, \bstr{1001}, \bstr{1010}\}$ & \gcmark & \gcmark & \gcmark & \gcmark & $3/372$ \\
    10. & $\{\bstr{0001}, \bstr{0010}, \bstr{0111}, \bstr{1011}, \bstr{1100}\}$ & \gcmark & \gcmark & \gcmark & \gcmark & $2/304$ \\
    11. & $\{\bstr{0001}, \bstr{0010}, \bstr{0100}, \bstr{1000}, \bstr{1111}\}$ & \gcmark &  &  &  & $1/90$ \\
    12. & $\{\bstr{0011}, \bstr{0101}, \bstr{0110}, \bstr{1000}, \bstr{1111}\}$ & \gcmark & \gcmark & \gcmark & \gcmark & $1/220$ \\
    13. & $\{\bstr{0001}, \bstr{0010}, \bstr{0111}, \bstr{1000}, \bstr{1111}\}$ & \gcmark & \gcmark & \gcmark & \gcmark & $6/279$ \\
    14. & $\{\bstr{0011}, \bstr{0100}, \bstr{0110}, \bstr{1000}, \bstr{1001}, \bstr{1111}\}$ & \gcmark & \gcmark & \gcmark & \gcmark & $1/200$ \\
\end{longtblr}

\paragraph{Arity $5$.}
For $k=5$ there are a total of $18\,666\,623$ predicates of which $1\,290\,862$ are tractable, $17\,375\,572$ are NP-hard and $189$ are unknown. There are $32$ maximal tractable predicates, described in detail in \cref{sec:tractable5}.  There is a rather large number of $241$ minimal hard predicates and we refrain from listing these.  Of more interest are the minimal and maximal of the $189$ unknown predicates.    There are $25$ and $19$ such predicates, respectively, and they can be found in \cref{sec:fiPCSP unknown list}.    It is also interesting to look at the distribution of tractable vs NP-hard predicates as a function of the number of satisfying assignments, and this is shown in \cref{fig:fiPCSP-by-weight}.

\begin{figure}[ht]
    \centering
    \includegraphics[width=0.9\textwidth]{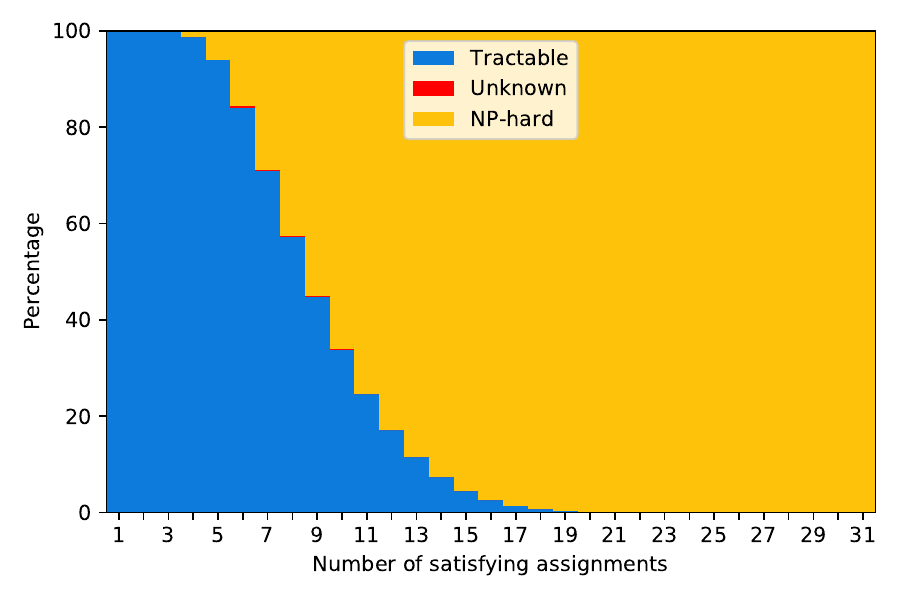}
    \caption{Distribution of predicates $A$ of arity $5$ such that $\fiPCSP(A, \OR)$ is tractable or NP-hard, by weight.  The number of unknown predicates is very small but exists for weights between 6 and 12 (both inclusive), and can be glimpsed as thin red lines at weights $6$, $7$, and $8$.}
    \label{fig:fiPCSP-by-weight}
\end{figure}

\subsection{Results for Promise-SAT (without idempotence)}
\label{small arity fPCSP}

Let us now turn to the non-idempotent setting and the complexity of $\fPCSP(A, \OR)$.  Here we use \cref{lemma:idempotent}, which as discussed in \cref{sec:non-idempotent OR} effectively says that the complexity of $\fPCSP(A, \OR)$ is captured by either $\fiPCSP(A, \OR)$ or $\fiPCSP(A \XOR 1^k, \OR)$, whichever is easier (with the second option only being possible when $1^k \not\in A$).  

In this setting, there is an additional symmetry causing the equivalence classes of predicates to be larger: if $1^k \not\in A$, then the complexity of $\fPCSP(A, \OR)$ is the same as the complexity of $\fPCSP(A \XOR 1^k, \OR)$ and thus we consider the two predicates $A$ and $A \XOR 1^k$ equivalent.  As a result the total number of non-equivalent predicates is approximately $25\%$ smaller in this setting than in the $\fiPCSP$ setting (approximately half the predicates $A$ satisfy $1^k \not\in A$ and these are paired up in equivalent pairs by this additional symmetry).

Using \cref{lemma:idempotent} together with our classification of $\fiPCSP(A, \OR)$, it is straightforward to obtain a similar classification for $\fPCSP(A, \OR)$.  For brevity we only give a summary in \cref{tab:table_overview_nonidem.tex} and refrain from performing a deeper dive into the data, as it looks largely similar to the idempotent case in \cref{sec:small arities fiPCSP}.  However one interesting detail to note is that the relative number of unknown predicates for arity $5$ is noticeably smaller here than in \cref{sec:small arities fiPCSP}.  In other words many of the unknown predicates $A$ for $\fiPCSP(A, \OR)$ are such that $\fPCSP(A, \OR)$ is easy by virtue of having non-idempotent block-symmetric polymorphisms (in particular either $\overline{\Maj}$ or $\overline{\Par}$).

\begin{longtblr}[
          caption = {Summary of classification of complexity of $\fPCSP(A, \OR)$ for $A$ of arity up to $5$.},
          label = {tab:table_overview_nonidem.tex},
]{
  colspec = {l|rrrr},
  rowhead = 1,
  row{odd} = {blue9},
  row{1} = {gray9},
  columns = {colsep=4pt},
} 
     & Total & Tractable & NP-hard & Unknown \\
\hline
    $k=2$ & $5$ & $5$ & $0$ & $0$ \\
    $k=3$ & $32$ & $28$ & $4$ & $0$ \\
    $k=4$ & $1\,549$ & $848$ & $701$ & $0$ \\
    $k=5$ & $14\,003\,603$ & $1\,253\,003$ & $12\,750\,541$ & $59$ \\
\end{longtblr}

\subsection{Results for Promise-Usefulness}\label{small arity usefuleness}

Finally let us consider the classification of which predicates $A$ are promise-useful.
As shown in \cref{promise constants irrelevant}, the two á priori possibly different notions of $\fiPCSP$-usefulness and $\fPCSP$-usefulness cannot be distinguished by BLP+AIP or small fixing assignments, so the data we provide here applies for both $\fiPCSP$-usefulness and $\fPCSP$-usefulness.

After having classified the complexity of $\fiPCSP(A, \OR)$ for all (or most) $A$, \cref{lemma:OR} gives a simple procedure to tell if a predicate $A$ is $\fiPCSP$-useful or $\fiPCSP$-useless.  It states that $A$ is $\fiPCSP$-useful if and only if there exists $b \not \in A$ such that $\fiPCSP(A \XOR b, \OR)$ is tractable.  If for a single $b$ this is the case we can conclude that $A$ is promise useful.  On the other hand we need to establish that it is NP-hard for all possible $b$ to conclude that $A$ is promise-useless.
This enables us to fully classify which predicates are promise-useful for $k\leq 4$ while for $k=5$ the complexity of several predicates remains unknown.

In addition to the already applied symmetry of the input bits (\cref{sec:equivalence classes}), there is now additional symmetry -- clearly from the characterization above it follows that $A$ and $A \XOR b$ are equivalent (w.r.t.~promise-usefulness) for any $b \not\in A$.  Thus the total number of (equivalence classes of) predicates is smaller here than in \cref{sec:small arities fiPCSP} by a factor of roughly $2^{k-1}$ (since this is the typical number of $b \not\in A$).

\cref{tab:promise_overview} gives an overview of the number of promise-useful and promise-useless (and unknown) predicates of arities up to $5$.

\begin{longtblr}[
  caption = {Summary of classification of promise-usefulness of predicates $A$ of arity up to $5$.},
  label = {tab:promise_overview},
]{
  colspec = {l|rrrr},
  rowhead = 1,
  row{odd} = {blue9},
  row{1} = {gray9},
  columns = {colsep=4pt},
} 
     & Total & Useful & Useless & Unknown \\
\hline
    $k=2$ & $4$ & $4$ & $0$ & $0$ \\
    $k=3$ & $20$ & $16$ & $4$ & $0$ \\
    $k=4$ & $400$ & $230$ & $170$ & $0$ \\
    $k=5$ & $1\,228\,156$ & $156\,135$ & $1\,071\,962$ & $59$ \\
\end{longtblr}

\paragraph{Arity $2$.}

Again for $k=2$ all predicates are useful as they are indeed useful for $\OR$.

\paragraph{Arity $3$.}

For $k=3$ it is again parity and $\twoSAT$ that give the maximal useful 
predicates as indicated in \cref{tab:promise_3hardest_useful}. 
We also have two minimally useless predicates as indicated by \cref{tab:promise_3easiest_useless}.
\begin{longtblr}[
  caption = {A list of the two maximal promise-useful predicates for $k=3$.},
  label = {tab:promise_3hardest_useful},
]{
  colspec = {l@{}l|rrr},
  rowhead = 1,
  row{odd} = {blue9},
  row{1} = {gray9},
  columns = {colsep=4pt},
} 
     & Predicate & $\MajFam$ & $\ParFam$ & Dep. \\
\hline
    1. & $\{\bstr{00*}, \bstr{01*}, \bstr{10*}\}$ & \gcmark &  & $12/15$ \\
    2. & $\{\bstr{000}, \bstr{011}, \bstr{101}, \bstr{110}\}$ &  & \gcmark & $1/4$ \\
\end{longtblr}
\begin{longtblr}[
  caption = {A list of the two minimal promise-useless predicates for $k=3$.},
  label = {tab:promise_3easiest_useless},
]{
  colspec = {l@{}l|r},
  rowhead = 1,
  row{odd} = {blue9},
  row{1} = {gray9},
  columns = {colsep=4pt},
} 
     & Predicate & Dep. \\
\hline
    1. & $\{\bstr{000}, \bstr{001}, \bstr{011}, \bstr{101}, \bstr{110}\}$ & $2/3$ \\
    2. & $\{\bstr{001}, \bstr{010}, \bstr{011}, \bstr{100}, \bstr{101}, \bstr{110}\}$ & $1/2$ \\
\end{longtblr}

\paragraph{Arity $4$.}

The maximal useful predicates for $k=4$ are shown in \cref{tab:promise_4hardest_useful}.  
The two first maximal predicates correspond to $\twoSAT$ and (non-strict) majority of $4$ variables (recall that we can negate variables -- in our ordering the representatives chosen are the negated versions of these two predicates).  The other two predicates are parity constraints on three or four variables.
It turns out that there are $5$ minimal useless predicates of arity $4$ and these are given in \cref{tab:promise_4easiest_useless}.

\begin{longtblr}[
  caption = {A list of the $4$ maximal promise-useful predicates for $k=4$.},
  label = {tab:promise_4hardest_useful},
]{
  colspec = {l@{}l|rrr},
  rowhead = 1,
  row{odd} = {blue9},
  row{1} = {gray9},
  columns = {colsep=4pt},
} 
     & Predicate & $\MajFam$ & $\ParFam$ & Dep. \\
\hline
    1. & $\{\bstr{00**}, \bstr{01**}, \bstr{10**}\}$ & \gcmark &  & $91/216$ \\
    2. & \begin{tabular}{@{}l@{}l@{}l@{}l@{}l@{}l@{}l@{}l@{}l}$\{ $&$\bstr{0000},\,$&$ \bstr{0001},\,$&$ \bstr{0010},\,$&$ \bstr{0011},\,$&$ \bstr{0100},\,$&$ \bstr{0101},\,$&$ \bstr{0110},\,$&$ \bstr{1000}, $ \\
 &$ $&$  $&$  $&$  $&$  $&$ \bstr{1001},\,$&$ \bstr{1010},\,$&$ \bstr{1100}\}$\end{tabular} & \gcmark &  & $8/132$ \\
    3. & $\{\bstr{000*}, \bstr{011*}, \bstr{101*}, \bstr{110*}\}$ &  & \gcmark & $3/21$ \\
    4. & $\{\bstr{0001}, \bstr{0010}, \bstr{0100}, \bstr{0111}, \bstr{1000}, \bstr{1011}, \bstr{1101}, \bstr{1110}\}$ &  & \gcmark & $1/15$ \\
\end{longtblr}

\begin{longtblr}[
  caption = {A list of the $5$ minimal promise-useless predicates for $k=4$.},
  label = {tab:promise_4easiest_useless},
]{
  colspec = {l@{}l|r},
  rowhead = 1,
  row{odd} = {blue9},
  row{1} = {gray9},
  columns = {colsep=4pt},
} 
     & Predicate & Dep. \\
\hline
    1. & $\{\bstr{0000}, \bstr{0111}, \bstr{1001}, \bstr{1010}, \bstr{1100}\}$ & $19/132$ \\
    2. & $\{\bstr{0001}, \bstr{0010}, \bstr{0011}, \bstr{0111}, \bstr{1001}, \bstr{1010}, \bstr{1100}\}$ & $3/105$ \\
    3. & $\{\bstr{0011}, \bstr{0100}, \bstr{0111}, \bstr{1001}, \bstr{1010}, \bstr{1100}\}$ & $10/138$ \\
    4. & $\{\bstr{0011}, \bstr{0101}, \bstr{0110}, \bstr{0111}, \bstr{1000}, \bstr{1001}, \bstr{1010}, \bstr{1100}\}$ & $1/39$ \\
    5. & $\{\bstr{0011}, \bstr{0100}, \bstr{0110}, \bstr{0111}, \bstr{1000}, \bstr{1001}, \bstr{1011}, \bstr{1100}\}$ & $1/29$ \\
\end{longtblr}

\paragraph{Arity $5$.}

For $k=5$ we have $6$ maximal useful predicates shown in \cref{tab:promise_5hardest_useful} -- the $4$ maximal predicates of arity $4$, and two obvious generalizations of them to arity $5$.  Note that the second predicate does not quite correspond to a simply majority predicate but is a bit more complicated, whereas the last predicate is simply even parity on all five variables.
There is a large number of $73$ minimal promise-useless predicates and we refrain from listing these.  Of more interest are the minimal and maximal predicates whose promise-usefulness we are unable to classify.  There are $9$ and $7$ of these, respectively, and they can be found in \cref{sec:promise unknown list}.

\begin{longtblr}[
  caption = {A list of the $6$ maximal promise-useful predicates for $k=5$.},
  label = {tab:promise_5hardest_useful},
]{
  colspec = {l@{}l|rrr},
  rowhead = 1,
  row{odd} = {blue9},
  row{1} = {gray9},
  columns = {colsep=4pt},
} 
     & Predicate & $\MajFam$ & $\ParFam$ & Dep. \\
\hline
    1. & $\{\bstr{00***}, \bstr{01***}, \bstr{10***}\}$ & \gcmark &  & $91\,485/141\,063$ \\
    2. & \begin{tabular}{@{}l@{}l@{}l@{}l@{}l@{}l@{}l}$\{ $&$\bstr{00000},\,$&$ \bstr{00001},\,$&$ \bstr{00010},\,$&$ \bstr{00011},\,$&$ \bstr{00100},\,$&$ \bstr{00101}, $ \\
 &$\bstr{00110},\,$&$ \bstr{00111},\,$&$ \bstr{01000},\,$&$ \bstr{01001},\,$&$ \bstr{01010},\,$&$ \bstr{01011}, $ \\
 &$\bstr{01100},\,$&$ \bstr{01101},\,$&$ \bstr{01110},\,$&$ \bstr{10000},\,$&$ \bstr{10001},\,$&$ \bstr{10010}, $ \\
 &$ $&$  $&$  $&$  $&$ \bstr{10100},\,$&$ \bstr{11000}\}$\end{tabular} & \gcmark &  & $355/30\,027$ \\
    3. & \begin{tabular}{@{}l@{}l@{}l@{}l@{}l@{}l@{}l}$\{ $&$\bstr{0000*},\,$&$ \bstr{0001*},\,$&$ \bstr{0010*},\,$&$ \bstr{0011*},\,$&$ \bstr{0100*},\,$&$ \bstr{0101*}, $ \\
 &$ $&$ \bstr{0110*},\,$&$ \bstr{1000*},\,$&$ \bstr{1001*},\,$&$ \bstr{1010*},\,$&$ \bstr{1100*}\}$\end{tabular} & \gcmark &  & $12\,549/60\,975$ \\
    4. & $\{\bstr{000**}, \bstr{011**}, \bstr{101**}, \bstr{110**}\}$ &  & \gcmark & $205/645$ \\
    5. & \begin{tabular}{@{}l@{}l@{}l@{}l@{}l@{}l@{}l}$\{ $&$\bstr{0001*},\,$&$ \bstr{0010*},\,$&$ \bstr{0100*},\,$&$ \bstr{0111*},\,$&$ \bstr{1000*},\,$&$ \bstr{1011*}, $ \\
 &$ $&$  $&$  $&$  $&$ \bstr{1101*},\,$&$ \bstr{1110*}\}$\end{tabular} &  & \gcmark & $58/458$ \\
    6. & \begin{tabular}{@{}l@{}l@{}l@{}l@{}l@{}l@{}l}$\{ $&$\bstr{00000},\,$&$ \bstr{00011},\,$&$ \bstr{00101},\,$&$ \bstr{00110},\,$&$ \bstr{01001},\,$&$ \bstr{01010}, $ \\
 &$\bstr{01100},\,$&$ \bstr{01111},\,$&$ \bstr{10001},\,$&$ \bstr{10010},\,$&$ \bstr{10100},\,$&$ \bstr{10111}, $ \\
 &$ $&$  $&$ \bstr{11000},\,$&$ \bstr{11011},\,$&$ \bstr{11101},\,$&$ \bstr{11110}\}$\end{tabular} &  & \gcmark & $50/150$ \\
\end{longtblr}

\section{Asymptotic Bounds for Large Arities} \label{sec:large arity}

In this section, we study promise-usefulness and $\fiPCSP(A, \OR)$ for predicates $A$ of large arity $k$, and show that almost all such $A$, even with a relatively small number of satisfying assignments, are hard.

\newcommand{\randA}{\mathcal{A}}

Let $\randA_{k,p}$ denote the distribution over a random predicate $A$ on $\{0,1\}^{k}$, where each $x \in \{0,1\}^{k}$ is included in $A$ with probability $p$, independently.  We condition our results on $A$ not containing $0^k$.  As the $p$ we consider is very small this is a very mild conditioning but allowing this to happen makes some of our arguments more uniform, avoiding special cases.

\begin{theorem}
\label{thm:random hard}
    We have
    \[
    \Pr_{A \sim \randA_{k,p}}\left[ \text{$\fPCSP(A, \OR)$ admits a non-dictator} \,|\, 0^k \not\in A \right] \le O\left(2^k (e^{-p2^{k/6}}+e^{-p^2 2^k})\right).
    \]
\end{theorem}

In particular, for $p$ is slightly larger than $2^{-k/6}$, a random predicate $A$ from $\randA_{k,p}$ is NP-hard with high probability as $k \rightarrow \infty$.  It is not unlikely that the value for $p$ can be improved but currently we see no way of establishing hardness in the very sparse regime with a $p$ of the form $2^{o(k)-k}$, i.e., for predicates accepting only $2^{o(k)}$ strings.  

The problem is that very sparse predicates allow for all polymorphisms of small arity.  For example, suppose $A$ contains $K$ strings, and we are interested in polymorphisms of arity 6.  In this case there are $K^6$ possible obstruction matrices, each being, except for repetitions, totally random.  The probability that an individual random matrix lacks any fixed row, say $000001$ is at most $(63/64)^k$.  Thus if $K=2^{o(k)}$ all rows are likely to be present in all matrices and thus all arity 6 functions are polymorphisms.  This implies that to get results for this small range of $p$s we need to analyze polymorphisms of large arity and this seems to require new methods.

Let us turn to the proof of \cref{thm:random hard} and first prove that it is enough to study polymorphisms of arity 6.

\begin{claim}
\label{small nondictator minions}
   Let $\Minion$ be a minion containing a non-dictator function $f: \{0,1\}^\ell \rightarrow \{0,1\}$.  Then $\Minion$ contains a non-dictator function $g: \{0,1\}^6 \rightarrow \{0,1\}$ of arity $6$.
\end{claim}

\begin{proof}
    A function is a non-dictator iff it depends on at least two variables and hence there are inputs $x$ and $y$ and $i\not= j$ such that $f(x)\not= f(x \XOR e_i)$ and  $f(y)\not= f(y \XOR e_j)$.  We can divide the variables outside indices $i$ and $j$ into four groups depending on the value of the pair $(x_k, y_k)$.  Identify the variables inside a group by a single variable.  These four variables jointly with $i$ and $j$ gives a six variable minor of $f$ which is not a dictator.
\end{proof}

Next we have the following observation.

\begin{lemma} \label{lemma:outside}
Let $S\subseteq [k]$ be of cardinality $L$.  The probability that there is some assignment $\alpha$ on $S$ such that $A$ contains no string that is equal to $\alpha$ on $S$ is bounded by $2^{L} e^{-p 2^{k-L}}$.
\end{lemma}

\begin{proof}
    We do a union bound over different values of $\alpha$ (resulting in the factor $2^L$).  Once we fix $\alpha$, there are $2^{k-L}$ candidate strings that that are equal to $\alpha$ on $S$.  The probability that none of them belongs to $A$ is $(1-p)^{2^{k-L}} \leq e^{-p 2^{k-L}}$.
\end{proof}

We need one more simple fact.

\begin{lemma}\label{lemma:xor}
    Fix an $\alpha \in \{ 0, 1\}^k$, $\alpha \ne 0^k$.  The probability that $A$ does not contain two strings $a$ and $b$ such that $a \XOR b =\alpha$ is at most $e^{-p^2 2^{k-1}}$.
\end{lemma}

\begin{proof}  There are $2^{k-1}$ possible pairs $(a,b)$ with the given property.  Each pair is in $A$ with probability $p^2$ and as they are disjoint, the events are independent, and the lemma follows by a simple calculation.
\end {proof}

\begin{proof}[Proof of \cref{thm:random hard}]
    By \cref{small nondictator minions}, it suffices to prove that with high probability, $A \sim \randA_{k,p}$ has no non-dictator polymorphisms of arity $6$.  Let $f$ be a non-dictator function of arity 6 and let us construct an obstruction matrix for it. Let us first observe that $f$, for each $j$ either can written as $x_j \XOR g(x')$ where $x'$ gives the other five variables or there is some input $x^{(j)}$ such that
    $f(x^{(j)})=f(x^{(j)} \XOR e_j)=0$.  Using this observation and permuting the variables we can assume that $f(x)=\oplus_{i=1}^t x_i \XOR g(x')$ where $x'$ corresponds to the variables $x_{t+1} \ldots x_6$ and we have assignments $x^{(j)}$ such that $f(x^{(j)})=f(x^{(j)} \XOR e_j)=0$ for $t+1 \leq j \leq 6$.

If $g$ is constant then $f$ is a parity function.  In this case $t$ is odd size since $f$ is folded.  If $t=1$ then $f$ is a dictator and there is nothing to prove.  If $t=5$ we can make two columns equal reducing to the case $t=3$.  In this case, fix the first column to any element of $A$.  Then, since we condition on $0^k \not\in A$, \cref{lemma:xor} implies that we can, with high probability, pick the second and third columns to complete the obstruction matrix.  Even using a union bound over all values of the first column, the failure probability is bounded by $2^ke^{-p^2 2^{k-1}}$.

Now suppose that $g$ is non-constant.  We construct an obstruction matrix and suppose for notational convenience that $k/6$ is an integer.  For each value $1 \leq j \leq 6$ we assign $k/6$ rows.  Let $y$ and $z$ be such that $g(y)=0$ and $g(z)=1$.   Now for $1 \leq j \leq t$  we assign $k/6$ rows that are either $e_j$ followed by $z$ or $0^t$ followed by $y$,  For $t+1 \leq j \leq  6$ we have $k/6$ rows that are either $x^{(j)}$ or $x^{(j)} \XOR e_j$ where $x^{(j)}$ is as discussed above.   Note that, by construction, $f$ evaluates to 0 on each row and we claim that, with high probability, we can make the choices such that each column belongs to $A$.

Let us first make the choices for the first $t$ columns.  Note for these columns, all values outside the $k/6$ rows that belongs the column itself are independent of the choices for other rows as the two alternatives give equal values in this coordinate.   Hence using \cref{lemma:outside}, except with probability $2^{5k/6} e^{-p 2^{k/6}}$ we can make choices to make this column belong to $A$.

 The argument can now be repeated for the remaining $6-t$ columns.  The values outside its own rows are now fixed and we can apply \cref{lemma:outside}.
\end{proof}

\begin{remark}
If we do not allow negations (and hence are interested in $f$ which are not folded), then the equivalent of \cref{thm:random hard} still holds.  In this situation we need to study $\PCSP(A, \NAE)$.  In the above proof we used twice that a potential polymorphism is folded.  Firstly to conclude that if $f$ is a parity then this is of odd size.  It is, however, straightforward to deal with parities of even size.  We also used the fact that there are inputs $x^{(j)}$ such that $f(x^{(j)})=f(x^{(j)} \oplus e_j)=0$.  This is no longer true if $f$ is not folded.  The corresponding fact is that $f$ takes the value $b$ on at least half the inputs then for each $j$ it is possible to find a $x^{(j)}$ such that $f(x^{(j)})=f(x^{(j)} \oplus e_j)=b$.  It is then possible to find an obstruction matrix where the output is $b^k$. The rest of the proof is unchanged.  We omit the details.
\end{remark}

As observed above the probability that $A$ is useful for $\OR$ is bounded by $2^{-\omega(k)}$ whenever $p = \omega (k 2^{-k/6})$.  The analysis that $A$ is useful for any $\OR \oplus b$ is identical and as there are only $2^k$ possible values for $b$ we conclude that any such $A$ is useless except with probability $2^{-\omega(k)}$.  We record this as a corollary.

\begin{corollary}
\label{random useless}
Suppose $A$ is selected according to $\randA_{k,p}$.  If  $p = \omega (k 2^{-k/6})$ then, except with probability $2^{-\omega(k)}$, we have that $A$ is promise-useless.  This is true with or without negations, as well as with or without allowed constants.
\end{corollary}

We next proceed to prove BLP+AIP is only applicable to very sparse random predicates.

\subsection{Threshold for the BLP+AIP Algorithm}

In this section it is easier to work in $\pm 1$-notation with $-1$ taking the role of $1$.  We start with a definition.  

\begin{definition}
Let $\epsilon > 0$.  A set, $A \subseteq \{ -1, 1\}^k$ is {\em $\epsilon$-somewhat spread} if for any unit vector $u$ there is an $a \in A$ such that $\scalprod{a}{u} \leq -\epsilon$.
\end{definition}

We get an immediate consequence.
\begin{lemma}\label{lemma:addone}
    Suppose $\Vert v \Vert \geq k/\epsilon \ge 1$ and that $A \subseteq \{-1,1\}^k$ is $\epsilon$-somewhat spread.  Then there is an $a\in A$ such that $\Vert v+a \Vert^2 \leq \Vert v \Vert^2 - k$.
\end{lemma}

\begin{proof}  Take a vector in $a$ with $\scalprod{a}{v} \leq - \epsilon \Vert v \Vert$. Then
    \[
    \scalprod{v+a}{v+a} = \scalprod{v}{v} +2 \scalprod{a}{v} + \scalprod{a}{a} \leq \Vert v \Vert^2 - 2\epsilon \Vert v \Vert + k \leq \Vert v \Vert^2 - k.
    \]
\end{proof}

We are interested in a second property of vectors and in order to formulate this we need to study integer lattices.  We start by recalling some basic definitions.  We do not state definitions in full generality and for a more careful treatment we refer to a text-book such as \cite{MG02}.  In our case we are interested in lattices that are subsets of $\mathbb{Z}^n$ and let $L$ be a generic lattice. It is usually given by a basis, $(b_i)_{i=1}^n$  with integer coordinates and  $L$ is given by the points $\sum_{i=1}^n a_i b_i$ where $a_i \in \mathbb{Z}$.   In a basis the vectors are linearly independent over $\mathbb{R}$ but it is possible to consider a larger set of possibly linearly dependent vectors and use them as a generating set.  The set of all integer combinations of such a set is also a lattice and it is easy to compute a basis for this lattice given the generators.

The determinant, $\mbox{det}(L)$,  of $L$ is the absolute value of the determinant of the matrix which has the basis vectors as columns.  As we are considering only lattices which are subsets of $\mathbb{Z}^n$ a lattice can be described by a set of modular equations.  The system can be written as $\sum_{j=1}^n a_{ij} x_j \equiv 0$ modulo $p_i^{m_i}$ where $p_i$ are (not necessarily distinct) prime numbers and $m_i \geq 1$.  The equations can be chosen such that $\prod_i p_i^{m_i}= \mbox{det}(L)$. 

For any $a\in A$ we create a vector $a'$ of dimension $k+1$ by adding an additional coordinate that always takes the value 1.  Let $L$ be the lattice in $k+1$ dimension that is generated by these vectors.   We say that $A$ is {\em modular free} if $L$ consists of all integer vectors where all coordinates have the same parity.  As all $a'$ in the generating set have this property property this is as large as $L$ can be. The two defined properties just defined turn out to imply that $A$ does not have any (block)-symmetric polymorphisms by the following lemma.

\begin{lemma}\label{lemma:combine}
Suppose $A \subseteq \{-1,1\}^k$ is modular free and $\epsilon$-somewhat spread for some $\epsilon > 0$.
Then for all $x \in \{ -1 ,1 \}^k$,  and all sufficiently large odd integers $\ell$, the vector $(x_1, \ldots, x_k, \ell)$ can be written as a non-negative integer combination of $a'$ where $a\in A$.  
\end{lemma}

Before proving this lemma we give the important corollary.
\begin{corollary}
Suppose $A$ is modular-free and $\epsilon$-somewhat spread then for all sufficiently large odd integers $\ell$, $\fPCSP(A, \OR)$ does not admit a block-symmetric polymorphism with block sizes $\ell$ and $\ell+1$.
\end{corollary}

\begin{proof}
    Let $x$ be an arbitrary element of $A$.
    By \cref{lemma:combine}, the vectors $(1^k, \ell)$, $((-1)^k, \ell)$ and $(-x, \ell)$ can all be written as non-negative combinations of $a'$ for $a \in A$.  Clearly then also $(0^k, \ell+1)$ can be written as such a combination.

    A block-symmetric function $f: \{-1,1\}^{2\ell+1} \rightarrow \{-1,1\}$ with block sizes $\ell$ and $\ell+1$ can be written as $f(x) = g(\sum_{i=1}^{\ell} {x_i}, \sum_{i=\ell+1}^{2\ell+1} x_i)$.    
    Since $(1^k, \ell)$, $((-1)^k, \ell)$ and $(0^k, \ell+1)$ can all be written as non-negative integer combinations as described above, we can create two $k \times (2\ell+1)$ matrices $M^{(+)}$ and $M^{(-)}$ such that $f(x,y) = g(1,0)$ for all rows of $M^{(+)}$ and $f(x,y) = g(-1,0)$ for all rows of $M^{(-)}$.  However if $f$ is folded, one of $g(1,0)$ and $g(-1,0)$ is false, yielding an obstruction to $f$ being a polymorphism of $\fPCSP(A, \OR)$.
\end{proof}

\begin{proof}[Proof of \cref{lemma:combine}]
By the assumption of being modular-free we know that any integer vector of the form $(v,0)$ where all components of $v$ are even can be written in form $\sum_{a \in A} b_a a'$ for some integers $b_a$.  Let $M$ be the maximal absolute value needed as a coefficient to represent any vector $v$ with norm at most $2 k / \epsilon$.  Define $w=M \sum_{a \in A} a'$ and notice that $\Vert w \Vert \leq k 2^k M$. Using \cref{lemma:addone} we can add additional vectors $a'$ with $a\in A$ to $w$ such that for all sufficiently large odd $\ell$ we can get a vector of the form $(v,\ell)$ where the length of $v$ is bounded  by $k/\epsilon$.  Using that $(x-v,0)$ can be written as an integer combination of $a'$ with coefficients bounded by $M$ we conclude that $(x,\ell)$ can be written as a positive integer combination of $a'$ with $a \in A$.  This follows as any coefficient of $w$ already is at least $M$ preventing it from turning negative. 
\end{proof}

We need to prove that a random $A$ has the two properties.  On road to this we have the following definition.

\begin{definition}
    A set $B \subseteq \mathbb{R}$ of unit vectors is an {\em $\epsilon$-net} if for each unit vector $v\in \mathbb{R}^k$ there is $b\in B$ such that $\Vert b-v \Vert \leq \epsilon$.
\end{definition}

It is well known that it is possible to construct an $\epsilon$-net of cardinality $(1/\epsilon)^{O(k)}$ \cite{kochol94constructive}.  We want to prove that if $A$ has $Ck$ random elements for a sufficiently large value of $C$, then it is very likely to be $c$-somewhat spread for an absolute constant $c$.  Before embarking on the proof let us point out that the result is straightforward if $A$ has $C k \log k $ elements.  

Let us consider an $1/4 \sqrt{k}$-net $B$.  It has $2^{O(k \log k)}$ elements.  Take any $b\in B$.  The probability that $\scalprod{a}{b} \leq -1/2$ for a random $a$ is easily seen to be $\Omega (1)$.  By the union bound, for a sufficiently large constant $C$, if we pick $C k \log k$ random $a$ as elements of $A$, then, with high probability, for each $b\in B$ there exists an $a\in A$ such that $\scalprod{a}{b}\leq -1/2$.  This implies that for any unit vector $u$, if $b$ is the closest element of $B$, then $\scalprod{a}{u} \leq -1/4$ for any $a$ such that $\scalprod{a}{b} \leq -1/2$. To make do with $O(k)$ vectors in $A$ we need a more complicated argument and we use the techniques of \cite{pajhrandomsupport}.   Let us recall the highlights.  

We choose $B$ to be an $\epsilon$-net for a small, but absolute constant $\epsilon$ specified below.  Let $K$ be the convex body $\sum_a c_a a'$ where $c_a \in [0,1]$ and the sum is over $a\in A$.  For any direction $b\in B$ we look at the number $m_b= \min_{x \in K}\scalprod{b}{x}$ and $M_b=\max_{x \in K} \scalprod{b}{x}$. By inspection we have that $m_b = \sum_a \min (0, \scalprod{b}{a'})$.  When picking a random $a$ we have that $\scalprod{a'}{b}$ is a symmetric random variable with mean 0 and variance $1$.  It is not difficult to see that the expectation of $|\scalprod{a}{b}|$ is $e_b$, where there are absolute constants, $m$ and $M$ such that $0< m \leq  e_b \leq M$ for any unit vector $b$.  By symmetry it follows that $E[\min (0, \scalprod{b}{x}]= -e_b/2$.

Set $\epsilon = m/16M$ and note that $\epsilon < 1/16$.  Using standard Chernoff estimates it follows that, for sufficiently large constant $C$, with high probability when $A$ has $Ck$ random vectors, then  $m_b \leq -Ck e_b/4$ and $M_b \leq Ck e_b$ for any $b \in B$.  Suppose this is the case.  We have some simple claims.

\begin{claim}
    For any $x \in K$ we have $\Vert x \Vert  \leq 2CkM$.
\end{claim}

\begin{proof}
Take any vector $x\in K$ and let $b\in B$ be as close as possible to $x/\Vert x \Vert$. It is not difficult to see that $\scalprod{b}{x} \geq \Vert x \Vert /2$ and the claim follows by the bound on $M_b$.
\end{proof}

\begin{claim}\label{lemma: boundu}
For any unit vector $u$ we have $m_u \leq -mCk/8$.
\end{claim}

\begin{proof}
As $B$ is an $\epsilon$-net we can write $u=b+w$ where $b\in B$ and $\Vert w \Vert \leq \epsilon$.  Take $x \in K$ that minimizes $\scalprod{b}{x}$, then
$$
\scalprod{u}{x}= \scalprod{b}{x}+\scalprod{w}{x} \leq -m C k/4 + \Vert w \Vert \Vert x \Vert 
\leq -m C k / 4 + 2Ck M \epsilon \leq -mCk/8.$$
    \end{proof}

As $m_u =\sum_{a \in A} \min(0, \scalprod{u}{a'})$ and the sum has $Ck$ terms it follows from \cref{lemma: boundu} that $A$ is $m/8$-somewhat spread and we turn to the question of $A$ being modular free.

We start by establishing that we get a full dimensional lattice.
\begin{lemma}
    Suppose $A$ contains $3k$ random elements, then there is a constant $c<0$ such that with probability $1-2^{-ck}$ the vectors $a'$ for $a\in A$ span $\mathbb{R}^{k+1}$.
\end{lemma}

\begin{proof} We select the elements of $A$ one by one.  We claim that if the corresponding vectors $a'$ chosen so far do not span $\mathbb{R}^{k+1}$ then with probability at least $1/2$ the next vector is linearly independent of the vectors so far. This follows as any $k-1$-dimensional space in $\mathbb{R}^k$ only contains at most half the points of the hypercube.  This implies that the probability of not spanning the entire space is at most the probability of getting $k-1$ zeroes in a random binary string of length $3k$.  It is well known that the probability of this event is as stated in the lemma.
\end{proof}

The first time the vectors $a'$ span $\mathbb{R}^{k+1}$ we can consider the full dimensional lattice, $L^0$ spanned by these vectors.  Hadamard's inequality, saying that the determinant of a matrix is at most the product of the lengths of the row vectors, implies that the determinant of $L^0$ is at most $(k+1)^{(k+1)/2}$. Let $L^t$ be the lattice where we have added an additional $t$ elements to $A$.  Suppose $L^0$ is defined by the equations $\sum_{j} a_{ij} x_j=0$ modulo $p_i^{m_i}$ for $i=1, \ldots r$.  By the bound on the determinant of $L^0$ we know that $r \leq k \log k$.  We first treat powers of prime different from 2.  

\begin{lemma}
Fix any prime $p > 2$ that divides $det(L^0)$.   The probability that $p$ divides $det(L^ {3k})$ is bounded by $2^{-ck}$ for some constant $c>0$.
\end{lemma}

\begin{proof}
Each time we add a vector to $A$, the probability that $a'$ satisfies any existing non-trivial relation modulo $p$ is bounded by 1/2.  To see this, assume without loss of generality the relations depends on $x-1$. Then, fixing all coordinate of $a$ except the first, only one of the two possible values for $a_1$ makes the added vector satisfy the relation.  We have at most $k-1$ linearly independent equations modulo $p$ satisfied by $L^0$.  Thus for $p$ to divide $det(L^{3k})$ an event that happens with probability 1/2 must happen at least $2k$ times out of $3k$ possibilities.  This is the same event as we considered in the proof of the previous lemma.
\end{proof}

The situation for $p=2$ is slightly different as each vector added has only odd coordinates.  Instead of ruling out relations modulo 2 we need to rule out relations modulo 4 but nothing essential differs.  We omit the details.

We sum up the above reasoning in a theorem.

\begin{theorem}
\label{thm:BLP+AIP random threshold}
Suppose $A$ is selected according to $\randA_{k,p}$.  If  $p \geq Ck 2^{-k}$, then, for a sufficiently large value of $C$,
the probability that BLP+AIP can be used to solve $\fiPCSP(A, \OR)$ is  $2^{-\Omega(k)}$.
\end{theorem}

\section{Concluding Remarks}
\label{sec:conclusions}

We have introduced the notion of promise-usefulness, and studied this in the folded Boolean setting together with the closely related Promise-SAT problem.  While many questions remain on the road to a complete characterization, our results are strong enough to characterize almost all predicates of arity up to five, as well as the vast majority of all predicates of large arity.

The most obvious open problem is of course to obtain a complete classification of promise-usefulness and Promise-SAT (a classification of the latter yields a classification of the former, but not necessarily the other way around), and beyond that Boolean PCSPs in general.

While we only studied the folded setting in this paper, the non-folded setting also warrants investigation.  When it comes to the distinction between the idempotent and non-idempotent setting, the current techniques do not distinguish between them, and it would be interesting to understand if this is true in general or not (see discussion in \cref{sec:non-idempotent OR}).

The algorithmic results in this paper are applications of the known characterization of the BLP+AIP algorithm.  A key question here is of course whether there are tractable cases of Promise-SAT (or again, more generally Boolean PCSP) which are not solvable by BLP+AIP.  Less ambitiously, it is not clear if there are examples of Promise-SAT problems solvable by BLP+AIP but requiring other block-symmetric polymorphisms than one of the five families used in this paper.  If such examples turn out not to exist, then the simple sufficient conditions of \cref{thm:useful_condition} are an exact characterization of what can be shown promise-useful using BLP+AIP.  The evidence for such a nice state of affairs is not strong at this point but it is a possibility.

The hardness results are more involved, and we are quite certain that these can be both improved and simplified.  Among the four conditions for a minion having small fixing assignments, \cref{match+ADA,invmatch+ADA,thm:unate+ADA,thm:split}, the first three are in our opinion relatively natural.  However, the fourth result based on UnCADA and UnDADA, \cref{thm:split}, is different. It was specifically tailored to show the hardness of $\fiPCSP(\{\bstr{0011}, \bstr{0101}, \bstr{0110}, \bstr{1000}, \bstr{1001}\}, \OR)$ as mentioned in \cref{example:UnDADA}. The issue with \cref{thm:split} is that the two families of polymorphisms involved, UnCADA and UnDADA, are somewhat ad-hoc, and from a practical viewpoint they have relatively high arity making them computationally expensive to test for.  It is highly doubtful that the route taken here is the right approach and there may exist a much simpler underlying hardness condition that can replace \cref{thm:split}.

In the asymptotic regime, it would be interesting to determine the threshold on the number of satisfying assignments (as a function of the arity $k$) where most predicates become useless.  There is a rather large gap between the upper bound $\approx 2^{5k/6}$ given by \cref{random useless}, and the lower bound of $\approx k$.

\subsection{Concrete Predicates}

It would be interesting to classify also the small number of unknown predicates for $k=5$.  Recall that, as stated in \cref{sec:tractcondition}, we have established that BLP+AIP is not sufficient to prove tractability for these.  However, while our concrete hardness conditions are not sufficient to establish NP-hardness, it is conceivable that all remaining unknown predicates have small fixing assignments and are thus NP-hard by \cref{thm:fixing} -- we have not found any families of polymorphisms with arbitrarily large minimal fixing assignments for any of the unknown predicates.  Being less ambitious, one can try to classify some of the concrete remaining individual predicates of arity five.  Let us highlight a few of them that we consider of particular interest.

\begin{enumerate}
\item
Consider the predicate $A$ that accepts the $10$ cyclic strings of weight $2$ or $3$ where all ones (and hence all zeroes) are consecutive. In other words, 
\begin{align*}
    A &= \{\bstr{00011}, \bstr{00110}, \bstr{01100}, \bstr{11000},\bstr{10001},\\
&\hphantom{=\,\,\{}\bstr{11100}, \bstr{11001}, \bstr{10011}, \bstr{00111}, \bstr{01110}\}.
\end{align*}
This is a nice symmetric predicate, closed under both cyclic shifts and negations.
In our tables in \cref{sec:fiPCSP unknown list} listing minimal and maximal unknown predicates for $k=5$, the representative of $A$'s equivalence class is the 25th predicate in \cref{tab:5easiest_unknown} and the 16th predicate in \cref{tab:5hardest_unknown}.  Notably, it is both a minimal and a maximal unknown predicate, meaning that removing any assignment from $A$ makes $\fiPCSP(A, \OR)$ tractable, and adding any assignment to $A$ makes it NP-hard. 

This predicate is perhaps even more interesting and natural in the non-folded non-idempotent setting.  In particular $\PCSP(A, \NAE)$ is a natural discrepancy-type hypergraph coloring problem: given an ordered hypergraph which has a $2$-coloring with discrepancy $1$ \emph{where all red vertices are consecutive within an edge}, the objective is to find a normal $2$-coloring.  This can be compared with the early results on $\PCSP$s \cite{2pluseps}, establishing that this problem without the guarantee on red vertices being consecutive is NP-hard (corresponding to the setting when $A$ is all strings of weight $2$ and $3$).

\item
Consider the predicate $A = \{\bstr{00111}, \bstr{01011}, \bstr{01100}, \bstr{10001}, \bstr{10010}, \bstr{10100}\}$ (predicate number 23 in \cref{tab:5easiest_unknown}.  Unlike the previous example this has no clear structure or symmetries, but it is still interesting for a few reasons: (i) $\fiPCSP(A, \OR)$ has a quite rich set of polymorphisms.  For example, it admits $\AT_{13}$ (but not $\AT_{15}$), which does not have fixing assignments of size $6$.  This is in contrast to the NP-hard predicates of arity five, where we are typically able to bound the size of fixing assignments by $4$ or less.  This suggests that some new ideas might be needed to handle it.  (ii) Almost half ($93$ out of $189$) of the remaining unknown predicates of arity $5$ are supersets of $A$, so if $\fiPCSP(A, \OR)$ is NP-hard this alone would reduce the number of unknown predicates significantly.

This predicate and several other remaining ones in fact even have $(7,8)$-block-symmetric polymorphisms.  The predicates of \cref{tab:5easiest_unknown} that have this property are predicates number 8, 13, 16, 21, 23, and 24.

%(7, 9, 10, 12, 16, 17) \\
%(5, 7, 11, 12, 17, 18) \\
%(7, 11, 12, 13, 17, 18) \\
%(7, 11, 12, 17, 18, 19) \\
%(7, 11, 12, 17, 18, 20) <- This is the predicate with AT13
%(3, 13, 14, 19, 20, 24) \\

\end{enumerate}

\bibliographystyle{alpha}
\bibliography{bib}

\appendix

\section{Computational Aspects of The Conditions}
\label{sec:computational aspects}

In this section we discuss some practical details of how to check for the various conditions described in the previous sections for concrete predicates.  We first briefly discuss testing applicability of the BLP+AIP algorithm and then discuss the various hardness conditions in more detail. An implementation of the methods described in this section are available at \protect\href{https://github.com/bjorn-martinsson/On-the-Usefulness-of-Promises}{github.com/bjorn-martinsson/On-the-Usefulness-of-Promises}.

\subsection{Testing Applicability of BLP+AIP}\label{sec:tractcondition}

Our tractability condition is to check whether $\fiPCSP(A, \OR)$ can be solved using the BLP + AIP algorithm.  Here we use \cref{thm:blpaffine} and hence we need to study the possible existence of (block-)symmetric polymorphisms of arbitrary large arity. We do this in two steps.

We first test if $\Pol(\fiPCSP(A, \OR))$ contains infinitely many functions from any of the five families of (block-)symmetric functions from \cref{lemma:easy} -- $\MajFam, \ParFam, \ATFam, \idneg{\MajFam}, \idneg{\ParFam}$. In the case of $\MajFam, \ATFam, \idneg{\MajFam}$, this test can be done using an LP-solver by finding the separating hyperplanes described in \cref{lemma:test_maj,lemma:test_AT,lemma:condition_mino}. In the case of $\ParFam$ and $\idneg{\ParFam}$, the affine conditions given by \cref{lemma:test_odd,lemma:condition_pario} can be checked efficiently using basic bit operations on integers.

If this test fails, then we test whether $\Pol(\fiPCSP(A, \OR))$ does not contain large block-symmetric polymorphisms of other kinds.  According to \cref{thm:blpaffine} this is equivalent to $\Pol(\fiPCSP(A, \OR))$ not containing an $(\ell, \ell+1)$-block-symmetric polymorphism for some $\ell$. To use this condition we keep incrementing $\ell$ until we can establish that no $(\ell, \ell+1)$-block-symmetric polymorphisms exist. 
In principle this process might not even terminate (indeed, as recently shown by Larrauri \cite{LarrauriPCSPSearch}, checking whether BLP+AIP solves an arbitrary PCSP template is in general undecidable), but it turns out that for $k \leq 5$ this search terminates relatively quickly for all remaining predicates.
Even so this second step is relatively time-consuming, but in our case when analyzing all predicates of a given arity, we can be speed it up significantly by only performing it for (inclusion-wise) minimal predicates that are not resolved already in the first step.

So following these steps, we can completely determine for which predicates $A$ of arity at most five $\fiPCSP(A, \OR)$ can be solved using the BLP+AIP algorithm (and as a by-product we also establish that all such predicates can be handled using the five families of block-symmetric functions described above).

\subsection{Ruling out Restricted Polymorphisms} \label{sec:rulingout}

In general for a $k$-ary predicate $A$, the set of $\ell$-ary polymorphisms $\Pol(\fiPCSP(A, \OR))$ is the solution set of a large $k$-SAT formula , where each matrix $M \in A^\ell$ gives rise to a clause $f(M_1) \vee f(M_2) \vee \ldots \vee f(M_k)$, and we additionally require $f(x) = \neg f(x)$ (folded) and $f(0^k) = 0$ (idempotent).  All of our hardness conditions boil down to proving that there are no $\ell$-ary polymorphisms that take certain predefined values at a given list of inputs.  In principle this can then simply be checked by generating the aforementioned SAT formula, fixing the predefined values, and then asking a SAT solver whether the formula is satisfiable.  In practice the formulas we are interested in are quite large (sometimes on the order billions of clauses) and it is essential to use the predefined values already when generating the SAT formula to avoid generating already satisfied clauses.

While these SAT formulas in general have a very large number of variables and clauses and could be computationally infeasible to analyze, it turns out that they are in practice often very easy.  For the specific forbidden functions in our hardness conditions, it surprisingly turns out that it is often possible to rule them out using simple unit propagation of the restricted SAT formulas.

Several hardness conditions also involve the existence of certain functions $f \in \Pol(\fiPCSP(A, \OR))^0$ (the set of functions obtained by taking a polymorphism and fixing some number of variables to $0$).  Any such function $f$ of arity $\ell$ can be obtained from an actual polymorphism $\tilde{f}$ of arity $\ell+1$ by fixing a single bit to $0$.   I.e., $\Pol(\fiPCSP(A, \OR))^0$ contains a function $f$ satisfying $f(x)=y$ for all predetermined values $(x,y) \in P$ if and only if $\Pol(\fiPCSP(A, \OR))$ contains a function $\tilde{f}$ satisfying $f(0x)=y$ for all $(x,y) \in P$.

\subsection{Sufficient parameters for small arities}

In practice it can be difficult to find obstructions for polymorphisms of large arity because the SAT formulas become very large. However, it is not always needed to consider high arity polymorphisms. To illustrate this, we provide a table, \cref{tab:parameters}, listing sufficient values of the parameter $t$ used in hardness conditions given by Theorems \ref{match+ADA}, \ref{invmatch+ADA}, \ref{thm:unate+ADA} and \ref{thm:split}. The table contains the parameters that are sufficient to show that at least one of the hardness conditions, Theorems \ref{match+ADA}, \ref{invmatch+ADA}, \ref{thm:unate+ADA} and \ref{thm:split}, can be applied.

\begin{longtblr}[
  caption={Given any predicate $A$ that we have identified is NP-hard using Theorems \ref{match+ADA}, \ref{invmatch+ADA}, \ref{thm:unate+ADA} and \ref{thm:split}, this table contains the smallest values of $t$ that we have identified to be sufficient to apply at least one of theorems.},
  label={tab:parameters},
]{
  colspec = {l|ccccc},
  rowhead = 1,
  row{odd} = {blue9},
  row{1} = {gray9},
  columns = {colsep=4pt},
} 
     & Matching & Inv. Matching & $t$-ADA-free & $t$-UnCADA-free & $t$-UnDADA-free \\
\hline
    $k=3$ & $1$ & $2$ & $2$ & $2$ & $3$ \\
    $k=4$ & $3$ & $2$ & $3$ & $3$ & $4$ \\
    $k=5$ & $3$ & $3$ & $5$ & $4$ & $4$\\
\end{longtblr}

%\begin{longtblr}[
%  caption={The maximum possible value of the parameter $t$ taken over all predicates that are not tractable using BLP+affine. An inequality is used in the cases where we do not know what the maximum is.},
%  label={tab:parameters2},
%]{
%  colspec = {l|ccccc},
%  rowhead = 1,
%  row{odd} = {blue9},
%  row{1} = {gray9},
%  columns = {colsep=4pt},
%} 
%     & Matching & Inv. Matching & $t$-ADA-free & $t$-UnCADA-free & $t$-UnDADA-free \\
%\hline
%    $k=3$ & $2$ & $2$ & $2$ & $2$ & $3$ \\
%    $k=4$ & $3$ & $3$ & $3$ & $3$ & $4$ \\
%    $k=5$ & $5$ & $4$ & $5$ & $4$ & $5$\\
%\end{longtblr}

Out of the different conditions given in the table, the most difficult one to handle computationally is UnCADA. When $k=5$ and $t=4$, the polymorphisms in question have arity $11$. In our experience, identifying all predicates that are $4$-UnCADA-free using a Python script on a laptop is instantaneous for $k=3$, requires several minutes for $k=4$, and takes several days for $k=5$.

Another parameter of interest is the largest value of $\ell$ such that $\Pol(\fiPCSP(A, \OR))$ contains a $(\ell, \ell+1 )$-block-symmetric polymorphism but no $(\ell +1 ,\ell +2)$-block-symmetric polymorphism. Or in other words, what is the maximum value of $\ell$ for which $\fiPCSP(A, \OR)$ contains at least one $(\ell, \ell+1)$-block-symmetric polymorphism but where the $\PCSP$ is not solvable by BLP+AIP. We verified numerically, by an exhaustive search, that if $k=3$ then $\ell=1$, if $k=4$ then $\ell=3$, and if $k=5$ then $\ell=7$.

\subsection{Unate Minions}

The condition of a minion being unate is easy to check because of the following lemma.

\begin{lemma}
    A minion $\Minion$ contains only unate functions if and only if $\Minion$ does not contain a function $f$ of arity $5$ satisfying
    \begin{align*}
    f(00011) &= 0  & f(00101) = 1 \\
    f(10011) &= 1  & f(10101) = 0
    \end{align*}
\end{lemma}

\begin{proof}  Suppose $g \in \Minion$ is not unate and let $a$ and $b$ witness this for some variable $i$.  Partition all variables $j \ne i$ into four groups depending on the values of $a_j$ and $b_j$.  Create a minor $f$ of $g$ by identifying all variables in the same group with new variables and keeping $x_i$ as its own variable.  It is easy to check that with an appropriate ordering of the five variables ($x_i$ becomes the first variable) this function $f$ satisfies the constraints above.
\end{proof}

\subsection{Bounded Matching and Inverted Matching}

For bounded matchings, we base our test on the observation that having a matching  family of size $t$ is in fact determined by having a polymorphism satisfying a condition that only depends on $t$.

\begin{lemma} \label{lemma:test_matching}
    All functions in a minion $\Minion$ have matching number $\leq t$ if and only if $\Minion$ does not contain a $(t+2)$-ary function $f$ such that $f(\{i\}) = 1$ for all $1 \le i \le t+1$.
\end{lemma}

\begin{proof}
    If all $f \in \Minion$ have matching number $\leq t$ then clearly $\Minion$ does not contain such a function since it would have matching number $t+1$.

    Conversely, assume $\Minion$ does contains an $\ell$-arity function $g$ with matching number $t+1$, where the matching consists of the sets $S_1,\ldots,S_{t+1}$. We can without loss of generality assume that $\bigcup_i S_i \subsetneq [\ell]$ since otherwise we could extend $f$ to arity $\ell +1$ by adding an extra unused variable to $f$. 

    Now the $(t+2)$-ary minor $f$ of $g$ obtained by identifying all variables inside each $S_i$, as well as the variables in $[\ell] - (S_1 \cup \ldots \cup S_{t+1})$ is a function $f$ of the prescribed type.
\end{proof}

\begin{remark}
\cref{lemma:test_matching} gives a method to test whether all $f \in \Minion$ have matching number $t$, but it does not say anything about how large $t$ can be for our polymorphism minions of interest coming from Promise-SAT problems.  Clearly if the block-symmetric functions $\Par$ and $\AT$ exist of arbitrarily large arity then $\Minion$ does not have bounded matching number, but this is not particularly interesting since the existence of these already guarantee tractability of the underlying Promise-SAT problem anyway.  A more useful property, which is applicable also in non-tractable cases, is whether $\Minion^0$ contains arbitrarily large $\OR$ functions -- it is easy to see that if this is the case then $\Minion$ does not have bounded matching number (and also does not have bounded fixing sets).  Somewhat surprisingly, it turns out that for arity up to five, all $\fiPCSP(A, \OR)$ problems that are not tractable by BLP+AIP either have matching number at most $5$, or have arbitrarily large $\OR$ functions in $\Minion^0$ and therefore unbounded matching number. 
\end{remark}

Analogously we have a similar characterization of inverted matching number.

\begin{lemma} \label{lemma:test_inv_matching}
    All functions in a minion $\Minion$ have inverted matching number $\leq t$ if and only if $\Minion$ does not contain a $(t+3)$-ary function $f$ such that $f(\{t+3\})=1$ and $f(\{i,t+3\})=0$ for all $1 \le i \le t+1$.
\end{lemma}

The proof is identical to the preceding lemma and we omit it here.

\subsection{The obstructions of \texorpdfstring{$\AND_t \in \Minion^0$}{AND\_t in M\^{}0}}

In \cref{sec:hard} there are two types of hardness conditions, those based on $\AND_t \not \in \Minion^0$ for some $t$,
\cref{lemma:matching,lemma:invmatch,lemma:unate+g}, and the strictly stronger results based on $\Minion^0$ being $t$-ADA-free for some $t$, \cref{match+ADA,,invmatch+ADA,,thm:unate+ADA,,thm:split}. One of the reasons why we choose to introduce hardness conditions based on $\AND_t \not \in \Minion^0$, despite these results being strictly weaker than those derived from being $t$-ADA-free, is that it is straightforward to determine if $\AND_t \not \in \Pol(\fiPCSP(A, \OR))^0$ for some $t$.
% \cref{claim:AND_free_ADA_free}

Let us see how to construct an obstruction matrix for $\AND_t \in \Pol(\fiPCSP(A, \OR))^0$. This is a matrix $M \in A^{t+1}$ where the first column corresponds to the variables fixed to $0$ as discussed at the end of \cref{sec:rulingout}. Using the fact that the polymorphisms of $\fiPCSP(A, \OR)$ are folded, the matrix is an obstruction of $\AND_t \in \Pol(\fiPCSP(A, \OR))^0$ if and only if
\begin{itemize}
        \item for every row $i$ that starts with a $0$, $\AND_t(M_i^2,\ldots,M_i^{t+1}) = 0$.
        \item for every row $i$ that starts with a $1$, $\negate{\AND}_t(\neg M_i^2,\ldots, \neg M_i^{t+1}) = 0$. Or in other words, $M_i^2 = \ldots = M_i^{t+1} = 0$.
\end{itemize}

Start by fixing the first column of $M$. The second condition restricts the set of columns to choose from to the subset of $A$ with only zeroes in the positions where the first column is one.  For each such element, $a$, let $S_a$ denote the set of coordinates where $a$ is zero.  If the union of these sets does not cover the set of rows where the first column is zero, then it is not possible to construct an obstruction matrix with this first column. If the union of the $S_a$ cover all rows where the first columns has a zero, then this coverage can be used to construct an obstruction matrix if $t \geq k - 1$. Note that we need at most $k-1$ columns to cover all rows that start with a zero \footnote{As $a$ is not $0^k$, only at most $k-1$ rows start with a
zero.}. This yields an obstruction matrix for $\AND_{k-1}$. Appending additional columns to this matrix gives an obstruction matrix for larger $t$.  We summarize this in a lemma.

\begin{lemma} \label{lemma:test_mon}
    We have $\AND_{k-1} \in \Pol(\fiPCSP(A, B))^0$ if and only if $\AND_{t} \in \Pol(\fiPCSP(A, B))^0$ for every $t \geq k - 1$.
\end{lemma}

\subsection{The obstructions of \texorpdfstring{$\xNOR_t \in \Minion^0$}{ANDNOR\_t in M\^{}0}}

In order to be able to apply hardness conditions given by \cref{lemma:unate+g,thm:unate+ADA} for unate polymorphisms, we need to be able to tell if there exists a $t$ such that $\xNOR_t \not \in \Pol(\fiPCSP(A, \OR))^0$. The obstruction matrices of $\xNOR_t$ can be constructed in a similar manner to those of $\AND_t$. Using the fact that the polymorphisms of $\fiPCSP(A, \OR)$ are folded, a matrix $M \in A^{t+1}$ is an obstruction matrix of $\xNOR_t \in \Pol(\fiPCSP(A, \OR))^0$ if and only if
\begin{itemize}
        \item for every row $i$ that starts with a $0$, $\xNOR_t(M_i^2,\ldots,M_i^{t+1}) = 0$.
        \item for every row $i$ that starts with a $1$, $\negate{\xNOR}_t(\neg M_i^2,\ldots, \neg M_i^{t+1}) = 0$.
\end{itemize}
This equates to the following conditions.
\begin{itemize}
        \item No row starts with $11$.
        \item For every row $i$ that starts with $01$, $\OR_{t-1}(M_i^3,\ldots,M_i^{t+1}) = 1$.
        \item for every row $i$ that starts with $10$, $M_i^3 = \ldots = M_i^{t+1} = 1$.
\end{itemize}

Analogously as with $\AND$, we start by fixing the first two columns such that no row starts with $11$. The third condition then restricts the set of columns to choose from to the subset of $A$ with only ones in the positions where the first two columns are $10$. For each such element $a$, let $S_a$ be the set of coordinates where $a$ is $1$. If the union the $S_a$ does not cover the set of all rows that start with $01$, then it is impossible to satisfy the second condition, implying that this choice of the first two columns cannot yield an obstruction matrix. But if the union does cover all rows, then we can construct an obstruction matrix. Note that we need at most $k-1$ columns to cover to cover all the rows that start with $01$ since there are at most $k-1$ rows to be covered. This yields an obstruction matrix for $\xNOR_{k}$. Appending additional columns results in an obstruction matrix for larger $t$. This is summarized in the following lemma, which is an analogue of \cref{lemma:test_mon}.

\begin{lemma} \label{lemma:test_gmon}
    We have $\xNOR_{k} \in \Pol(\fiPCSP(A, B))^0$ if and only if $\xNOR_{t} \in \Pol(\fiPCSP(A, B))^0$ for every $t \geq k$.
\end{lemma}

\subsection{Monotonicity of ADA-free minions}

When it comes to ADAs, these are more complicated to rule out.  But we can at least establish some basic monotonicity properties which reduce the amount of different $(c,d)$-ADAs we have to forbid, and establishing that a $t$-ADA-free minion must also be $(t+1)$-ADA-free.

\begin{claim}
    Let $\Minion$ be a minion.  If $\Minion$ does not contain a $(c,d)$-ADA for some $c,d \ge 1$, then it also does not contain a $(c+1,d)$-ADA.
\end{claim}

\begin{proof}
    The contrapositive is easily proved: if $f$ is a $(c+1,d)$-ADA then the minor obtained by identifying two $y$-variables is a $(c,d)$-ADA.
\end{proof}

\begin{claim}
    Let $\Minion$ be a minion.  If $\Minion$ does not contain a $(c,d)$-ADA for some $d \ge c \ge 1$, then it also does not contain a $(c,d+1)$-ADA.
\end{claim}

\begin{proof}
    To prove the contrapositive, assume $\Minion$ contains a $(c,d+1)$-ADA $f: \{0,1\}^{d+1} \times \{0,1\}^{c} \times \{0,1\}^{d+1}$.  Consider the minor $g: \{0,1\}^d \times \{0,1\}^c \times \{0,1\}^d$ obtained by identifying $x_{d}$ with $x_{d+1}$ and $z_d$ with $z_{d+1}$, i.e.,
    \[
    g(x,y,z) = f(x',y',z')
    \]
    where $x' = (x_1, \ldots, x_d, x_d)$, $y'=y$, and $z' = (z_1, \ldots, z_d, z_d)$.  We claim that $g$ is a $(c,d)$-ADA.  Let us verify the properties.
    \begin{enumerate}
    \item[(a)] $g(1^d, 1^c, 0^d) = g(0^d, 1^c, 1^d) = 1$ by construction.
    \item[(b)] $g(x, y, z) = f(x',y',z') =0$ if $w(x)+w(y)+w(z) < c+d$.  Note that if either $x_d=0$ or $z_d=0$, then $w(x')+w(y')+w(z') \le w(x)+w(y)+w(z)+1 < c+d+1$ so $f(x',y',z')=0$ due to $f$ being a $(c,d+1)$-ADA.  If $x_d=z_d=1$ then $w(x')+w(y')+w(z')=w(x)+w(y)+w(z)+2$.  If $w(x)+w(y)+w(z) < c+d-1$ we are again done, so the remaining case is $x_d=z_d=1$ and $w(x')+w(y')+w(z')=c+d+1$.  Note that since $d \ge c$, not both of $x'$ and $z'$ can equal $1^{d+1}$.  But since $x'$ and $z'$ are both non-zero, this implies that at most one of $x',y',z'$ is all-ones, so $f(x',y',z')=0$ as desired since $f$ is an ADA.
    \item[(c)] $g(x, y, z) = f(x',y',z') = 0$ unless at least two of $x, y, z$ are all-$1$s is easily verified by construction.
    \end{enumerate}
\end{proof}

\subsection{Monotonicity of UnCADA-free and UnDADA-free minions}

The two families UnCADA (Unate Controlled Approximate Double-AND) and UnDADA (Unate Double-controlled Approximate Double-AND) are part of the hardness condition of \cref{thm:split}. While their definitions arise naturally from the inductive proof of \cref{thm:split}, they are nonetheless quite intricate. Some monotonic properties of being UnCADA-free and UnDADA-free follows almost directly from their definitions. Such as $\Minion$ not containing $(c,d)$-UnCADA implying that $\Minion$ does not contain a $(c+1,d)$-UnCADA and $\Minion$ does not contain a $t$-UnDADA implying that $\Minion$ does not contain a $(t+1)$-UnDADA. 

\begin{claim}
\label{split inc v}
    Let $\Minion$ be a minion. If $\Minion$ does not contain a $(c,d)$-UnCADA for some $c,d \ge 1$, then it also does not contain a $(c+1,d)$-UnCADA.
\end{claim}

\begin{proof}
    The contrapositive is easily proved: if $f: \{0,1\}^{(c+1)+2d+1} \times \{0,1\}^3 \rightarrow \{0,1\}$ is a $(c+1,d)$-UnCADA then the minor obtained from $f(x,y)$ by identifying $x_{d+1}$ and $x_{d+2}$ is a $(c,d)$-UnCADA.
\end{proof}

\begin{claim}
\label{doublesplit inc t}
    Let $\Minion$ be a minion.  If $\Minion$ does not contain a $t$-UnDADA for some $t \ge 3$, then it also does not contain a $(t+1)$-UnDADA.
\end{claim}

\begin{proof}
    According to the definition of UnDADA, \cref{def:UnDADA}, a function $f: \{0,1\}^t \times \{0,1\}^4 \rightarrow \{0,1\}$ is a $t$-UnDADA if and only if
    \begin{enumerate}
        \item[(a)] $f(1^{t-1}0, 0011) = f(01^{t-1}, 1001) = 1$
        \item[(b)] for every $x \in \{0,1\}^t$ and $y \in \{0,1\}^3$ such that $w(x) \le t-1$ and $w(y) \ge 2$, it holds that $f(x, y1) = 0$
    \end{enumerate}

    From these conditions it follows that identifying the second and third variable of $f$ results in minor that is a $(t-1)$-UnDADA, assuming $t\geq4$. This shows that if $\Minion$ does not contain a $t$-UnDADA for some $t\geq3$, then it also does not contain a $(t+1)$-UnDADA.
\end{proof}

The following lemma establishes that if $\Minion$ is $t$-UnCADA-free, then it is also $(t+1)$-UnCADA-free, a fact that is not immediately evident from its definition.

\begin{lemma}
    Let $\Minion$ be a minion.  If $\Minion$ is $t$-UnCADA-free for some $t \geq 2$, then it is also $(t+1)$-UnCADA-free.
\end{lemma}

\begin{proof}
    By \cref{split inc v}, the assumption that $\Minion$ is $t$-UnCADA-free implies that it does not contain a $(c+1,t-c)$-UnCADA for any $1 \le c \le t-1$, so it remains to prove that $\Minion$ does not contain a $(1, t)$-UnCADA.

    Suppose for contradiction that $\Minion$ contains a $(1,t)$-UnCADA $f: \{0,1\}^{2t+2} \times \{0,1\}^3$.  Consider the minor $g(x,y)$ of $f(x',y)$ obtained by identifying $x'_t$, $x'_{t+1}$, and $x'_{t+2}$.  The function $g$ satisfies the following properties:
    \begin{enumerate}
        \item[(a)] $g(1^{t-1}10^{t-1}0, 011) = f(1^{t}110^{t-1}0, 011) \ge f(1^t10^t0,011) = 1$ (since $f$ is positive in $x$), and analogously\\
        $g(0^{t-1}11^{t-1}0, 101) = 1$.

        \item[(b)] If $w(x) < t-1$, or if $w(x) < t+1$ and $x_t=0$, then $g(x0, y1) = 0$ for all $y$ of weight $w(y) \ge 1$.  This follows since such an input corresponds to an input $x',y$ for $f$ such that $w(x') < t+1$ and $w(y) \ge 1$ and hence $f(x'0, y1)=0$ since $f$ is a $(1,t)$-UnCADA.
    \end{enumerate}
    Thus, since $g$ cannot be a $(t-1,1)$-UnCADA (by the assumption that $\Minion$ is $t$-UnCADA-free), there must $x'$ of $w(x')=t+1$ such that $x'_t=x'_{t+1}=x'_{t+2}=1$ and $y$ with $w(y) \ge 1$, such that $f(x'0, y1)=1$.  Without loss of generality we may assume $w(y)=1$ since $f$ is negative in $y$.  Suppose $y=01$ (the other case $y=10$ is symmetric).  Then $f$ satisfies
    $f(x'0, 011) = 1$ and $f(1^t10^t, 101)=1$ since $f$ is a $(1,t)$-UnCADA.  Since $w(x')=t+1$ and $x_t=x_{t+1}=1$, there must be some $t+2 \le i \le 2t+1$ such that $x'_i=0$.  Let $I$ be the set of all such $i$, and consider the minor $h$ of $f$ obtained by identifying all coordinates of $I$ with the last $x$-variable.  We claim that, after applying an appropriate reordering of variables, $h$ is a $(|I|+1, t-|I|)$-UnCADA:
    \begin{enumerate}
        \item[(a)] There are two assignments $x^{(1)}$ and $x^{(2)}$ of weight $t+1$ and overlap $|I|+1$ such that $h(x^{(1)}0, 011) = h(x^{(2)}0, 101) = 1$ -- these are the two assignments $x'$ and $1^t10^t$ with the coordinates of $I$ removed.
        \item[(a)] $h(x0,y1)=0$ for all $x$ of $w(x) \le t$ and $w(y) \ge 1$ follows since our identification of the variables of $I$ with the last $x$-variable is effectively just fixing those variables to $0$.
    \end{enumerate}
    But this is now a contradiction to the property, noted above, that $\Minion$ does not contain a $(c+1,t-c)$-UnCADA for any $1 \le c \le t-1$.
\end{proof}

\section{Lists of Predicates of Arity 5}\label{sec:unkn5}

In this section we list various interesting categories of predicates of arity $5$.

\subsection{Maximal Tractable Predicates for \texorpdfstring{$\fiPCSP(A, \OR)$}{fiPCSP(A, OR)}}
\label{sec:tractable5}

The maximal tractable (as far as we know) predicates for $\fiPCSP(A, \OR)$ of arity $5$ are given in \cref{tab:5hardest_easy}.   Easy to recognize predicates are parity of three, four or five variables (number 29, 27, and 28).  Of course $\twoSAT$ is present as number 30 as well as the (non-strict) majority of four variables (number 31).  Predicate 32 is the closely related function which is a threshold function where the first coordinate has weight 2 and the other coordinates have weight 1.

\begin{longtblr}[
  caption = {A list of the 32 maximal tractable predicates for $k=5$.  See \cref{tab:3hardest_easy} (\cref{sec:fiPCSP detailed}) for explanation of the last column.},
  label = {tab:5hardest_easy},
]{
  colspec = {l@{}l|cccccr},
  rowhead = 1,
  row{odd} = {blue9},
  row{1} = {gray9},
  columns = {colsep=4pt},
} 
     & Predicate & $\MajFam$ & $\ParFam$ & $\ATFam$ & $\idneg{\MajFam}$ & $\idneg{\ParFam}$ & Dep. \\
\hline
    1. & $\{\bstr{00011}, \bstr{00101}, \bstr{00110}, \bstr{01000}, \bstr{10000}\}$ &  &  & \gcmark &  &  & $1/11$ \\
    2. & \begin{tabular}{@{}l@{}l@{}l@{}l@{}l@{}l}$\{ $&$\bstr{00011},\,$&$ \bstr{00101},\,$&$ \bstr{00110},\,$&$ \bstr{01001},\,$&$ \bstr{01010}, $ \\
 &$ $&$  $&$  $&$ \bstr{01100},\,$&$ \bstr{10000}\}$\end{tabular} &  &  & \gcmark &  &  & $3/21$ \\
    3. & \begin{tabular}{@{}l@{}l@{}l@{}l@{}l@{}l}$\{ $&$\bstr{00011},\,$&$ \bstr{00100},\,$&$ \bstr{00110},\,$&$ \bstr{01000},\,$&$ \bstr{01001}, $ \\
 &$ $&$  $&$  $&$ \bstr{10000},\,$&$ \bstr{10001}\}$\end{tabular} &  &  &  & \gcmark &  & $2/48$ \\
    4. & \begin{tabular}{@{}l@{}l@{}l@{}l@{}l@{}l}$\{ $&$\bstr{00011},\,$&$ \bstr{00101},\,$&$ \bstr{00111},\,$&$ \bstr{01000},\,$&$ \bstr{01001}, $ \\
 &$ $&$  $&$  $&$ \bstr{10000},\,$&$ \bstr{10001}\}$\end{tabular} &  &  &  & \gcmark &  & $1/51$ \\
    5. & $\{\bstr{0011*}, \bstr{0101*}, \bstr{0110*}, \bstr{1000*}\}$ &  &  & \gcmark &  & \gcmark & $12/77$ \\
    6. & \begin{tabular}{@{}l@{}l@{}l@{}l@{}l@{}l}$\{ $&$\bstr{00011},\,$&$ \bstr{00101},\,$&$ \bstr{00111},\,$&$ \bstr{01000},\,$&$ \bstr{01010}, $ \\
 &$ $&$ \bstr{01100},\,$&$ \bstr{01110},\,$&$ \bstr{10000},\,$&$ \bstr{10001}\}$\end{tabular} &  &  &  & \gcmark &  & $8/253$ \\
    7. & \begin{tabular}{@{}l@{}l@{}l@{}l@{}l@{}l}$\{ $&$\bstr{00010},\,$&$ \bstr{00110},\,$&$ \bstr{00111},\,$&$ \bstr{01000},\,$&$ \bstr{01001}, $ \\
 &$ $&$ \bstr{01100},\,$&$ \bstr{01101},\,$&$ \bstr{10001},\,$&$ \bstr{10010}\}$\end{tabular} &  &  &  & \gcmark &  & $4/320$ \\
    8. & \begin{tabular}{@{}l@{}l@{}l@{}l@{}l@{}l}$\{ $&$\bstr{00111},\,$&$ \bstr{01001},\,$&$ \bstr{01010},\,$&$ \bstr{01101},\,$&$ \bstr{01110}, $ \\
 &$ $&$  $&$  $&$ \bstr{10000},\,$&$ \bstr{10011}\}$\end{tabular} &  &  &  &  & \gcmark & $2/67$ \\
    9. & \begin{tabular}{@{}l@{}l@{}l@{}l@{}l@{}l}$\{ $&$\bstr{00011},\,$&$ \bstr{00110},\,$&$ \bstr{00111},\,$&$ \bstr{01000},\,$&$ \bstr{01001}, $ \\
 &$ $&$ \bstr{01100},\,$&$ \bstr{01101},\,$&$ \bstr{10001},\,$&$ \bstr{10011}\}$\end{tabular} &  &  &  & \gcmark &  & $2/342$ \\
    10. & \begin{tabular}{@{}l@{}l@{}l@{}l@{}l@{}l}$\{ $&$\bstr{00100},\,$&$ \bstr{00101},\,$&$ \bstr{00110},\,$&$ \bstr{01010},\,$&$ \bstr{01101}, $ \\
 &$ $&$ \bstr{10000},\,$&$ \bstr{10001},\,$&$ \bstr{10010},\,$&$ \bstr{10011}\}$\end{tabular} &  &  &  & \gcmark &  & $4/330$ \\
    11. & $\{\bstr{0011*}, \bstr{0100*}, \bstr{0110*}, \bstr{1000*}, \bstr{1001*}\}$ &  &  &  & \gcmark &  & $36/449$ \\
    12. & \begin{tabular}{@{}l@{}l@{}l@{}l@{}l@{}l}$\{ $&$\bstr{00101},\,$&$ \bstr{00110},\,$&$ \bstr{01001},\,$&$ \bstr{01010},\,$&$ \bstr{01101}, $ \\
 &$\bstr{01110},\,$&$ \bstr{10000},\,$&$ \bstr{10001},\,$&$ \bstr{10010},\,$&$ \bstr{10011}\}$\end{tabular} &  &  &  & \gcmark &  & $11/320$ \\
    13. & \begin{tabular}{@{}l@{}l@{}l@{}l@{}l@{}l}$\{ $&$\bstr{00110},\,$&$ \bstr{01001},\,$&$ \bstr{01100},\,$&$ \bstr{01101},\,$&$ \bstr{01110}, $ \\
 &$ $&$ \bstr{10000},\,$&$ \bstr{10001},\,$&$ \bstr{10010},\,$&$ \bstr{10011}\}$\end{tabular} &  &  &  & \gcmark &  & $1/227$ \\
    14. & \begin{tabular}{@{}l@{}l@{}l@{}l@{}l@{}l}$\{ $&$\bstr{00010},\,$&$ \bstr{00100},\,$&$ \bstr{01010},\,$&$ \bstr{01011},\,$&$ \bstr{01100}, $ \\
 &$ $&$ \bstr{01101},\,$&$ \bstr{10001},\,$&$ \bstr{10010},\,$&$ \bstr{10100}\}$\end{tabular} &  &  &  & \gcmark &  & $2/233$ \\
    15. & \begin{tabular}{@{}l@{}l@{}l@{}l@{}l@{}l}$\{ $&$\bstr{00011},\,$&$ \bstr{01001},\,$&$ \bstr{01010},\,$&$ \bstr{01101},\,$&$ \bstr{01110}, $ \\
 &$ $&$  $&$  $&$ \bstr{10011},\,$&$ \bstr{10100}\}$\end{tabular} &  &  &  &  & \gcmark & $2/52$ \\
    16. & \begin{tabular}{@{}l@{}l@{}l@{}l@{}l@{}l}$\{ $&$\bstr{00111},\,$&$ \bstr{01001},\,$&$ \bstr{01010},\,$&$ \bstr{01101},\,$&$ \bstr{01110}, $ \\
 &$ $&$  $&$  $&$ \bstr{10011},\,$&$ \bstr{10100}\}$\end{tabular} &  &  &  &  & \gcmark & $1/57$ \\
    17. & \begin{tabular}{@{}l@{}l@{}l@{}l@{}l@{}l}$\{ $&$\bstr{00110},\,$&$ \bstr{01000},\,$&$ \bstr{01001},\,$&$ \bstr{01010},\,$&$ \bstr{01101}, $ \\
 &$ $&$ \bstr{10000},\,$&$ \bstr{10001},\,$&$ \bstr{10011},\,$&$ \bstr{10100}\}$\end{tabular} &  &  &  & \gcmark &  & $3/215$ \\
    18. & \begin{tabular}{@{}l@{}l@{}l@{}l@{}l@{}l}$\{ $&$\bstr{00110},\,$&$ \bstr{01001},\,$&$ \bstr{01011},\,$&$ \bstr{01100},\,$&$ \bstr{01110}, $ \\
 &$ $&$ \bstr{10000},\,$&$ \bstr{10001},\,$&$ \bstr{10011},\,$&$ \bstr{10100}\}$\end{tabular} &  &  &  & \gcmark &  & $4/329$ \\
    19. & \begin{tabular}{@{}l@{}l@{}l@{}l@{}l@{}l}$\{ $&$\bstr{00110},\,$&$ \bstr{01001},\,$&$ \bstr{01010},\,$&$ \bstr{01101},\,$&$ \bstr{01110}, $ \\
 &$ $&$ \bstr{10000},\,$&$ \bstr{10001},\,$&$ \bstr{10011},\,$&$ \bstr{10100}\}$\end{tabular} &  &  &  & \gcmark &  & $2/325$ \\
    20. & \begin{tabular}{@{}l@{}l@{}l@{}l@{}l@{}l}$\{ $&$\bstr{00110},\,$&$ \bstr{01001},\,$&$ \bstr{01100},\,$&$ \bstr{01101},\,$&$ \bstr{01110}, $ \\
 &$ $&$ \bstr{10000},\,$&$ \bstr{10001},\,$&$ \bstr{10011},\,$&$ \bstr{10100}\}$\end{tabular} &  &  &  & \gcmark &  & $2/358$ \\
    21. & \begin{tabular}{@{}l@{}l@{}l@{}l@{}l@{}l}$\{ $&$\bstr{00110},\,$&$ \bstr{01011},\,$&$ \bstr{01101},\,$&$ \bstr{01110},\,$&$ \bstr{10000}, $ \\
 &$ $&$  $&$  $&$ \bstr{10011},\,$&$ \bstr{10101}\}$\end{tabular} &  &  &  &  & \gcmark & $2/63$ \\
    22. & \begin{tabular}{@{}l@{}l@{}l@{}l@{}l@{}l}$\{ $&$\bstr{00101},\,$&$ \bstr{01010},\,$&$ \bstr{01011},\,$&$ \bstr{01100},\,$&$ \bstr{01101}, $ \\
 &$ $&$ \bstr{10001},\,$&$ \bstr{10010},\,$&$ \bstr{10011},\,$&$ \bstr{10101}\}$\end{tabular} &  &  &  & \gcmark &  & $2/184$ \\
    23. & \begin{tabular}{@{}l@{}l@{}l@{}l@{}l@{}l}$\{ $&$\bstr{00101},\,$&$ \bstr{01010},\,$&$ \bstr{01100},\,$&$ \bstr{01101},\,$&$ \bstr{01110}, $ \\
 &$ $&$ \bstr{10001},\,$&$ \bstr{10010},\,$&$ \bstr{10011},\,$&$ \bstr{10101}\}$\end{tabular} &  &  &  & \gcmark &  & $1/127$ \\
    24. & \begin{tabular}{@{}l@{}l@{}l@{}l@{}l@{}l}$\{ $&$\bstr{00011},\,$&$ \bstr{00100},\,$&$ \bstr{01001},\,$&$ \bstr{01011},\,$&$ \bstr{01100}, $ \\
 &$\bstr{01110},\,$&$ \bstr{10010},\,$&$ \bstr{10011},\,$&$ \bstr{10100},\,$&$ \bstr{10101}\}$\end{tabular} &  &  &  & \gcmark &  & $16/446$ \\
    25. & \begin{tabular}{@{}l@{}l@{}l@{}l@{}l@{}l}$\{ $&$\bstr{00011},\,$&$ \bstr{00101},\,$&$ \bstr{00110},\,$&$ \bstr{01001},\,$&$ \bstr{01010}, $ \\
 &$\bstr{01100},\,$&$ \bstr{10001},\,$&$ \bstr{10010},\,$&$ \bstr{10100},\,$&$ \bstr{11000}\}$\end{tabular} &  &  & \gcmark &  &  & $11/33$ \\
    26. & \begin{tabular}{@{}l@{}l@{}l@{}l@{}l@{}l}$\{ $&$\bstr{00101},\,$&$ \bstr{00110},\,$&$ \bstr{01011},\,$&$ \bstr{01110},\,$&$ \bstr{10011}, $ \\
 &$ $&$  $&$  $&$ \bstr{10101},\,$&$ \bstr{11000}\}$\end{tabular} &  &  &  &  & \gcmark & $1/53$ \\
    27. & \begin{tabular}{@{}l@{}l@{}l@{}l@{}l@{}l}$\{ $&$\bstr{0001*},\,$&$ \bstr{0010*},\,$&$ \bstr{0100*},\,$&$ \bstr{0111*},\,$&$ \bstr{1000*}, $ \\
 &$ $&$  $&$ \bstr{1011*},\,$&$ \bstr{1101*},\,$&$ \bstr{1110*}\}$\end{tabular} &  & \gcmark &  &  &  & $2\,388/3\,875$ \\
    28. & \begin{tabular}{@{}l@{}l@{}l@{}l@{}l@{}l}$\{ $&$\bstr{00001},\,$&$ \bstr{00010},\,$&$ \bstr{00100},\,$&$ \bstr{00111},\,$&$ \bstr{01000}, $ \\
 &$\bstr{01011},\,$&$ \bstr{01101},\,$&$ \bstr{01110},\,$&$ \bstr{10000},\,$&$ \bstr{10011}, $ \\
 &$\bstr{10101},\,$&$ \bstr{10110},\,$&$ \bstr{11001},\,$&$ \bstr{11010},\,$&$ \bstr{11100}, $ \\
 &$ $&$  $&$  $&$  $&$ \bstr{11111}\}$\end{tabular} &  & \gcmark &  &  &  & $674/1\,087$ \\
    29. & $\{\bstr{**001}, \bstr{**010}, \bstr{**100}, \bstr{**111}\}$ &  & \gcmark &  &  &  & $4\,313/7\,099$ \\
    30. & $\{\bstr{***01}, \bstr{***10}, \bstr{***11}\}$ & \gcmark &  &  &  &  & $1\,118\,234/1\,249\,651$ \\
    31. & \begin{tabular}{@{}l@{}l@{}l@{}l@{}l@{}l}$\{ $&$\bstr{*0011},\,$&$ \bstr{*0101},\,$&$ \bstr{*0110},\,$&$ \bstr{*0111},\,$&$ \bstr{*1001}, $ \\
 &$\bstr{*1010},\,$&$ \bstr{*1011},\,$&$ \bstr{*1100},\,$&$ \bstr{*1101},\,$&$ \bstr{*1110}, $ \\
 &$ $&$  $&$  $&$  $&$ \bstr{*1111}\}$\end{tabular} & \gcmark &  &  &  &  & $31\,133/157\,103$ \\
    32. & \begin{tabular}{@{}l@{}l@{}l@{}l@{}l@{}l}$\{ $&$\bstr{00011},\,$&$ \bstr{00101},\,$&$ \bstr{00111},\,$&$ \bstr{01001},\,$&$ \bstr{01011}, $ \\
 &$\bstr{01101},\,$&$ \bstr{01110},\,$&$ \bstr{01111},\,$&$ \bstr{10001},\,$&$ \bstr{10011}, $ \\
 &$\bstr{10101},\,$&$ \bstr{10110},\,$&$ \bstr{10111},\,$&$ \bstr{11001},\,$&$ \bstr{11010}, $ \\
 &$\bstr{11011},\,$&$ \bstr{11100},\,$&$ \bstr{11101},\,$&$ \bstr{11110},\,$&$ \bstr{11111}\}$\end{tabular} & \gcmark &  &  &  &  & $355/43\,951$ \\
\end{longtblr}

\subsection{Minimal and Maximal Unknown Predicates for \texorpdfstring{$\fiPCSP(A, \OR)$}{fiPCSP(A, OR)}}
\label{sec:fiPCSP unknown list}

\cref{tab:5easiest_unknown} lists the minimal predicates $A$ of arity $5$ where we have been unable to determine the complexity of $\fiPCSP(A, \OR)$.   \cref{tab:5hardest_unknown} lists the maximal such predicates.

\begin{longtblr}[
  caption = {A list of the 25 minimal unknown predicates for Promise-SAT of arity $k=5$.  A checkmark in the ``unate'' column indicates that all polymorphisms are unate.  A cross in the ADA (resp.~UnCADA or UnDADA) columns indicates that the polymorphism minion is $t$-ADA-free (resp.~$t$-UnCADA-free or $t$-UnDADA-free) for some $t$.  The second value in the ``Dep.'' column gives the number of unknown predicates implied by this predicate (in particular this number of unknown Promise-SAT problems would be NP-hard if this predicate is shown NP-hard), while the first value in the ``Dep.'' column gives the number of such predicates that are not implied by any other predicate in the table.},
  label = {tab:5easiest_unknown},
]{
  colspec = {l@{}l@{}|cccc@{}r},
  rowhead = 1,
  row{odd} = {blue9},
  row{1} = {gray9},
  columns = {colsep=2pt},
} 
     & Predicate & Unate & ADA & UnCADA & UnDADA & Dep. \\
\hline
    1. & $\{\bstr{00011}, \bstr{00111}, \bstr{01001}, \bstr{01010}, \bstr{01100}, \bstr{10000}\}$ & \gcmark & \rxmark &  & \rxmark & $1/2$ \\
    2. & $\{\bstr{00101}, \bstr{00110}, \bstr{00111}, \bstr{01011}, \bstr{01100}, \bstr{10000}\}$ &  & \rxmark & \rxmark & \rxmark & $1/3$ \\
    3. & $\{\bstr{00110}, \bstr{00111}, \bstr{01001}, \bstr{01011}, \bstr{01100}, \bstr{10000}\}$ & \gcmark & \rxmark &  & \rxmark & $1/3$ \\
    4. & $\{\bstr{00011}, \bstr{00110}, \bstr{00111}, \bstr{01011}, \bstr{01101}, \bstr{10000}\}$ &  & \rxmark & \rxmark & \rxmark & $1/8$ \\
    5. & $\{\bstr{00011}, \bstr{00111}, \bstr{01011}, \bstr{01101}, \bstr{01110}, \bstr{10000}\}$ &  & \rxmark &  &  & $1/6$ \\
    6. & $\{\bstr{00110}, \bstr{00111}, \bstr{01001}, \bstr{01010}, \bstr{01100}, \bstr{10001}\}$ & \gcmark & \rxmark &  & \rxmark & $1/4$ \\
    7. & $\{\bstr{00011}, \bstr{00110}, \bstr{01010}, \bstr{01100}, \bstr{10000}, \bstr{10001}\}$ & \gcmark & \rxmark &  &  & $1/2$ \\
    8. & $\{\bstr{00111}, \bstr{01001}, \bstr{01010}, \bstr{01100}, \bstr{10000}, \bstr{10001}\}$ &  & \rxmark &  &  & $3/5$ \\
    9. & $\{\bstr{00101}, \bstr{00111}, \bstr{01011}, \bstr{01100}, \bstr{10000}, \bstr{10001}\}$ & \gcmark & \rxmark &  &  & $1/4$ \\
    10. & $\{\bstr{00111}, \bstr{01011}, \bstr{01100}, \bstr{01110}, \bstr{10000}, \bstr{10001}\}$ &  & \rxmark &  &  & $1/8$ \\
    11. & $\{\bstr{00111}, \bstr{01011}, \bstr{01101}, \bstr{01110}, \bstr{10000}, \bstr{10001}\}$ &  & \rxmark &  &  & $4/7$ \\
    12. & $\{\bstr{00011}, \bstr{00111}, \bstr{01001}, \bstr{01100}, \bstr{10001}, \bstr{10010}\}$ &  & \rxmark &  &  & $2/14$ \\
    13. & $\{\bstr{00101}, \bstr{00111}, \bstr{01011}, \bstr{01100}, \bstr{10001}, \bstr{10010}\}$ & \gcmark & \rxmark &  &  & $3/28$ \\
    14. & $\{\bstr{00110}, \bstr{01010}, \bstr{01100}, \bstr{01101}, \bstr{10001}, \bstr{10010}\}$ &  & \rxmark &  &  & $1/10$ \\
    15. & $\{\bstr{00111}, \bstr{01010}, \bstr{01100}, \bstr{01101}, \bstr{10001}, \bstr{10010}\}$ & \gcmark & \rxmark &  &  & $2/25$ \\
    16. & $\{\bstr{00111}, \bstr{01011}, \bstr{01100}, \bstr{01101}, \bstr{10001}, \bstr{10010}\}$ &  & \rxmark &  &  & $9/81$ \\
    17. & $\{\bstr{00111}, \bstr{01011}, \bstr{01100}, \bstr{10000}, \bstr{10001}, \bstr{10010}\}$ &  & \rxmark &  &  & $1/5$ \\
    18. & $\{\bstr{00111}, \bstr{01010}, \bstr{01100}, \bstr{01101}, \bstr{10000}, \bstr{10011}\}$ &  & \rxmark & \rxmark & \rxmark & $1/4$ \\
    19. & $\{\bstr{00111}, \bstr{01010}, \bstr{01101}, \bstr{10000}, \bstr{10001}, \bstr{10011}\}$ &  & \rxmark & \rxmark & \rxmark & $1/4$ \\
    20. & \begin{tabular}{@{}l@{}l@{}l@{}l@{}l@{}l@{}l}$\{ $&$\bstr{00111},\,$&$ \bstr{01001},\,$&$ \bstr{01010},\,$&$ \bstr{01100},\,$&$ \bstr{10001},\,$&$ \bstr{10010}, $ \\
 &$ $&$  $&$  $&$  $&$  $&$ \bstr{10011}\}$\end{tabular} &  & \rxmark &  &  & $2/3$ \\
    21. & $\{\bstr{00111}, \bstr{01011}, \bstr{01100}, \bstr{10001}, \bstr{10010}, \bstr{10011}\}$ &  & \rxmark &  &  & $4/57$ \\
    22. & $\{\bstr{00101}, \bstr{01011}, \bstr{01100}, \bstr{10001}, \bstr{10010}, \bstr{10100}\}$ & \gcmark & \rxmark &  &  & $3/21$ \\
    23. & $\{\bstr{00111}, \bstr{01011}, \bstr{01100}, \bstr{10001}, \bstr{10010}, \bstr{10100}\}$ & \gcmark & \rxmark &  &  & $10/93$ \\
    24. & $\{\bstr{00011}, \bstr{01101}, \bstr{01110}, \bstr{10011}, \bstr{10100}, \bstr{11000}\}$ & \gcmark & \rxmark &  & \rxmark & $4/18$ \\
    25. & \begin{tabular}{@{}l@{}l@{}l@{}l@{}l@{}l@{}l}$\{ $&$\bstr{00110},\,$&$ \bstr{01001},\,$&$ \bstr{01100},\,$&$ \bstr{01101},\,$&$ \bstr{01110},\,$&$ \bstr{10001}, $ \\
 &$ $&$  $&$ \bstr{10010},\,$&$ \bstr{10011},\,$&$ \bstr{10110},\,$&$ \bstr{11001}\}$\end{tabular} &  &  & \rxmark & \rxmark & $1/1$ \\
\end{longtblr} % warning, big table

\begin{longtblr}[
  caption = {A list of the 19 maximal unknown predicates for Promise-SAT of arity $k=5$.  Most columns are as in \cref{tab:5easiest_unknown} but the ``Dep.''~column differs.  The second value in this column now gives the number of unknown predicates that imply this predicate (in particular this number of unknown Promise-SAT problems would be tractable if this predicate is shown tractable), while the first value gives the number of such predicates that do not imply any other predicate in the table.},
  label = {tab:5hardest_unknown},
]{
  colspec = {l@{}l|ccccr},
  rowhead = 1,
  row{odd} = {blue9},
  row{1} = {gray9},
  columns = {colsep=4pt},
} 
     & Predicate & Unate & ADA & UnCADA & UnDADA & Dep. \\
\hline
    1. & \begin{tabular}{@{}l@{}l@{}l@{}l@{}l@{}l}$\{ $&$\bstr{00011},\,$&$ \bstr{00101},\,$&$ \bstr{00111},\,$&$ \bstr{01001},\,$&$ \bstr{01011}, $ \\
 &$ $&$  $&$ \bstr{01101},\,$&$ \bstr{01110},\,$&$ \bstr{10000}\}$\end{tabular} &  & \rxmark & \rxmark & \rxmark & $3/6$ \\
    2. & \begin{tabular}{@{}l@{}l@{}l@{}l@{}l@{}l}$\{ $&$\bstr{00110},\,$&$ \bstr{00111},\,$&$ \bstr{01001},\,$&$ \bstr{01010},\,$&$ \bstr{01100}, $ \\
 &$ $&$  $&$  $&$ \bstr{10000},\,$&$ \bstr{10001}\}$\end{tabular} & \gcmark & \rxmark &  & \rxmark & $3/5$ \\
    3. & \begin{tabular}{@{}l@{}l@{}l@{}l@{}l@{}l}$\{ $&$\bstr{00110},\,$&$ \bstr{00111},\,$&$ \bstr{01010},\,$&$ \bstr{01011},\,$&$ \bstr{01101}, $ \\
 &$ $&$  $&$ \bstr{01110},\,$&$ \bstr{10000},\,$&$ \bstr{10001}\}$\end{tabular} &  & \rxmark & \rxmark & \rxmark & $1/8$ \\
    4. & \begin{tabular}{@{}l@{}l@{}l@{}l@{}l@{}l}$\{ $&$\bstr{00110},\,$&$ \bstr{00111},\,$&$ \bstr{01001},\,$&$ \bstr{01011},\,$&$ \bstr{01100}, $ \\
 &$ $&$  $&$ \bstr{10000},\,$&$ \bstr{10001},\,$&$ \bstr{10010}\}$\end{tabular} & \gcmark & \rxmark &  & \rxmark & $6/8$ \\
    5. & \begin{tabular}{@{}l@{}l@{}l@{}l@{}l@{}l}$\{ $&$\bstr{00111},\,$&$ \bstr{01011},\,$&$ \bstr{01100},\,$&$ \bstr{01101},\,$&$ \bstr{01110}, $ \\
 &$ $&$  $&$ \bstr{10000},\,$&$ \bstr{10001},\,$&$ \bstr{10010}\}$\end{tabular} &  & \rxmark &  &  & $2/10$ \\
    6. & \begin{tabular}{@{}l@{}l@{}l@{}l@{}l@{}l}$\{ $&$\bstr{00111},\,$&$ \bstr{01010},\,$&$ \bstr{01100},\,$&$ \bstr{01101},\,$&$ \bstr{01110}, $ \\
 &$ $&$  $&$ \bstr{10000},\,$&$ \bstr{10001},\,$&$ \bstr{10011}\}$\end{tabular} &  & \rxmark & \rxmark & \rxmark & $6/9$ \\
    7. & \begin{tabular}{@{}l@{}l@{}l@{}l@{}l@{}l}$\{ $&$\bstr{00101},\,$&$ \bstr{00111},\,$&$ \bstr{01011},\,$&$ \bstr{01100},\,$&$ \bstr{10001}, $ \\
 &$ $&$  $&$  $&$ \bstr{10010},\,$&$ \bstr{10011}\}$\end{tabular} & \gcmark & \rxmark &  & \rxmark & $1/3$ \\
    8. & \begin{tabular}{@{}l@{}l@{}l@{}l@{}l@{}l}$\{ $&$\bstr{00111},\,$&$ \bstr{01001},\,$&$ \bstr{01010},\,$&$ \bstr{01100},\,$&$ \bstr{10000}, $ \\
 &$ $&$  $&$ \bstr{10001},\,$&$ \bstr{10010},\,$&$ \bstr{10011}\}$\end{tabular} &  & \rxmark & \rxmark &  & $3/5$ \\
    9. & \begin{tabular}{@{}l@{}l@{}l@{}l@{}l@{}l}$\{ $&$\bstr{00111},\,$&$ \bstr{01011},\,$&$ \bstr{01101},\,$&$ \bstr{01110},\,$&$ \bstr{10000}, $ \\
 &$ $&$  $&$ \bstr{10001},\,$&$ \bstr{10010},\,$&$ \bstr{10011}\}$\end{tabular} &  & \rxmark & \rxmark &  & $2/4$ \\
    10. & \begin{tabular}{@{}l@{}l@{}l@{}l@{}l@{}l}$\{ $&$\bstr{00011},\,$&$ \bstr{00111},\,$&$ \bstr{01001},\,$&$ \bstr{01010},\,$&$ \bstr{01100}, $ \\
 &$ $&$  $&$ \bstr{10001},\,$&$ \bstr{10010},\,$&$ \bstr{10100}\}$\end{tabular} & \gcmark & \rxmark &  & \rxmark & $2/5$ \\
    11. & \begin{tabular}{@{}l@{}l@{}l@{}l@{}l@{}l}$\{ $&$\bstr{00111},\,$&$ \bstr{01001},\,$&$ \bstr{01010},\,$&$ \bstr{01100},\,$&$ \bstr{10001}, $ \\
 &$ $&$  $&$ \bstr{10010},\,$&$ \bstr{10100},\,$&$ \bstr{11000}\}$\end{tabular} & \gcmark & \rxmark &  &  & $1/3$ \\
    12. & \begin{tabular}{@{}l@{}l@{}l@{}l@{}l@{}l}$\{ $&$\bstr{00110},\,$&$ \bstr{00111},\,$&$ \bstr{01010},\,$&$ \bstr{01011},\,$&$ \bstr{01100}, $ \\
 &$\bstr{01101},\,$&$ \bstr{10001},\,$&$ \bstr{10010},\,$&$ \bstr{10100},\,$&$ \bstr{11000}\}$\end{tabular} & \gcmark & \rxmark &  & \rxmark & $15/30$ \\
    13. & \begin{tabular}{@{}l@{}l@{}l@{}l@{}l@{}l}$\{ $&$\bstr{00110},\,$&$ \bstr{01010},\,$&$ \bstr{01101},\,$&$ \bstr{01110},\,$&$ \bstr{10001}, $ \\
 &$ $&$  $&$ \bstr{10011},\,$&$ \bstr{10100},\,$&$ \bstr{11000}\}$\end{tabular} &  & \rxmark & \rxmark & \rxmark & $1/2$ \\
    14. & \begin{tabular}{@{}l@{}l@{}l@{}l@{}l@{}l}$\{ $&$\bstr{00111},\,$&$ \bstr{01001},\,$&$ \bstr{01011},\,$&$ \bstr{01100},\,$&$ \bstr{01110}, $ \\
 &$\bstr{10010},\,$&$ \bstr{10011},\,$&$ \bstr{10100},\,$&$ \bstr{10101},\,$&$ \bstr{11000}\}$\end{tabular} & \gcmark & \rxmark &  & \rxmark & $12/30$ \\
    15. & \begin{tabular}{@{}l@{}l@{}l@{}l@{}l@{}l}$\{ $&$\bstr{00111},\,$&$ \bstr{01001},\,$&$ \bstr{01010},\,$&$ \bstr{01011},\,$&$ \bstr{01100}, $ \\
 &$\bstr{01101},\,$&$ \bstr{01110},\,$&$ \bstr{10001},\,$&$ \bstr{10011},\,$&$ \bstr{10101}, $ \\
 &$ $&$  $&$  $&$ \bstr{10110},\,$&$ \bstr{11000}\}$\end{tabular} & \gcmark & \rxmark &  & \rxmark & $18/71$ \\
    16. & \begin{tabular}{@{}l@{}l@{}l@{}l@{}l@{}l}$\{ $&$\bstr{00110},\,$&$ \bstr{01001},\,$&$ \bstr{01100},\,$&$ \bstr{01101},\,$&$ \bstr{01110}, $ \\
 &$\bstr{10001},\,$&$ \bstr{10010},\,$&$ \bstr{10011},\,$&$ \bstr{10110},\,$&$ \bstr{11001}\}$\end{tabular} &  &  & \rxmark & \rxmark & $1/1$ \\
    17. & \begin{tabular}{@{}l@{}l@{}l@{}l@{}l@{}l}$\{ $&$\bstr{00011},\,$&$ \bstr{00111},\,$&$ \bstr{01010},\,$&$ \bstr{01011},\,$&$ \bstr{01101}, $ \\
 &$\bstr{01110},\,$&$ \bstr{10010},\,$&$ \bstr{10100},\,$&$ \bstr{10110},\,$&$ \bstr{11001}\}$\end{tabular} & \gcmark & \rxmark &  & \rxmark & $4/21$ \\
    18. & \begin{tabular}{@{}l@{}l@{}l@{}l@{}l@{}l}$\{ $&$\bstr{00011},\,$&$ \bstr{00110},\,$&$ \bstr{01010},\,$&$ \bstr{01011},\,$&$ \bstr{01101}, $ \\
 &$\bstr{01110},\,$&$ \bstr{10011},\,$&$ \bstr{10100},\,$&$ \bstr{10110},\,$&$ \bstr{11001}\}$\end{tabular} & \gcmark & \rxmark &  & \rxmark & $6/25$ \\
    19. & \begin{tabular}{@{}l@{}l@{}l@{}l@{}l@{}l}$\{ $&$\bstr{00111},\,$&$ \bstr{01011},\,$&$ \bstr{01100},\,$&$ \bstr{01101},\,$&$ \bstr{01110}, $ \\
 &$\bstr{10001},\,$&$ \bstr{10010},\,$&$ \bstr{10011},\,$&$ \bstr{10100},\,$&$ \bstr{10110}, $ \\
 &$ $&$  $&$  $&$ \bstr{11000},\,$&$ \bstr{11001}\}$\end{tabular} & \gcmark & \rxmark &  & \rxmark & $16/67$ \\
\end{longtblr}

\subsection{Minimal and Maximal Unknown Predicates for Promise-Usefulness}
\label{sec:promise unknown list}

\cref{tab:promise_5easiest_unknown} lists the minimal predicates of arity $5$ where we have been unable to determine whether they are promise-useful or not.  \cref{tab:promise_5hardest_unknown} lists the maximal such predicates.

\begin{longtblr}[
  caption = {A list of the $9$ minimal predicates with unknown promise-usefulness status for $k=5$.},
  label = {tab:promise_5easiest_unknown},
]{
  colspec = {l@{}l|r},
  rowhead = 1,
  row{odd} = {blue9},
  row{1} = {gray9},
  columns = {colsep=4pt},
} 
     & Predicate & Dep. \\
\hline
    1. & $\{\bstr{00011}, \bstr{01101}, \bstr{01110}, \bstr{10011}, \bstr{10100}, \bstr{11000}\}$ & $13/18$ \\
    2. & $\{\bstr{00110}, \bstr{00111}, \bstr{01001}, \bstr{01101}, \bstr{01110}, \bstr{10011}, \bstr{10100}, \bstr{11000}\}$ & $2/16$ \\
    3. & $\{\bstr{00111}, \bstr{01001}, \bstr{01010}, \bstr{01101}, \bstr{01110}, \bstr{10011}, \bstr{10100}, \bstr{11000}\}$ & $3/14$ \\
    4. & $\{\bstr{00110}, \bstr{01001}, \bstr{01100}, \bstr{01101}, \bstr{01110}, \bstr{10011}, \bstr{10100}, \bstr{11000}\}$ & $1/12$ \\
    5. & $\{\bstr{00110}, \bstr{01010}, \bstr{01101}, \bstr{01111}, \bstr{10000}, \bstr{10011}, \bstr{10100}, \bstr{11000}\}$ & $1/1$ \\
    6. & $\{\bstr{00110}, \bstr{01001}, \bstr{01101}, \bstr{01110}, \bstr{10001}, \bstr{10011}, \bstr{10100}, \bstr{11000}\}$ & $2/14$ \\
    7. & $\{\bstr{00111}, \bstr{01010}, \bstr{01101}, \bstr{01110}, \bstr{10001}, \bstr{10011}, \bstr{10100}, \bstr{11000}\}$ & $6/24$ \\
    8. & \begin{tabular}{@{}l@{}l@{}l@{}l@{}l@{}l@{}l@{}l@{}l}$\{ $&$\bstr{00111},\,$&$ \bstr{01011},\,$&$ \bstr{01100},\,$&$ \bstr{01101},\,$&$ \bstr{01110},\,$&$ \bstr{10001},\,$&$ \bstr{10010},\,$&$ \bstr{10011}, $ \\
 &$ $&$  $&$  $&$  $&$  $&$  $&$ \bstr{10100},\,$&$ \bstr{11000}\}$\end{tabular} & $1/4$ \\
    9. & \begin{tabular}{@{}l@{}l@{}l@{}l@{}l@{}l@{}l@{}l@{}l}$\{ $&$\bstr{00111},\,$&$ \bstr{01000},\,$&$ \bstr{01100},\,$&$ \bstr{01110},\,$&$ \bstr{01111},\,$&$ \bstr{10000},\,$&$ \bstr{10001},\,$&$ \bstr{10011}, $ \\
 &$ $&$  $&$  $&$  $&$  $&$  $&$ \bstr{10111},\,$&$ \bstr{11000}\}$\end{tabular} & $1/1$ \\
\end{longtblr}
\begin{longtblr}[
  caption = {A list of the $7$ maximal predicates with unknown promise-usefulness status for $k=5$.},
  label = {tab:promise_5hardest_unknown},
]{
  colspec = {l@{}l|r},
  rowhead = 1,
  row{odd} = {blue9},
  row{1} = {gray9},
  columns = {colsep=4pt},
} 
     & Predicate & Dep. \\
\hline
    1. & $\{\bstr{00011}, \bstr{00110}, \bstr{00111}, \bstr{01001}, \bstr{01100}, \bstr{01101}, \bstr{01110}, \bstr{10011}, \bstr{10100}, \bstr{11000}\}$ & $6/17$ \\
    2. & $\{\bstr{00011}, \bstr{00111}, \bstr{01001}, \bstr{01010}, \bstr{01100}, \bstr{01101}, \bstr{01110}, \bstr{10011}, \bstr{10100}, \bstr{11000}\}$ & $4/14$ \\
    3. & $\{\bstr{00110}, \bstr{01010}, \bstr{01101}, \bstr{01111}, \bstr{10000}, \bstr{10011}, \bstr{10100}, \bstr{11000}\}$ & $1/1$ \\
    4. & \begin{tabular}{@{}l@{}l@{}l@{}l@{}l@{}l@{}l@{}l@{}l@{}l@{}l}$\{ $&$\bstr{00101},\,$&$ \bstr{00110},\,$&$ \bstr{00111},\,$&$ \bstr{01011},\,$&$ \bstr{01100},\,$&$ \bstr{01101},\,$&$ \bstr{01110},\,$&$ \bstr{10001},\,$&$ \bstr{10010},\,$&$ \bstr{10011}, $ \\
 &$ $&$  $&$  $&$  $&$  $&$  $&$  $&$  $&$ \bstr{10100},\,$&$ \bstr{11000}\}$\end{tabular} & $8/19$ \\
    5. & \begin{tabular}{@{}l@{}l@{}l@{}l@{}l@{}l@{}l@{}l@{}l@{}l@{}l}$\{ $&$\bstr{00110},\,$&$ \bstr{00111},\,$&$ \bstr{01001},\,$&$ \bstr{01011},\,$&$ \bstr{01100},\,$&$ \bstr{01101},\,$&$ \bstr{01110},\,$&$ \bstr{10001},\,$&$ \bstr{10010},\,$&$ \bstr{10011}, $ \\
 &$ $&$  $&$  $&$  $&$  $&$  $&$  $&$  $&$ \bstr{10100},\,$&$ \bstr{11000}\}$\end{tabular} & $16/27$ \\
    6. & $\{\bstr{00111}, \bstr{01001}, \bstr{01011}, \bstr{01100}, \bstr{01110}, \bstr{10010}, \bstr{10011}, \bstr{10100}, \bstr{10101}, \bstr{11000}\}$ & $1/4$ \\
    7. & $\{\bstr{00111}, \bstr{01000}, \bstr{01100}, \bstr{01110}, \bstr{01111}, \bstr{10000}, \bstr{10001}, \bstr{10011}, \bstr{10111}, \bstr{11000}\}$ & $1/1$ \\
\end{longtblr}
\section{Proof of \texorpdfstring{\cref{thm:fixing}}{Theorem 2.12}}
\label{fixing proof appendix}

In what follows, a \emph{choice function} $C$ is a function defined on a minion which, given an $\ell$-ary function $f$, identifies a subset $C(f) \subseteq [\ell]$ of the coordinates of $f$.  Banakh and Kozik \cite{BK24} proved the following result (this is the special case of Theorem 3.1 in their paper with two layers).

\begin{theorem}
\label{injective choice hardness}
   Suppose there exists a constant $t$ and choice function $C$ on the polymorphism minion $\Pol(\PCSP(A, B))$ with the following properties:
   \begin{enumerate}
       \item $|C(f)| \le t$ for all polymorphisms $f$ of $\PCSP(A, B)$.
       \item For every polymorphism $f$ and minor $f_{\pi}$ of $f$ such that $\pi$ is injective on $C(f)$, it holds that $\pi(C(f)) \cap C(f_\pi) \ne \emptyset$.
   \end{enumerate}
   Then $\PCSP(A, B)$ is NP-hard.
\end{theorem}

As our choices come from fixing assignments the below simple claim is useful.  The proof is standard as if two fixing assignments give values to disjoint sets of variables then one can simultaneously force the function to take the value 0 and the value 1.
\begin{claim}
\label{fixing assigns intersect}
    Let $f: \{0,1\}^{\ell} \rightarrow \{0,1\}$ be a folded Boolean function.  Then any two fixing assignments $(S_1, \alpha_1)$, $(S_2, \alpha_2)$ must satisfy $S_1 \cap S_2 \ne \emptyset$.
\end{claim}

Let us now prove \cref{thm:fixing}, restated here for convenience.

\FixingAssignmentHardness*

\begin{proof}
    We define a choice function $C$ on the polymorphisms of $\PCSP(A, B)$ such that $C(f)$ returns the coordinates of an arbitrarily chosen (say, the lexicographically smallest) fixing assignment $T$ of size $t$.  Clearly this satisfies the first condition of \cref{injective choice hardness} and we now verify the second.

    Take an arbitrary polymorphism $f: \{0,1\}^{\ell}$ and minor $f_\pi: \{0,1\}^{\ell'}$ for which $\pi: [\ell] \rightarrow [\ell']$ is injective on $C(f)$.
    Since $\pi$ is injective, every partial assignment $x_{C(f)}$ for $f$ is the preimage (under $\pi$) of a partial assignment $y_{\pi(C(f))}$ for $f_\pi$.  In particular for
    $f$'s fixing assignment on $C(f)$ there is a corresponding fixing assignment for $f_\pi$ on $\pi(C(f))$.
    In other words, both $(C(f_\pi), \alpha_1)$ and $(\pi(C(f), \alpha_2))$ are fixing assignments for $f$ for some $\alpha_1, \alpha_2$, so by \cref{fixing assigns intersect} $C(f_\pi)$ and $\pi(C(f))$ cannot be disjoint.
\end{proof}

% Per: I don't think we need this remark.
%\begin{remark}
%The injective assumption, that $\pi$ is injective on $C(f)$, is not required in the special case %of fixing sets. For fixing sets, it can be shown that $\pi(C(f)) \cap C(f_\pi) = \emptyset$ %independent of if $\pi$ is injective on $C(f)$ or not. If we had only wanted to prove that small %fixing sets imply NP-hardness, then we could have done that by a direct proof starting with label %symmetricBooleanDichotomy}.
%\end{remark}
\section{Proofs for \texorpdfstring{$\MajFam$, $\ParFam$, and $\ATFam$}{Maj, Par, and AT}}
\label{blocksymmetry proofs}

In this section we, for completeness, give the proofs of lemmas from \cref{sec:Maj Par AT conditions}, characterizing $\MajFam$, $\ParFam$, and $\ATFam$ as polymorphisms of Promise-SAT.  We do not claim that these lemmas are new and for instance the lemmas for majority and AT appear in \cite{bgs2025}.

\MajConditionLemma*

\begin{proof}
    Suppose that $\Maj_\ell$ is not a polymorphism for some $\ell$, then there exists a $k \times \ell$ obstruction matrix $M$ such that $\sum_j M^j/\ell \in [0,1/2)^k$. As $\sum_j M^j/\ell$ is a point in $K(A)$ this shows that $[0,1/2)^k \cap K(A) \neq \emptyset$. We conclude that \ref{item_K} implies \ref{item_maj}.

    On the other hand, suppose that $[0,1/2)^k \cap K(A) \neq \emptyset$ and contains the point $x$. As all vectors in $A$ have rational entries,  $x$ can be chosen to be rational and can hence be written as $\sum_i \frac {\alpha_i}{\beta} a^i$, where $\sum_i \alpha_i = \beta $ and $\alpha_1,\alpha_2,\ldots$ and $\beta$ are non-negative integers and $a^i \in A$.   As the components of $x$ are rational numbers with denominator $\beta$, each such component is bounded from above by $\frac 12-\frac 1{2\beta}$.

    We claim that for $\ell = 2d\beta +\gamma$ where $d \geq \beta$ and $\gamma < 2\beta$, $\Maj_\ell$ is not a polymorphism.  Indeed, its obstruction matrix can be constructed by taking $2d \alpha_i$ columns of each $a^i$, joint with $\gamma$ arbitrary element of $A$.  In this matrix, the sum of each row is bounded by \[ 2d\beta \left(\frac 12 -\frac 1{2\beta}\right)+\gamma \leq d\beta-d+\gamma < \ell/2.\]  As any odd number greater than $2\beta^2$ can be written on the given form we conclude that \ref{item_maj} implies \ref{item_K}.

    Finally, \ref{item_sep} immediately implies \ref{item_K} and let us establish the reverse implication.  Firstly, note that $[0,1/2)^k \cap K(A) = (-\infty,1/2)^k \cap K(A)$ as all vectors in $K(A)$ only have non-negative components. We apply \cref{thm:separate} to obtain $c_1,\ldots,c_k$ and $b$ such that $\sum_{j=1}^k x_j c_j < b$ for all $x \in (-\infty,1/2)^k$ and $\sum_{j=1}^k y_j c_j\geq b$ for all $y \in K(A)$. It is easy to see that the first condition implies that $c_1,\ldots,c_k$ are non-negative and in view of this, $b$ can be modified to take the value $\sum_{j=1}^k c_j/2$. Finally the set of vectors $c$ that satisfy the inequalities in \ref{item_sep} is a rational polytope and hence if it is nonempty it contains a rational vector.  Since it is closed under multiplication by scalars it also contains an integral vector.   This completes the proof.
\end{proof}

\ParConditionLemma*

\begin{proof}
    Let $M$ be an obstruction matrix for $\Par_\ell$, $\ell$ odd. Note that the columns of $M$ are elements of $A$, and that $\XOR_j M^j = 0^k$. This implies the existence of an odd sized subset $B$ such that $\XOR_{s \in B} s = 0^k$. This shows that \ref{item_oddset} implies \ref{item_oddpari}.

    On the other hand, suppose that there exists an odd sized subset $B$ of $A$ such that $\XOR_{s \in B} s = 0^k$. We use this subset can be to create obstructions for $\Par_\ell$ for any odd $\ell \geq |B|$. The first $|B|$ columns are given by $B$ while the remaining columns all equal some fixed element of $A$. As this column appears an even number of times, it will not change the parity of the matrix. This shows that \ref{item_oddpari} implies \ref{item_oddset}.

    To show that \ref{item_oddset} and \ref{item_alwaysodd} are equivalent, note that $\{\XOR_{s \in B} s: B \subseteq A, |B| \text{ odd}\}$ is an affine subspace in $\mathbf{F}_2^k$. Any affine subspace can always be represented as the set of solutions $x$ of a linear equation system $M x = b$ for some matrix $M \in \mathbf{F}_2^{\ell \times k}$. The statement that the affine set does not contain $0^k$ is equivalent to $b \neq 0^\ell$ and we can take any $i \in [\ell]$ where $b_i\neq0$ and let  $M_i$ define the set $\beta$ in \ref{item_alwaysodd}. We conclude that \ref{item_oddset} implies \ref{item_alwaysodd} and the reverse implication is easy to see.
 \end{proof}

\ATConditionLemma*

\begin{proof}
The proof is very similar to the proof of Lemma~\ref{lemma:test_maj} but for completeness let us give most details.  Suppose $\AT_\ell \notin \Pol(\fPCSP(A, \OR))$ for some odd $\ell$ and let $M$ be an obstruction matrix $M$.  Let $z = \frac{2}{\ell - 1}\sum_{i=1}^{\ell - 1}(-1)^{i + 1} M_i$ and note that $z$ can written as $x - y$ where both $x$ and $y$ are in $K(A)$ by letting $x$ be the sum of the odd terms, and $y$ be the sum of the even terms.  Since $M$ is an obstruction for $\AT_\ell$, $\sum_{i=1}^{\ell}(-1)^{i + 1} M_i^j < 0$ for all $j$ and as the last term of this sum is positive $\sum_{i=1}^{\ell - 1}(-1)^{i + 1} M_i^j < 0$. This shows that $z \in (-\inf,0)^k$. We conclude that \ref{item_diff} implies \ref{item_AT} and let us establish the reverse inclusion.

 Assume that $\{x - y: x,y \in K(A)\} \cap (-\infty,0)^k \neq \emptyset$ and let $x,y \in K(A)$ such that $x - y \in (-\infty,0)^k$. We can assume that $x$ and $y$ have rational coefficients and hence can be expressed as $x = \sum_i \alpha_i^x/\beta a^i$ and $y = \sum_i \alpha^y_i/\beta a^i$, where $\sum_i \alpha_i^x= \sum_i \alpha^y_i = \beta$ and $\alpha^x_1,\alpha_2,\ldots, \alpha^y_1,\alpha^y_2,\ldots$ and $\beta$ are non-negative integers and $a^i \in A$. Note also than any coordinate of $x-y$ is bounded from above by $-1/\beta$.

 We construct an obstruction matrix $M$ for $\AT_\ell$ for any $\ell > 4\beta^2$. 
 Suppose $\ell =2\beta d + \gamma$ where $\gamma < 2\beta $.  At odd indices, put $d \alpha_i^x$ columns $a^i$ and at even columns $d \alpha_i^y$ columns $a^i$ and complete this by any $\gamma$ columns from $A$.  It is easy to see that the alternating sum defining $AT$ is bounded by $-d+\gamma$ in any coordinate and hence $M$ is an obstruction matrix.

It is easy to see that \ref{item_linear} implies \ref{item_diff} as the linear form given by the $c$-vector is 0 on $K(A)-K(A)$ and negative in the negative orthant.  The reverse implication follows, as previously, by the separation theorem, Theorem~\ref{thm:separate}.  We omit the details.

\end{proof}

\end{document}